\documentclass[a4paper,11pt]{article}

\RequirePackage[
     pdftex,
     pdftitle={On Enumerating Short Projected Models},
     pdfauthor={Sibylle Möhle, Roberto Sebastiani, Armin Biere},
     pdfkeywords={model enumeration, All-SAT, propositional logic, projection},
     pdfborder={0 0 0},
     unicode,colorlinks=true,implicit=true,
     bookmarksnumbered,
     bookmarks,
     bookmarksdepth=3,
     bookmarksopenlevel=3
     ]{hyperref}

\def\orcid#1{\smash{\href{http://orcid.org/#1}{\protect\raisebox{-1.25pt}{\protect\includegraphics{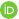}}}}}

\newcommand{\CHANGED}[1]{\textcolor{black}{#1}}
\newenvironment{changed}{\color{black}}{\normalcolor}

\usepackage{amsmath}
\usepackage{amssymb}
\usepackage{amsthm}
\usepackage{caption}
\usepackage{booktabs}
\usepackage{bussproofs}
\usepackage{colortbl}
\usepackage{enumitem}   
\usepackage{lipsum}
\usepackage{hyperref}
\usepackage{latexsym}  
\usepackage{marvosym}
\usepackage[final,tracking=true,kerning=true,spacing=true,stretch=10,shrink=10,letterspace=25]{microtype}
\usepackage{multirow}
\usepackage{tikz}
  \usetikzlibrary{arrows.meta,calc,backgrounds,positioning,shapes}
  \tikzset{>=to}
\usepackage{xspace}

\newcommand{\egt}{e.\,g.\xspace}
\newcommand{\iet}{i.\,e.\xspace}
\newcommand{\ifft}{iff\xspace}
\newcommand{\stf}{\,.\,}
\newcommand{\stt}{s.\,t.\xspace}

\newcommand{\wrtt}{w.\,r.\,t.\xspace}

\newcommand{\note}{\skipbeforenote\noindent\emph{Note:}\skipafternotetitle}

\newcommand{\skipafternotetitle}[0]{\enspace\,}
\newcommand{\skipafterinvname}[0]{\enspace}
\newcommand{\skipafterrulename}[0]{\enspace}
\newcommand{\skipafterstepnumber}[0]{\enspace}

\newcommand{\skipbeforenote}[0]{\bigskip}
\newcommand{\skipbetweeninvs}[0]{\smallskip}
\newcommand{\skipbetweenproofs}[0]{\vspace{20pt}}
\newcommand{\skipbetweenrules}[0]{\medskip}
\newcommand{\skipbetweensteps}[0]{\medskip}
\newcommand{\vspaceafterrulename}[0]{\smallskip}
\newcommand{\vspacebetweeninvs}[0]{7pt}
\newcommand{\vspacebetweenrules}[0]{7pt}
\newcommand{\vspaceinrule}[0]{2pt}
\newcommand{\vspacemidrule}[0]{0.5pt}
\newcommand{\vspaceaftermidrule}[0]{-15pt}

\newcommand{\assign}[0]{\sigma\xspace}
\newcommand{\assumed}[1]{{#1}^a}

\newcommand{\comp}[1]{\widebar{#1}}
\newcommand{\compell}[0]{\widebar{\ell}}
\newcommand{\decided}[1]{{#1}^d}
\newcommand{\decs}[0]{\functionname{decs}\xspace}
\newcommand{\decsf}[1]{\decs_{#1}\xspace}
\newcommand{\defas}{\ensuremath{\stackrel{\text{\tiny def}}{=}}\xspace}
\newcommand{\enumipname}[0]{EnumIrred}
\newcommand{\encdec}[0]{1}
\newcommand{\encprop}[0]{0}
\newcommand{\encunass}[0]{2}
\newcommand{\filter}[2]{{#1}_{#2}}
\newcommand{\functionname}[1]{\textsf{#1}}
\newcommand{\implup}[0]{\vdash_1\xspace}
\newcommand{\istar}[0]{I^{\star}}
\newcommand{\leadstoname}[1]{\leadsto_{\scalebox{0.65}{\functionname{#1}}}}  
\newcommand{\length}[1]{\lvert{#1}\rvert}
\renewcommand{\leq}[0]{\leqslant}
\newcommand{\level}[0]{\delta\xspace}
\newcommand{\levelupd}[1]{[\,#1\,]}
\newcommand{\lexlower}[0]{>_{\scalebox{0.65}{\functionname{lex}}}}
\newcommand{\maximum}[0]{\functionname{max}\xspace}

\newcommand{\mods}[0]{\functionname{models}\xspace}
\newcommand{\negated}[1]{\neg{#1}}

\newcommand{\project}[2]{\pi(#1, #2)}
\newcommand{\propagated}[2]{{#1}^{#2}}

\newcommand{\redundant}[1]{{#1}^r}
\newcommand{\relationenumip}[0]{\succ_{\scalebox{0.65}{\functionname{\enumipname}}}}
\newcommand{\residual}[2]{{#1}|{}_{#2}}

\newcommand{\state}[5]{(#1, #2, #3, #4, #5)}
\newcommand{\subvars}[2]{#1 - #2}
\newcommand{\succlex}[0]{\succ_{lex}}
\newcommand{\transenumirred}[0]{\leadsto_{\scalebox{0.65}{\functionname{\enumipname}}}}   
\newcommand{\tseitin}[0]{\functionname{Tseitin}\xspace}
\newcommand{\unassvars}[2]{#2 - #1}
\newcommand{\units}[0]{\functionname{units}\xspace}
\newcommand{\var}[0]{{V}\xspace}

\newcommand{\mainalgirredname}[0]{\functionname{EnumerateIrredundant}\xspace}
\newcommand{\mainalgredname}[0]{\functionname{EnumerateRedundant}\xspace}
\newcommand{\conflanaalgirredname}[0]{\functionname{AnalyzeConflict}\xspace}
\newcommand{\decidealgirredname}[0]{\functionname{Decide}\xspace}
\newcommand{\unitpropalgirredname}[0]{\functionname{PropagateUnits}\xspace}
\newcommand{\backtrackalgname}[0]{\functionname{Backtrack}\xspace}
\newcommand{\Bool}[0]{\mathbb{B}}

\newcommand{\Nat}[0]{\mathbb{N}}
\newcommand{\false}[0]{0}
\newcommand{\true}[0]{1}
\newcommand{\emptytrail}[0]{{\varepsilon}\xspace}
\newcommand{\emptylist}[0]{{\varepsilon}\xspace}
\newcommand{\dpllalgname}[0]{\functionname{DPLL}\xspace}
\newcommand{\cdclalgname}[0]{\functionname{CDCL}\xspace}

\newcommand{\SAT}[0]{\functionname{SAT}\xspace}
\newcommand{\SATone}[0]{the first \nolinebreak SAT solver\xspace}
\newcommand{\SATtwo}[0]{\(S\kern-0.12em AT\)\xspace}
\newcommand{\SATthree}[0]{\(S\kern-0.12em AT\)\xspace}

\newtheorem{example}{Example}

\newtheorem{proposition}{Proposition}

\newtheorem{theorem}{Theorem}

\def\L#1{\hspace{1cm}\raise .15ex\hbox{\scriptsize #1}&}
\def\LHS#1{\hbox to .80em{#1\hfill}}
\def\C#1{\mbox{\bf//}\ \small\relax#1}
\def\S#1{\mbox{\functionname{#1}}}
\def\K#1{\textbf{#1}}
\def\N{\\[.1ex]}
\def\NB{\\[.1ex]\color{blue}}

\def\NH{\\[.4ex]}

\def\NS{\\[2.5ex]}

\def\I{\hspace{1em}}
\def\B{\color{blue}}

\def\mdef{:=}

\makeatletter
\newcommand*\rel@kern[1]{\kern#1\dimexpr\macc@kerna}
\newcommand*\widebar[1]{%
  \begingroup
  \def\mathaccent##1##2{%
    \rel@kern{0.8}%
    \overline{\rel@kern{-0.8}\macc@nucleus\rel@kern{0.2}}%
    \rel@kern{-0.2}%
  }%
  \macc@depth\@ne
  \let\math@bgroup\@empty \let\math@egroup\macc@set@skewchar
  \mathsurround\z@ \frozen@everymath{\mathgroup\macc@group\relax}%
  \macc@set@skewchar\relax
  \let\mathaccentV\macc@nested@a
  \macc@nested@a\relax111{#1}%
  \endgroup
}
\makeatother

\newcommand{\EtopName}[0]{\functionname{End\(\true\)}\xspace}
\newcommand{\Etop}[0]{
  \EtopName:
  &
  \(\state{P}{N}{M}{I}{\level}\)
  \,\(\leadstoname{\EtopName}\)
  \(M \vee m\)
  ~\,if\,~
  \(\residual{P}{I} \neq \false\)
  ~and\\[\vspaceinrule]
  &
  \(\subvars{(X \cup Y \cup S)}{I} = \emptyset\)
  ~and~
  \(\var(\decs(I)) \cap X = \emptyset\)
  ~and~
  \(m \defas \project{I}{X}\)
}

\newcommand{\EbotName}[0]{\functionname{End\(\false\)}\xspace}
\newcommand{\Ebot}[0]{
  \EbotName:
  &
  \(\state{P}{N}{M}{I}{\level}\)
  \,\(\leadstoname{\EbotName}\)
  \(M\)
  ~\,if\,~
  exists \(C \in P\)
  ~with~
  \(\residual{C}{I} = \false\)
  and\\[\vspaceinrule]
  &
  \(\level(C) = 0\)
}

\newcommand{\UnitName}[0]{\functionname{Unit}\xspace}
\newcommand{\Unit}[0]{
  \UnitName:
  &
  \(\state{P}{N}{M}{I}{\level}\)
  \,\(\leadstoname{\UnitName}\)
  \(\state{P}{N}{M}{I\propagated{\ell}{C}}{\level\levelupd{\ell \mapsto a}}\)
  ~\,if\,~
  \(\residual{P}{I} \neq \false\)
  ~and\\[\vspaceinrule]
  &
  exists \(C \in P\)
  ~with~
  \(\{\ell\} = \residual{C}{I}\)
  ~and~
  \(a \defas \level(I)\)
}

\newcommand{\BtopIrredName}[0]{\functionname{Back\(\true\)}\xspace}
\newcommand{\BtopIrred}[0]{
  \BtopIrredName:
  &
  \(\state{P}{N}{M}{I}{\level}\)
  \,\(\leadstoname{\BtopIrredName}\)
  \(\state{P \wedge B}{O}{M \vee m}{J\propagated{\ell}{B}}{\level\levelupd{K \mapsto
      \infty}\levelupd{\ell \mapsto b}}\)\\[\vspaceinrule]
  &
  if\,~
  \(\subvars{(X \cup Y \cup S)}{I} = \emptyset\)
  ~and~
  exists \(\istar \leq \project{I}{X \cup Y}\)
  ~with\\[\vspaceinrule]
  &
  \(JK = I\)
  ~such that~
  \(N \wedge \istar \implup \false\)
  ~and~
  \(m \defas \project{\istar}{X}\)
  ~and\\[\vspaceinrule]
  &
  \(B \defas \negated{\decs(m)}\)
  ~and~
  \(b+1 \defas \level(B) = \level(m)\)
  ~and~
  \(\ell \in B\)
  ~and\\[\vspaceinrule]
  &
  \(\residual{\ell}{K} = \false\)
  ~and~
  \(b = \level(B \setminus \{\ell\}) = \level(J)\)
  ~and~
  \(O = \tseitin(N \vee \negated{B})\)
}

\newcommand{\BbotName}[0]{\functionname{Back\(\false\)}\xspace}
\newcommand{\Bbot}[0]{
  \BbotName:
  &
  \(\state{P}{N}{M}{I}{\level}\)
  \,\(\leadstoname{\BbotName}\)
  \(\state{P \wedge \redundant{D}}{N}{M}{J\propagated{\ell}{D}}{\level\levelupd{K \mapsto
      \infty}\levelupd{\ell \mapsto j}}\)\\[\vspaceinrule]
  &
  if\,~
  exists \(C \in P\)
  ~and~
  exists \(D\)
  ~with~
  \(JK = I\)
  ~and~
  \(\residual{C}{I} = \false\)
  ~and\\[\vspaceinrule]
  &
  \(\level(C) = \level(D) > 0\)
  ~such that~
  \(\ell \in D\)
  ~and~
  \(\negated{\ell} \in \decs(I)\)
  ~and\\[\vspaceinrule]
  &
  \(\residual{\negated{\ell}}{K} = \false\)
  ~and~
  \(P \models D\)
  ~and~
  \(j \defas \level(D \setminus \{\ell\}) = \level(J)\)
}

\newcommand{\DecXIrredName}[0]{\functionname{DecX}\xspace}
\newcommand{\DecXIrred}[0]{
  \DecXIrredName:
  &
  \(\state{P}{N}{M}{I}{\level}\)
  \,\(\leadstoname{\DecXIrredName}\)
  \(\state{P}{N}{M}{I\ell}{\level\levelupd{\ell \mapsto d}}\)
  ~\,if~
  \(\residual{P}{I} \neq \false\)
  ~and\\[\vspaceinrule]
  &
  \(\units(\residual{P}{I}) = \emptyset\)
  ~and~
  \(\level(\ell) = \infty\)
  ~and~
  \(d \defas \level(I)+1\)
  ~and~
  \(\var(\ell) \in X\)
}

\newcommand{\DecYSIrredName}[0]{\functionname{DecYS}\xspace}
\newcommand{\DecYSIrred}[0]{
  \DecYSIrredName:
  &
  \(\state{P}{N}{M}{I}{\level}\)
  \,\(\leadstoname{\DecYSIrredName}\)
  \(\state{P}{N}{M}{I\ell}{\level\levelupd{\ell \mapsto d}}\)
  ~\,if\,~
  \(\residual{P}{I} = \false\)
  ~and\\[\vspaceinrule]
  &
  \(\units(\residual{P}{I}) = \emptyset\)
  ~and~
  \(\level(\ell) = \infty\)
  ~and~
  \(d \defas \level(I)+1\)
  ~and\\[\vspaceinrule]
  &
  \(\var(\ell) \in Y \cup S\)
  ~and~
  \(\unassvars{I}{X} = \emptyset\)
}

\newcommand{\BtopRedName}[0]{\functionname{Back\(\true\)red}\xspace}
\newcommand{\BtopRed}[0]{
  \BtopRedName:
  &
  \(\state{P}{N}{M}{I}{\level}\)
  \,\(\leadstoname{\BtopRedName}\)
  \(\state{P}{N}{M \vee m}{J\propagated{\ell}{B}}{\level\levelupd{K \mapsto
      \infty}\levelupd{\ell \mapsto b}}\)\\[\vspaceinrule]
  &
  if\,~
  \(\subvars{(X \cup Y \cup S)}{I} = \emptyset\)
  ~and~
  exists \(\istar \leq \project{I}{X \cup Y}\)
  ~with\\[\vspaceinrule]
  &
  \(JK = I\)
  ~such that~
  \(N \wedge \istar \implup \false\)
  ~and~
  \(m \defas \project{\istar}{X}\)
  ~and\\[\vspaceinrule]
  &
  \(B \defas \negated{\decs(m)}\)
  ~and~
  \(b+1 \defas \level(B) = \level(m)\)
  ~and~
  \(\ell \in B\)
  ~and\\[\vspaceinrule]
  &
  \(\residual{\ell}{K} = \false\)
  ~and~
  \(b = \level(B \setminus \{\ell\}) = \level(J)\)
}

\newcommand{\InvDualPNName}[0]{\functionname{InvDualPN}\xspace}
\newcommand{\InvDualPN}[0]{
  \InvDualPNName:
  &
  \(\exists \,S\,[\,P(X,Y,S) \,]\, \equiv
  \negated{\exists \,T\,[\,N(X,Y,T)\,]\,}\)
}

\newcommand{\InvDecsName}[0]{\functionname{InvDecs}\xspace}
\newcommand{\InvDecs}[0]{
  \InvDecsName:
  &
  \(\forall b \in \{1, \ldots, \level(I)\}\, \exists ! \ell\in\decs(I)\,[\level(\ell)=b]\)
}

\newcommand{\InvImplIIrredName}[0]{\functionname{InvImplIrred}\xspace}
\newcommand{\InvImplIIrred}[0]{
  \InvImplIIrredName:
  &
  \(\forall n \in \Nat \stf P \wedge \decsf{\leq n}(I) \models \filter{I}{\leq n}\)
}

\newcommand{\InvImplIRedName}[0]{\functionname{InvImplRed}\xspace}
\newcommand{\InvImplIRed}[0]{
  \InvImplIRedName:
  &
  \(\forall n \in \Nat \stf P \wedge \negated{M} \wedge \decsf{\leq n}(I) \models \filter{I}{\leq n}\)
}

\newcommand{\InvDSOPIrredName}[0]{\functionname{InvDSOP}\xspace}
\newcommand{\InvDSOPIrred}[0]{
  \InvDSOPIrredName:
  &
  \(M\) is a DSOP
}

\begin{document}

  \title{On Enumerating Short Projected Models\thanks{The first author's contribution was mostly carried out at the
    LIT Correct and Secure Systems Lab and the Institute for Formal Models and
    Verification, Johannes Kepler University Linz, Linz, Austria.} }
  
  \author{Sibylle M{\"o}hle\orcid{0000-0001-7883-7749}\\
    Max Planck Institute for Informatics\\
    Saarland Informatics Campus\\
    Saarbr{\"u}cken, Germany\\
    \and
    Roberto Sebastiani\orcid{0000-0002-0989-6101}\\
    Department of Information Engineering\\and Computer Science (DISI)\\
    University of Trento, Italy\\
    \and
    Armin Biere\orcid{0000-0001-7170-9242}\\
    Faculty of Engineering\\
    University of Freiburg\\
    Germany
  }
  
    \date{}
  
    \maketitle

  \begin{abstract}
    Propositional model enumeration, or All-SAT, is the task to record all
    models of a propositional formula. 
    It is a key task in software and hardware verification, system engineering,
    and predicate abstraction, to mention a few.
    It also provides a means to convert a \nolinebreak CNF formula into
    \nolinebreak DNF, which is relevant in circuit design.
    While in some applications enumerating models multiple times causes no harm,
    in others avoiding repetitions is crucial.
    We therefore present two model enumeration algorithms which adopt dual
    reasoning in order to shorten the found models.
    The first method enumerates pairwise contradicting models.
    Repetitions are avoided by the use of so-called blocking clauses for which
    we provide a dual encoding.
    In our second approach we relax the uniqueness constraint.
    We present an adaptation of the standard conflict-driven clause learning
    procedure to support model enumeration without blocking clauses.
    Our procedures are expressed by means of a calculus and proofs of
    correctness are provided.
  \end{abstract}

\newcounter{Line}

\section{Introduction} 
\label{sec:intro}

The \emph{satisfiability problem of propositional logic (SAT)} consists in
determining whether for a propositional formula there exists an assignment to
its variables which evaluates the formula to true and which we call
\emph{satisfying assignment} or \emph{model}.
For proving that a formula is satisfiable, it is sufficient to provide one
single model.
However, sometimes determining satisfiability is not sufficient but all models are required.
\emph{Propositional model enumeration (All-SAT)}\footnote{For the sake of
  readability, we use the term \emph{All-SAT} also if not all models are
  required since in principle such an algorithm could always be extended to
  determine all models.} is the task of enumerating (all)
satisfying assignments of a propositional formula.
It is a key task in, \egt, bounded and unbounded model checking \nolinebreak
\cite{DBLP:conf/tacas/BiereCCZ99,DBLP:conf/cav/Shtrichman00,DBLP:conf/charme/Shtrichman01,DBLP:conf/cav/McMillan02,DBLP:conf/cav/McMillan03,DBLP:conf/tacas/JinHS05},
image computation \nolinebreak
\cite{DBLP:conf/fmcad/GuptaYAG00,DBLP:conf/date/LiHS04,DBLP:conf/date/ShengH03,DBLP:conf/fmcad/GrumbergSY04}, 
system engineering
\nolinebreak \cite{DBLP:conf/icfem/SullivanMK19}, predicate abstraction
\nolinebreak \cite{DBLP:conf/cav/LahiriNO06}, and lazy Satisfiability Modulo
Theories \nolinebreak \cite{DBLP:journals/jsat/Sebastiani07}.

Model enumeration also provides a means to convert a formula in
Conjunctive Normal Form \nolinebreak(CNF) into a logically equivalent formula in
Disjunctive Normal Form \nolinebreak (DNF) composed of the models of the
\nolinebreak CNF
formula.
This conversion is used in, \egt, circuit design \nolinebreak
\cite{DBLP:journals/tcs/MiltersenRW05} and has also been studied from a
computational complexity point of view \cite{DBLP:books/teu/Wegener87,DBLP:conf/mfcs/MiltersenRW03},
and in the worst case it is exponential in the size of the original formula
due to its exponential blowup.
If the models found are pairwise contradicting, the resulting \nolinebreak DNF
is a Disjoint Sum-of-Product \nolinebreak (DSOP) formula, which is relevant in
circuit design \cite{DBLP:conf/iccad/MinatoM98,DBLP:journals/mst/BernasconiCLP13}, and whose models
can be enumerated in polynomial time in the number of its disjuncts \nolinebreak
\cite{DBLP:conf/gcai/MohleB19}
by simply returning them, as they represent implicants of the formula.
If the models found are not pairwise contradicting, the resulting formula is
still a \nolinebreak DNF but does not support polytime model counting.
Our model enumeration algorithm basically executes a CNF-to-DNF conversion
and, from this point of view, it can be interpreted as a
knowledge compilation algorithm. 

The aim of knowledge compilation is to transform a formula into another
language\footnote{A language in this context refers to one of the various forms a
  formula can be expressed in, \egt, \nolinebreak CNF and \nolinebreak DNF denote
  the languages we are mostly interested in in this article.}
on which certain operations can be executed in polynomial time \nolinebreak
\cite{DBLP:journals/aicom/CadoliD97,DBLP:journals/jair/DarwicheM02}.
This can be done, for instance, by recording the trace of an exhaustive search
\nolinebreak
\cite{DBLP:journals/jair/HuangD07,DBLP:conf/ai/MuiseMBH12,DBLP:conf/ijcai/LagniezM17},
and the target language in these approaches is the deterministic Decomposable Normal
Form \nolinebreak (d-DNNF),\footnote{A formula is in \nolinebreak d-DNNF, if
  (1) the sets of variables of the conjuncts of each conjunction are pairwise
  disjoint, and (2) the disjuncts of each disjunction are pairwise contradicting
  \nolinebreak \cite{DBLP:journals/jair/DarwicheM02}.
  Whereas in its original definition a d-DNNF formula is defined as a directed
  acyclic graph (DAG), in this work we refer to its representation made of
  conjunctions, disjunctions, negations, variables, and truth values.} 
which was applied, for instance, in planning
\cite{DBLP:conf/aips/PalaciosBDG05}.
In contrast, in our work we record the models of the input formula,
and the resulting formula is in
\nolinebreak d-DNNF only if the detected models are pairwise contradicting.

Enumerating models requires to process the search space exhaustively and is
therefore a harder task than determining satisfiability. 
However, since state-of-the-art \nolinebreak SAT solvers are successfully
applied in industrial applications, it seems natural to use
them as a basis for model enumeration.
Modern \nolinebreak SAT solvers implement \emph{conflict-driven clause learning
\nolinebreak (CDCL)} \nolinebreak
\cite{DBLP:conf/iccad/SilvaS96,DBLP:journals/tc/Marques-SilvaS99,DBLP:conf/dac/MoskewiczMZZM01} with
\emph{non-chronological backtracking}.\footnote{Also referred to as
  \emph{backjumping} in the literature.}
If a CDCL-based \nolinebreak SAT solver is extended to support model
enumeration, adequate measures need be taken to avoid enumerating models
multiple times as
demonstrated by the following small example.

\begin{example}[Multiple model enumeration]
  \label{ex:multenumcdcl}
  Consider the propositional formula
  \begin{equation*}
    F = \underbrace{(a \vee c)}_{C_1} \wedge
    \underbrace{(a \vee \negated{c})}_{C_2} \wedge
    \underbrace{(b \vee d)}_{C_3} \wedge
    \underbrace{(b \vee \negated{d})}_{C_4}
  \end{equation*}
  which is defined over the set of variables \(V = \{a, b, c, d\}\).
  Its total models\footnote{In total models all
  variables occur.} are given by \(\mods(F) = \{a\,b\,c\,d, a\,b\,c\,\negated{d},
  a\,b\,\negated{c}\,d, a\,b\,\negated{c}\,\negated{d}\}\).
  These models may be represented by \nolinebreak \(a\,b\), \iet, they are given
  by all total extensions of \nolinebreak \(a\,b\).

  Let our model enumerator be based on \nolinebreak CDCL with non-chronological
  backtracking. 
  Assume we first decide \nolinebreak \(a\), \iet, assign \nolinebreak \(a\) the value
  true, and then \nolinebreak \(b\).
  This (partial) assignment \nolinebreak \(a\,b\) is a model of \nolinebreak \(F\).
  As in our previous work on propositional model counting \nolinebreak \cite{DBLP:conf/gcai/MohleB19},
  we flip the second decision literal, \iet, assign \nolinebreak \(b\) the value
  false, in order to explore the second branch, upon which the literal \nolinebreak \(d\) is forced to
  true in order to satisfy clause \nolinebreak \(C_3\).
  The resulting assignment \(a\,\negated{b}\,d\) now falsifies clause
  \nolinebreak \(C_4\), \iet, sets all its literal to false.
  Conflict analysis yields the unit clause \(C_5 = (b)\), which is added
  to \nolinebreak \(F\).
  The enumerator then backtracks to decision level zero, \iet, unassigns
  \nolinebreak \(d\), \nolinebreak \(b\) and \nolinebreak \(a\), and propagates
  \nolinebreak \(b\) with reason \nolinebreak \(C_5\).
  No literal is enforced by the assignment \nolinebreak \(b\), and a decision
  need be taken.
  If we choose \nolinebreak \(a\), \nolinebreak \(F\) is satisfied.
  The model found is \(b\,a\), which is the one we had found earlier.
\end{example}

Multiple model enumeration in \autoref{ex:multenumcdcl} is caused by the fact
that after conflict analysis the same satisfying assignment is repeated, albeit
in reverse order.
More generally, the same
satisfying assignment might be found again if the enumerator backtracks past a
flipped decision literal. 
Avoiding enumerating models multiple times is crucial in, \egt,
weighted model counting \nolinebreak (WMC)
\cite{DBLP:conf/aaai/SangBK05,DBLP:journals/ai/ChaviraD08,DBLP:conf/esa/FichteHWZ18,DBLP:conf/cp/DudekPV20}
and Bayesian inference \nolinebreak \cite{DBLP:journals/eccc/ECCC-TR03-003},
which require enumerating the models in order to compute their weight or
probability.
Another example is weighted model integration \nolinebreak (WMI)
\cite{DBLP:conf/ijcai/MorettinPS17,DBLP:journals/ai/MorettinPS19} which
generalizes \nolinebreak WMC for hybrid domains.
In some applications, repeating models might lead to inefficiency and
harm scalability \nolinebreak \cite{DBLP:conf/icfem/SullivanMK19}.
In the context of model counting but also relevant in model enumeration, Bayardo
and Pehoushek \nolinebreak \cite{DBLP:conf/aaai/Pehoushek00} identified the need for good learning
similarly to its learning counterpart in \nolinebreak CDCL, and
various measures have therefore been proposed to avoid the multiple enumeration of models.

One possibility is to rule out a model which was already found by adding a \emph{blocking
  clause} to the formula \nolinebreak
\cite{DBLP:conf/cav/McMillan02,DBLP:conf/ictai/MorgadoS05,DBLP:conf/ictai/MohleB18} which in essence is
the negation of the model or the decision literals in the model to be blocked
\nolinebreak \cite{DBLP:conf/ictai/MorgadoS05}.
Whenever a satisfying assignment is repeated, the clause blocking it is
falsified, and thus this model is not enumerated again.
As soon as all models are found and the relevant blocking clauses added, the
formula becomes unsatisfiable.
However, there might be an exponential number of models and adding a blocking clause for
each of them might result in a significant negative impact
on the enumerator performance.
In these cases, multiple model enumeration need be prevented by other measures.  
Toda and Soh \nolinebreak \cite{DBLP:journals/jea/TodaS16} address this issue by
adopting a variant of conflict analysis which is inspired by Gebser et al.
\nolinebreak \cite{DBLP:conf/cpaior/GebserKS09} and is exempt from blocking clauses.

The use of blocking clauses can also be avoided by adopting the
Davis-Putnam-Logemann-Loveland \nolinebreak (DPLL) algorithm \nolinebreak
\cite{DBLP:journals/cacm/DavisLL62}. 
In \nolinebreak DPLL, after a conflict or a model the last decision literal is
flipped causing the solver to 
find only pairwise contradicting models.
This idea was applied in the context of model counting by Birnbaum and Lozinskii
\nolinebreak \cite{DBLP:journals/jair/BirnbaumL99} but can readily be adapted to
support model enumeration.
Chronological backtracking in Grumberg et al. \nolinebreak \cite{DBLP:conf/fmcad/GrumbergSY04}
and in part in Gebser et al. \nolinebreak \cite{DBLP:conf/cpaior/GebserKS09} ensures that the
search space is traversed in a systematic manner similarly to \nolinebreak DPLL,
and that the use of blocking clauses can be avoided. 
An apparent drawback of \nolinebreak DPLL-based solvers, however, is that they might spend
a significant amount of time in regions of the search space having no satisfying
assignments, since---unlike \nolinebreak CDCL-based solvers---they lack the
possibility to escape those regions early.

This last issue can be addressed by the use of chronological
\nolinebreak CDCL introduced by Nadel
and Ryvchin \nolinebreak \cite{DBLP:conf/sat/NadelR18,DBLP:conf/sat/MohleB19}.
Chronological \nolinebreak CDCL combines the power of conflict-driven clause
learning with chronological backtracking.
Specifically, after finding a model, the last open (left)
decision literal is flipped in
order to process neighboring regions of the search space, while in case of a conflict
the solver is able to escape solution-less regions early.
In our earlier work \nolinebreak \cite{DBLP:conf/gcai/MohleB19},
we developed a calculus for propositional model counting based
on chronological \nolinebreak CDCL
and provided a proof of its correctness.
We took a model enumeration approach making our
method readily applicable in the context of model enumeration without
repetition.
However, while finding short models is crucial in, \egt, weighted model
integration \nolinebreak
\cite{DBLP:conf/ijcai/MorettinPS17,DBLP:journals/ai/MorettinPS19},
only total models are
detected as is usual in CDCL-based \nolinebreak SAT solvers.
The reason for this is a simple one.

To detect when a partial assignment\footnote{In a partial assignment not all
  variables occur.} is a model of the input formula, the \nolinebreak SAT solver
would have to carry out satisfiability checks before every decision, as done by
Birnbaum and Lozinskii \nolinebreak \cite{DBLP:journals/jair/BirnbaumL99}.
These satisfiability checks are expensive, and assigning the remaining variables
instead is more efficient, computationally.
If all variables are assigned and no conflict has occurred, the \nolinebreak SAT
solver knows to have found a model.
This makes sense in \nolinebreak SAT solving.
Model enumeration, however, is a harder task, and therefore more expensive methods might
pay off. 

One such method is \emph{dual reasoning} \nolinebreak
\cite{DBLP:conf/ictai/MohleB18,DBLP:conf/ysip/BiereHM17}. 
Our dual model counter Dualiza\footnote{\url{https://github.com/arminbiere/dualiza}}
takes as input the formula under consideration together with its negation.
The basic idea is to execute \nolinebreak CDCL on both formulae simultaneously
maintaining one single trail.
Whenever a conflict in the negated formula occurs, the current (partial) assignment is a
model of the formula.
Although developed for model counting, its adaptation for model
enumeration is straightforward.

Another idea enabling the detection of short models was to check whether all total extensions of the current (partial)
assignment evaluate the 
input formula to true 
before taking a decision, \iet, whether the current assignment logically entails
the input formula \nolinebreak \cite{SebastianiPartial,DBLP:conf/sat/MohleSB20}. 

Partial assignments evaluating the input formula to true 
represent sets of total models of the input formula.
However, these sets might not be disjoint
as is
demonstrated by the following example.

\begin{example}[Short redundant models]
  Let \nolinebreak \(F = (a \wedge b) \vee (a \wedge c)\) be a propositional
  formula defined over variables \(V = \{a, b, c\}\).
  Notice that \nolinebreak \(F\) is not in \nolinebreak CNF and significantly
  differs from the one in our previous example. 
  Its total models are \(\mods(F) = \{a\,b\,c, a\,b\,\negated{c},
  a\,\negated{b}\,c\}\).
  These models may also be represented by the two partial models \(a\,b\) and
  \(a\,c\).
  The former represents \(a\,b\,c\) and \(a\,b\,\negated{c}\), whereas the
  latter represents \(a\,b\,c\) and \(a\,\negated{b}\,c\). 
  Notice that \(a\,b\,c\) occurs twice.
\end{example}

Partial assignments evaluating the input formula to true result in blocking
clauses which are shorter than the ones blocking one single total model.
Adding short blocking clauses has a twofold effect.
First, a larger portion of the search space is ruled out.
Second, fewer blocking clauses need be added which mitigates their negative impact
on solver performance. 
Also, short blocking clauses generally propagate more eagerly than long
ones.
The need for shrinking or minimizing models has been pointed out by Bayardo and
Pehoushek \nolinebreak \cite{DBLP:conf/aaai/Pehoushek00} and addressed further
\nolinebreak \cite{DBLP:conf/tacas/JinHS05,DBLP:conf/tacas/RaviS04,DBLP:conf/sat/AzizCMS15}.
Notice that with blocking clauses \nolinebreak CDCL can be used as in \nolinebreak SAT solving,
while in the absence of blocking clauses it 
need be adapted.

The reason is as follows.
If a CDCL-based \nolinebreak SAT solver encounters a conflict, it analyzes it
and \emph{learns} a clause\footnote{We say that a clause is learned if it is added
to the formula.} in order to prevent the solver from repeating the same assignment which
caused the conflict. 
This clause is determined by traversing the trail in reverse assignment order and
resolving the reasons of the literals on the trail, starting with the
conflicting clause, until the resolvent contains one single literal at the
maximum decision level. 
If a model is found, the last decision literal is flipped in order to
explore another branch of the search space.
This leads to issues if this literal is encountered in later conflict
analysis and no blocking clause was added, since in this case it is neither a decision literal nor a propagated
literal.

To address this issue, Grumberg et al. \nolinebreak \cite{DBLP:conf/fmcad/GrumbergSY04} introduce
sub-levels for flipped decision literals treating them similarly to decision literals
in future conflict analysis.
Similarly to Gebser et al. \nolinebreak \cite{DBLP:conf/cpaior/GebserKS09},
Toda and Soh \nolinebreak \cite{DBLP:journals/jea/TodaS16} limit the level to
which the solver is allowed to backtrack.
These measures also ensure that enumerating overlapping partial models
is avoided.
However, in applications where repetitions cause no harm,
the power of finding even shorter models representing
larger, albeit not disjoint, sets of models, can be exploited.
Shorter models are also obtained in the case of model enumeration under \emph{projection}.

If not all variables are relevant in an application, we project the models of
the input formula onto the relevant variables, or, otherwise stated, we
existentially quantify the irrelevant variables.
Projection occurs in, \egt, model checking
\cite{DBLP:conf/cav/Shtrichman00,DBLP:conf/charme/Shtrichman01}, image
computation \nolinebreak
\cite{DBLP:conf/fmcad/GuptaYAG00,DBLP:conf/fmcad/GrumbergSY04}, quantifier
elimination \nolinebreak
\cite{DBLP:conf/cav/BrauerKK11,DBLP:conf/fmics/ZenglerK13}, and predicate
abstraction \nolinebreak \cite{DBLP:conf/cav/LahiriNO06}.
The breadth of these applications highlights the relevance of projection in
practice. 

\paragraph{Our contributions}
In this article we address the task of enumerating short projected models with
and without repetition.
We start by presenting a CDCL-based algorithm for the case where only pairwise contradicting,
\iet, \emph{irredundant},
models are sought.
Multiple model enumeration is prevented by the addition of
blocking clauses to the input formula, and 
dual reasoning is adopted for shrinking total models.
To ensure correctness of the latter, we introduce the concept of \emph{dual
  blocking clauses} which provides
a solution to an issue identified in our earlier
work \nolinebreak \cite{DBLP:conf/ictai/MohleB18}.
Dual reasoning in model shrinking enables us to obtain short models, and
\nolinebreak CDCL lets us exploit the strengths of state-of-the-art \nolinebreak
SAT solvers.
Short models result in short blocking clauses with the potential to
reduce their number and to rule out a larger portion of the search space.
We express our algorithm by means of a formal calculus and provide a correctness
proof.
A generalization of our algorithm to the case where partial satisfying
assignments are found and shrunken is presented. 
This generalization makes sense as we do not guarantee that our model
shrinking method gives us the minimal model.
We discuss the appropriate changes to our algorithm, calculus, proof, and its
generalization.

We then introduce a relaxed version of our algorithm for enumerating 
non-contradicting, \iet, \emph{redundant}, models.
This method is exempt of blocking clauses, and
consequently decision literals which were flipped after a model lack a
reason.
To fix this issue, we introduce an adaptation of \nolinebreak CDCL for
\nolinebreak SAT to All-SAT.
We discuss the changes to our previous algorithm needed in order to support redundant model
enumeration.

This article builds on our work presented at the 23rd International Conference on
Theory and Applications of Satisfiability Testing (SAT) 2020 \nolinebreak
\cite{DBLP:conf/sat/MohleSB20}.
It also uses concepts introduced by Sebastiani \nolinebreak 
\cite{SebastianiPartial} as well as presented at the Second Young Scientist's International
Workshop on Trends in Information Processing (YSIP2) 2017 \nolinebreak
\cite{DBLP:conf/ysip/BiereHM17} and the 30th
International Conference on Tools with Artificial Intelligence 
(ICTAI) 2018 \nolinebreak \cite{DBLP:conf/ictai/MohleB18}, the 22nd
International Conference on Theory and Applications of Satisfiability Testing
(SAT) 2019 \nolinebreak \cite{DBLP:conf/sat/MohleB19}, and the 5th
Global Conference on Artificial Intelligence (GCAI) 2019
\nolinebreak \cite{DBLP:conf/gcai/MohleB19}.

\paragraph{Structure of the paper}
In \autoref{sec:prelim}, 
we introduce our notation and basic concepts.
Dual reasoning is applied for shrinking models in \autoref{sec:shrinking}, and
an according dual encoding of blocking clauses is introduced in
\autoref{sec:dualblock}.
After presenting our algorithm for projected model enumeration without
repetition in \autoref{sec:pmeirred} and providing a formalization and correctness
proof and a generalization to the detection of partial models in
\autoref{sec:formalenumirred}, we turn our attention to projected model
enumeration with repetition. 
We adapt \nolinebreak CDCL for \nolinebreak SAT to support
conflict analysis in the context of model enumeration without the use of blocking
clauses in \autoref{sec:learn} and discuss the changes to our method needed to
support multiple model enumeration in \autoref{sec:pmered}, before we conclude
in \autoref{sec:conclusion}.

\section{Preliminaries}
\label{sec:prelim}

In this section we provide the concepts and notation on which our presentation
relies: propositional satisfiability \nolinebreak (SAT) and incremental
\nolinebreak SAT solving, projection, and the dual representation of a formula,
which constitutes the basis for dual reasoning.

\subsection{Propositional Satisfiability (SAT)}
\label{sec:sat}

The set containing the
Boolean constants \nolinebreak \(\false\) (false) and \nolinebreak \(\true\)
(true) is denoted with \(\Bool = \{\false, \true\}\).
Let \nolinebreak \(V\) be a set of propositional (or Boolean) variables.
A \emph{literal} is either a variable \nolinebreak \(v \in V\) or its negation
\nolinebreak \(\negated{v}\).
We write \nolinebreak \(\compell\) to denote the \emph{complement} of
\nolinebreak \(\ell\) assuming \(\compell = \negated{\ell}\) and
\(\overline{\negated{\ell}} = \ell\). 
The variable of a literal \nolinebreak \(\ell\) is obtained by \nolinebreak
\(\var(\ell)\).
This notation is extended to formulae, clauses, cubes, and sets of literals.

Most \nolinebreak SAT solvers work on formulae in \emph{Conjunctive Normal Form 
  \nolinebreak (CNF)} which are conjunctions of \emph{clauses}, which are disjunctions of
literals.
These \nolinebreak SAT solvers implement efficient algorithms tailored for
\nolinebreak CNFs, such as unit propagation, which will be presented below.
In contrast, a formula in \emph{Disjunctive Normal Form \nolinebreak (DNF)} is a disjunction
of \emph{cubes} which are conjunctions of literals.
We interpret formulae as sets of clauses and write \(C \in F\) to refer to a
clause \nolinebreak \(C\) occurring in the formula \nolinebreak \(F\).
Accordingly, we interpret clauses and cubes as sets of literals.
The empty \nolinebreak CNF formula and the empty cube are denoted by
\nolinebreak \(\true\), while the empty \nolinebreak DNF formula and the empty
clause are represented by \nolinebreak \(\false\).

A \emph{total assignment} \(\assign \colon V \mapsto \Bool\) maps \nolinebreak \(V\) to
the truth values \nolinebreak \(\false\) and \nolinebreak \(\true\).
It can be applied to a formula \nolinebreak \(F\) over a set of variables
\nolinebreak \(V\) to obtain the \emph{value of \nolinebreak \(F\) under
  \nolinebreak \(\assign\)}, denoted by \nolinebreak
\(\assign(F) \in \Bool\), also written \nolinebreak
\(\residual{F}{\assign}\).
A sequence \(I = \ell_1, \ldots, \ell_n\) with mutually exclusive variables
(\(\var(\ell_i) \neq \var(\ell_j)\) for \(i \neq j\)) is called a \emph{trail}.
If their variable sets are disjoint, trails and literals may be concatenated,
denoted \(I = I'\,I''\) and \(I = I'\,\ell\,I''\).
Consider as an example a set of variables \(V = \{a,b,c,d,e,f,g,h\}\) and let
\(I' = a\,\negated{b}\,\negated{c}\), \(I'' = e\,f\,\negated{g}\,h\), and
\(\ell = \negated{d}\) be two trails and a literal, respectively.
Now \(\var(I') = \{a,b,c\}\), \(\var(I'') = \{e,f,g,h\}\), and \(\var(\ell)=d\).
Since \(\var(I') \cap \var(I'') = \emptyset\) and
\(\var(\ell) \not\in \var(I') \cup \var(I'')\), they can be concatenated
obtaining, \egt, 
\(I = I'\,I'' = a\,\negated{b}\,\negated{c}\,e\,f\,\negated{g}\,h\) and
\(I = I'\,\ell\,I'' =
a\,\negated{b}\,\negated{c}\,\negated{d}\,e\,f\,\negated{g}\,h\).
We treat trails as conjunctions or sets of literals and write \(\ell \in I\) if \nolinebreak
\(\ell\) is contained in \nolinebreak \(I\).
Trails can also be interpreted as partial assignments with \(I(\ell) = \true\)
\ifft \(\ell \in I\).
Similarly, \(I(\ell) = \false\) \ifft \(\negated{\ell} \in I\), and \nolinebreak 
\(I(\ell)\) is undefined \ifft \(\var(\ell) \not\in \var(I)\).
The unassigned variables in \nolinebreak \(V\) are denoted by \(\unassvars{I}{V}\) and 
the empty trail by \nolinebreak \(\emptytrail\).

We call \emph{residual} of \nolinebreak \(F\) under \nolinebreak \(I\), denoted
\nolinebreak \(\residual{F}{I}\), the formula \nolinebreak \(I(F)\) obtained by
assigning the variables in \nolinebreak \(F\) their truth value.
If \nolinebreak \(F\) is in \nolinebreak CNF, this amounts to removing from
\nolinebreak \(F\) all clauses containing a literal \(\ell \in I\) and removing
from the remaining clauses all occurrences of \nolinebreak \(\negated{\ell}\).
For instance, given a formula \(F=(a \vee b)\wedge(\negated{a}\vee b\vee c)\) with
\(\var(F)=\{a,b,c\}\) and \(I=a\), \(\residual{F}{I}=(b\vee c)\).
If \nolinebreak \(\residual{F}{I} = \true\), we say that \nolinebreak \(I\)
\emph{satisfies} \nolinebreak \(F\) or that \nolinebreak \(I\) is a
\emph{model} of \nolinebreak \(F\).
If all variables are assigned, we call \nolinebreak \(I\) a \emph{total model}
of \nolinebreak \(F\).   
Following the distinction highlighted by Sebastiani \nolinebreak
\cite{SebastianiPartial}, if \nolinebreak \(I\) is a partial assignment, we say
that \emph{\nolinebreak \(I\) evaluates \nolinebreak \(F\) to \nolinebreak
  \true}, written \(I \vdash F\), if \(\residual{F}{I} = \true\), and that
\emph{\nolinebreak \(I\) logically entails \nolinebreak \(F\)}, written \(I
\models F\), or that \nolinebreak \(I\) is a \emph{partial model} of
\nolinebreak \(F\), if all total assignments extending \nolinebreak \(I\)
satisfy \nolinebreak \(F\).
Notice that \(I \vdash F\) implies that \(I \models F\) but not vice versa: \egt,
if \(F \defas (a \wedge b) \vee (a \wedge \negated{b})\) and \(I \defas a\),
then \(I \models F\) but \(I \not\vdash F\), because
\(\residual{F}{I} = (b \vee \negated{b}) \neq \true\).
If \nolinebreak \(F\) is in \nolinebreak CNF without valid clauses, \iet,
without clauses containing contradicting literals, then \(I \vdash F\) \ifft
\(\residual{F}{I} = \true\).
We say that \emph{\nolinebreak \(I\) evaluates \nolinebreak \(F\) to
  \nolinebreak \(\false\)} or that \nolinebreak \(I\) is a \emph{counter-model}
of \nolinebreak \(F\), \ifft \(\residual{F}{I} = \false\). 
If \nolinebreak \(F\) is in \nolinebreak CNF, its residual under \nolinebreak
\(I\) contains the empty clause, \(\false \in \residual{F}{I}\).

\subsection{The Davis-Putnam-Logemann-Loveland  (DPLL) Algorithm}
\label{sec:davis-putn-logem}

\begin{figure}
  \setcounter{Line}{1}

  \begin{tabular}{@{}r@{\hspace{.85em}}l@{}}
    \L{} \hspace{-3cm}\begin{minipage}[l]{2.9cm}
      {\hfill\textbf{Input:}\hspace{.75em}~}\\[-1ex] 
    \end{minipage}
    \begin{minipage}[l]{0.9\textwidth}
      formula \(F\)\\[-1ex]
    \end{minipage}
    \N
    \L{} \hspace{-3cm}\begin{minipage}[l]{2.9cm}
      {\hfill\textbf{Output:}\hspace{.75em}~}\\
    \end{minipage}
    \begin{minipage}[l]{0.7\textwidth}
      SAT if \(F\) is satisfiable, UNSAT if \(F\) is unsatisfiable\\
    \end{minipage} \N
    \L{}
    \kern-1.5em\scalebox{1.05}{\dpllalgname}\hspace{2.0pt}(\,\(F\)\,) \NH
    \L{\theLine} \(I \mdef \emptytrail\) \N
    \stepcounter{Line}
    \L{\theLine} \K{forever do} \N
    \stepcounter{Line}
    \L{\theLine} \I \(C \mdef\) \S{\unitpropalgirredname}\hspace{2.0pt}(\,\(F\),~\(I\)\,)   \N
    \stepcounter{Line}
    \L{\theLine} \I \K{if} \(C \neq \false\) \K{then} \N
    \stepcounter{Line}
    \L{\theLine} \I\I \K{if} \(\decs(I) = \emptyset\) \K{then} \N
    \stepcounter{Line}
    \L{\theLine} \I\I\I \K{return} UNSAT \N
    \stepcounter{Line}
    \L{\theLine} \I\I \K{else} \N
    \stepcounter{Line}
    \L{\theLine} \I\I\I flip most recent decision \(\decided{\ell}\in I\) \N
    \stepcounter{Line}
    \L{\theLine} \I \K{else}  \N
    \stepcounter{Line}
    \L{\theLine} \I\I \K{if} there are unassigned variables in \(\var(F)\) 
    \K{then} \N
    \stepcounter{Line}
    \L{\theLine} \I\I\I \S{\decidealgirredname}\hspace{2.0pt}(\,\(F\),~\(I\)\,) \N
    \stepcounter{Line}
    \L{\theLine} \I\I \K{else} \N
    \stepcounter{Line}
    \L{\theLine} \I\I\I \K{return} SAT  \NS
    \setcounter{Line}{1}
    \L{} \kern-1.5em\scalebox{1.05}{\unitpropalgirredname}\hspace{2.0pt}(\,\(F\),~\(I\)\,) \NH  
    \L{\theLine} \K{while} \(\residual{C}{I}=(\ell)\) for some \(C\in F\) \K{do} \N
    \stepcounter{Line}
    \L{\theLine} \I \(I \mdef I\,\ell\) \N
    \stepcounter{Line}
    \L{\theLine} \I \K{for all} clauses \(D \in F\) such that \(\negated\ell\in D\) \K{do} \N
    \stepcounter{Line}
    \L{\theLine} \I\I \K{if} \(\residual{D}{I} = \false\) \K{then} \K{return} \(D\) \N
    \stepcounter{Line}
    \L{\theLine} \K{return} \(\false\) \NS
    \setcounter{Line}{1}
    \L{} \kern-1.5em\scalebox{1.05}{\decidealgirredname}\hspace{2.0pt}(\,\(F\),~\(I\)\,) \NH  
    \L{\theLine} \(I\mdef I\,\decided{\ell}\) where \(\var(\ell)\in\var(F)
    \setminus \var(I)\)
  \end{tabular}
  \caption{
    \label{fig:dpllalg}
    Davis-Putnam-Logemann-Loveland  (DPLL) Algorithm.
    The main loop starts with exhaustive unit propagation.
    If a conflict occurs, the most recent decision is flipped.
    If there are no decisions left on the trail, the execution terminates
    returning UNSAT.
    If no conflict occurs and there are still unassigned variables, a decision
    is taken.
    Otherwise, the execution terminates returning SAT.
  }
\end{figure}

The satisfiability of a propositional formula \nolinebreak
\(F\) over a set of variables \nolinebreak \(V\) can be determined by the
Davis-Putnam-Logemann-Loveland  (DPLL) algorithm \nolinebreak
\cite{DBLP:journals/cacm/DavisLL62,DBLP:journals/jacm/DavisP60} depicted in
\nolinebreak \autoref{fig:dpllalg}.\footnote{As it is common practice in the SAT community, we do not consider the pure-literal rule from the original DPLL procedure because it is considered ineffective.}
Its main ingredient is the trail \nolinebreak \(I\) which is iteratively
extended by a literal \nolinebreak \(\ell\) which is either \emph{propagated} or 
\emph{decided}.
In the former case, there exists a clause \(C \in F\) containing \nolinebreak
\(\ell\) in which all literals except \nolinebreak \(\ell\) evaluate to false 
under the current (partial) assignment \nolinebreak \(I\).
The literal \nolinebreak \(\ell\) is called \emph{unit literal} or \emph{unit} and
\nolinebreak \(C\) a \emph{unit clause}.
In order to satisfy \nolinebreak \(C\), and thus \nolinebreak \(F\), the literal
\nolinebreak \(\ell\) need be assigned the value true.  
After being propagated,
the literal \nolinebreak \(\ell\) becomes a \emph{propagation literal}, and \nolinebreak
\(C\) is called its \emph{reason}.
If after the propagation of \nolinebreak \(\ell\) a clause \nolinebreak \(D\in
F\) becomes false, this clause is returned to indicate that a \emph{conflict}
occurred. 
The corresponding rule is the \emph{unit propagation} rule (function
\unitpropalgirredname).
Notice that \nolinebreak \(\residual{F}{I}\) may contain multiple multiple
reasons for a unit literal, and by speaking of ``its'' reason we refer to the
one chosen in the current execution. 
If a literal is decided, its value is chosen according to some heuristic by a
\emph{decision}, and it is called \emph{decision literal}. 
We annotate decision literals on the trail by a superscript, \egt, \nolinebreak
\(\decided{\ell}\), denoting open ``left'' branches (\decidealgirredname).
If a decision literal \nolinebreak \(\decided{\ell}\) is flipped, its complement
\nolinebreak \(\compell\) opens a ``right'' branch.
The set consisting of all decision literals on the trail \nolinebreak \(I\) is
obtained by \nolinebreak \(\decs(I) = \{\ell \mid \decided{\ell}\in I\}\).

In the main loop of the function \dpllalgname, first exhaustive unit propagation
is carried out (line \nolinebreak 3).
If it does not return the empty clause, a conflict has occurred.
If there is no decision literal left on the trail, both the left and right
branch of all decisions have been explored.
The trail \nolinebreak \(I\) can not be extended to a satisfying assignment of \nolinebreak
\(F\) and the execution stops returning UNSAT (lines \nolinebreak 4--6).
Otherwise, the literals assigned after the most recent decision \nolinebreak
\(\decided{\ell}\) are removed from \nolinebreak \(I\) and \(\decided{\ell}\) is
flipped (line \nolinebreak 8).
If exhaustive unit propagation returns the empty clause, no conflict has
occurred and there is no unit literal in \nolinebreak \(\residual{F}{I}\).
If not all variables are assigned a value, a decision need be taken (lines
\nolinebreak 10--11).
Otherwise, the execution terminates returning SAT (line \nolinebreak 13).

\subsection{Conflict-Driven Clause Learning (CDCL)}
\label{sec:cdclsat}

The DPLL algorithm lacks the possibility to escape early from solution-less
regions of the search space.
Conflict-driven clause learning \nolinebreak (CDCL)
enables the solver to learn a clause representing the
reason for the current conflict and to accordingly undo multiple decisions in
one step.
The trail is now partitioned into blocks, called \emph{decision levels}, which
extend from a decision literal to the last literal preceding the next decision.

The \emph{decision level function} \(\level \colon V \mapsto \Nat\) 
assigns and alters decision levels.
The decision level of a variable \nolinebreak \(v \in V\) is obtained by
\nolinebreak \(\level(v)\).
If \nolinebreak \(v\) is unassigned, we have \(\level(v) = \infty\).
We extend \nolinebreak \(\level\) accordingly to determine the decision level of
literals \nolinebreak \(\ell\),
non-empty clauses \nolinebreak \(C\), and non-empty trails
\nolinebreak \(I\), by defining \(\level(\ell) = \level(\var(\ell))\),
\(\level(C) = \maximum\{\level(\ell) \mid \ell \in C\}\), and \(\level(I) =
\maximum\{\level(\ell) \mid \ell \in I\} = \#\{\ell \mid \decided{\ell} \in I\}\).
The decision level of the trail \nolinebreak \(I\) therefore corresponds to the
number of decision literals on \nolinebreak \(I\), and if \nolinebreak \(I\)
contains only propagated literals, then \nolinebreak \(\level(I)=0\).
The subsequence of \nolinebreak \(I\) consisting of all literals with
decision level smaller or equal to \nolinebreak \(n\), is denoted by
\nolinebreak \(\filter{I}{\leq n}\).
Accordingly, we define \(\level(L) = \maximum\{\level(\ell) \mid \ell \in L\}\)
for a non-empty set of literals \nolinebreak \(L\).
If \nolinebreak \(v\) is unassigned, we have \(\level(v) = \infty\), and
\(\level(\false) = \level(\emptytrail) = \level(\emptyset)\) for the empty
clause, the empty sequence and the empty set of literals.
Whenever a variable is assigned or unassigned, the decision level function
\nolinebreak \(\level\) is updated.
If \nolinebreak \(\var(\ell)\) is assigned at decision level \nolinebreak \(d\),
we write \nolinebreak \(\level\levelupd{\ell \mapsto d}\).
If all variables in the set of variables \nolinebreak \(V\) are assigned
decision level \(\infty\), we
write \nolinebreak \(\level\levelupd{V \mapsto \infty}\) or \(\level \equiv \infty\) as a
shortcut.
Similarly, if all literals occurring on the trail \nolinebreak
\(I\) are unassigned, \iet, removed from \(I\), their decision level is
assigned \nolinebreak \(\infty\), and we write
\(\level\levelupd{I \mapsto \infty} = \level\levelupd{V(I) \mapsto
\infty}\).
The function \(\level\) is left-associative, \iet, \(\level\levelupd{I \mapsto
\infty}\levelupd{\ell \mapsto d}\) first unassigns all variables on \nolinebreak \(I\)
and then assigns literal \nolinebreak \(\ell\) at decision level \nolinebreak
\(d\).

Literals occurring before the first decision are assigned exclusively by unit
propagation at decision level zero. 
A propagated literal is annotated with its reason, as in
\(\propagated{\ell}{C}\), and assigned at the current 
decision level \(\level(I)\).
The trail can be represented graphically by the \emph{implication graph} which is
defined as follows.
Decision literals are represented as nodes on the left and annotated with their
decision level.
Propagated literals are internal nodes with one incoming arc originating from each
node representing a literal in their reason.
A conflict is represented by the special node \nolinebreak \(\kappa\) whose
incoming arcs are annotated with the conflicting clause.

\begin{example}[Trail and implication graph]
  \label{ex:impgraphtrail}
  Consider the formula
  \begin{equation*}
    F = \underbrace{(\negated{a} \vee b)}_{C_1} \wedge
    \underbrace{(\negated{c} \vee d)}_{C_2} \wedge
    \underbrace{(\negated{b} \vee \negated{c} \vee \negated{d})}_{C_3}
  \end{equation*}
  over the set of variables \(V = \{a, b, c, d\}\).
  Assume we first decide \nolinebreak \(a\), then propagate \nolinebreak \(b\)
  with reason \nolinebreak \(C_1\) followed by deciding \nolinebreak \(c\) and
  propagating \nolinebreak \(d\) with reason \nolinebreak \(C_2\).
  Under this assignment, the clause \nolinebreak \(C_3\) is falsified.
  The current trail is given by
  \begin{equation*}
    I = \decided{a}\,\propagated{b}{C_1}\,\decided{c}\,\propagated{d}{C_2}
  \end{equation*}
  where \(\level(a) = \level(b) = 1\), \(\level(c) = \level(d) = 2\),
    and \(\decs(I) = \{a, b\}\).
  The corresponding implication graph is
  \begin{center}
    \begin{tikzpicture}
      \node (a) {\(a@1\)};
      \node[right=50pt of a] (b) {\(b\)};
      \node[below=16pt of a] (c) {\(c@2\)};
      \node[below=16pt of b] (d) {\(d\)};
      \node[right=50pt of d] (kappa) {\(\kappa\)};
      \draw[-{to[length=5.5mm]}] (a) to node [above, midway] (TextNode) {\(C_1\)  } (b);
      \draw[-{to[length=5.5mm]}] (c) to node [above, midway] (TextNode) {\(C_2\)  } (d);
      \draw[-{to[length=5.5mm]}] (b) to [bend left=15] node [above, midway] (TextNode) {\(C_3\)  } (kappa);
      \draw[-{to[length=5.5mm]}] (d) to node [above, midway] (TextNode) {\(C_3\)  } (kappa);
      \draw[-{to[length=5.5mm]}] (c) to [bend right=30] node [below, midway] (TextNode) {\(C_3\)} (kappa);
    \end{tikzpicture}
  \end{center}
\end{example}

\begin{figure}
  \setcounter{Line}{1}

  \begin{tabular}{@{}r@{\hspace{.85em}}l@{}}
    \L{} \hspace{-3cm}\begin{minipage}[l]{2.9cm}
      {\hfill\textbf{Input:}\hspace{.75em}~}\\[-1ex] 
    \end{minipage}
    \begin{minipage}[l]{0.9\textwidth}
      formula \(F\)\\[-1ex]
    \end{minipage}
    \N
    \L{} \hspace{-3cm}\begin{minipage}[l]{2.9cm}
      {\hfill\textbf{Output:}\hspace{.75em}~}\\
    \end{minipage}
    \begin{minipage}[l]{0.7\textwidth}
      SAT if \(F\) is satisfiable, UNSAT if \(F\) is unsatisfiable\\
    \end{minipage} \N
    \L{}
    \kern-1.5em\scalebox{1.05}{\cdclalgname}\hspace{2.0pt}(\,\(F\)\,) \NH
    \L{\theLine} \(I \mdef \emptytrail\) \N
    \stepcounter{Line}
    \L{\theLine}  \(\level\levelupd{V \mapsto \infty}\) \N
    \stepcounter{Line}
    \L{\theLine} \K{forever do} \N
    \stepcounter{Line}
    \L{\theLine} \I \(C \mdef\) \S{\unitpropalgirredname}\hspace{2.0pt}(\,\(F\),~\(I\),~\(\level\)\,)   \N
    \stepcounter{Line}
    \L{\theLine} \I \K{if} \(C \neq \false\) \K{then} \N
    \stepcounter{Line}
    \L{\theLine} \I\I \K{if} \(\level(I) = 0\) \K{then} \N
    \stepcounter{Line}
    \L{\theLine} \I\I\I \K{return} UNSAT \N
    \stepcounter{Line}
    \L{\theLine} \I\I \K{else} \N
    \stepcounter{Line}
    \L{\theLine} \I\I\I \S{\conflanaalgirredname}\hspace{2.0pt}(\,\(F\),~\(I\),~\(C\),~\(\level\)\,) \N
    \stepcounter{Line}
    \L{\theLine} \I \K{else}  \N
    \stepcounter{Line}
    \L{\theLine} \I\I \K{if} there are unassigned variables in \(\var(F)\) 
    \K{then} \N
    \stepcounter{Line}
    \L{\theLine} \I\I\I \S{\decidealgirredname}\hspace{2.0pt}(\,\(F\),~\(I\),~\(\level\)\,) \N
    \stepcounter{Line}
    \L{\theLine} \I\I \K{else} \N
    \stepcounter{Line}
    \L{\theLine} \I\I\I \K{return} SAT \NS
    \setcounter{Line}{1}
    \L{} \kern-1.5em\scalebox{1.05}{\unitpropalgirredname}\hspace{2.0pt}(\,\(F\),~\(I\),~\(\level\)\,) \NH  
    \L{\theLine} \K{while} some \(C \in F\) is unit \((\ell)\) under \(I\) \K{do} \N
    \stepcounter{Line}
    \L{\theLine} \I \(I \mdef I\,\propagated{\ell}{C}\) \N
    \stepcounter{Line}
    \L{\theLine} \I \(\level(\ell) \mdef \level(I)\) \N
    \stepcounter{Line}
    \L{\theLine} \I \K{for all} clauses \(D \in F\) containing \(\negated\ell\) \K{do} \N
    \stepcounter{Line}
    \L{\theLine} \I\I \K{if} \(\residual{D}{I} = \false\) \K{then} \K{return} \(D\) \N
    \stepcounter{Line}
    \L{\theLine} \K{return} \(\false\) \NS
    \setcounter{Line}{1}
    \L{} \kern-1.5em\scalebox{1.05}{\conflanaalgirredname}\hspace{2.0pt}(\,\(F\),~\(I\),~\(C\),~\(\level\)\,) \NH
    \L{\theLine} \(D \mdef\) \S{Learn}\hspace{2.0pt}(\,\(I\),~\(C\)\,) \N
    \stepcounter{Line}
    \L{\theLine} \(F \mdef F \wedge D\) \N
    \stepcounter{Line}
    \L{\theLine}  \(\ell \mdef\) literal in \(D\) at decision level \(\level(I)\) \N
    \stepcounter{Line}
    \L{\theLine} \(j \mdef \level(D \setminus \{\ell\})\) \N
    \stepcounter{Line}
    \L{\theLine} \K{for all} literals \(k \in I\) with decision level \(> j\) \K{do} \N
    \stepcounter{Line}
    \L{\theLine} \I assign \(k\) decision level \(\infty\) \N
    \stepcounter{Line}
    \L{\theLine} \I remove \(k\) from \(I\) \N
    \stepcounter{Line}
    \L{\theLine} \(I \mdef I\,\propagated{\ell}{D}\) \N
    \stepcounter{Line}
    \L{\theLine} \(\level(\ell) \mdef j\) \NS
    \setcounter{Line}{1}
    \L{} \kern-1.5em\scalebox{1.05}{\decidealgirredname}\hspace{2.0pt}(\,\(F\),~\(I\),~\(\level\)\,) \NH  
    \L{\theLine} \(I\mdef I\,\decided{\ell}\) where
    \(\var(\ell)\in\var(F)\setminus\var(I)\) \N
    \stepcounter{Line}
    \L{\theLine} \(\level(\ell)\mdef\level(I)\)
  \end{tabular}
  \caption{
    \label{fig:cdclalg}
    CDCL-based satisfiability algorithm. In the main loop, first exhaustive unit
    propagation is executed.
    If a conflict at decision level zero occurs, the execution terminates
    returning UNSAT.
    Otherwise, the conflict is analyzed, a clause blocking the assignment
    responsible for the conflict is learned and backjumping occurs.
    If no conflict occurred and all variables are assigned, the execution
    terminates returning SAT, otherwise a decision is taken.
  }
\end{figure}

As in DPLL, the CDCL algorithm, see \autoref{fig:cdclalg}, executes a main loop
until either a 
satisfying assignment has been found or all possible assignments have been
checked without finding one.
It starts with exhaustive unit propagation (line \nolinebreak 4 and function
\unitpropalgirredname).
If a conflict occurs, we call the clause whose literals are set
to false under \nolinebreak \(I\) \emph{conflicting clause}.
If the trail \(I\) contains no decision literal, \iet, its decision level is
zero, the execution terminates returning UNSAT (lines \nolinebreak 6--7).
Otherwise, \emph{conflict analysis} is executed and backtracking 
occurs (line \nolinebreak 9).
If no conflict occurs and there are unassigned variables, a decision is taken
(line \nolinebreak 12), where the new decision literal is assigned to the decision
level of \nolinebreak \(I\).
Otherwise, the execution terminates returning SAT (line \nolinebreak 14).

Conflict analysis is described by procedure \conflanaalgirredname .
Suppose the current trail \nolinebreak \(I\) falsifies the clause \nolinebreak
\(C\in F\).
The basic idea is to compute a clause, let's say \nolinebreak \(D\),
containing the negated assignments responsible for the conflict.
By adding \nolinebreak \(D\) to \nolinebreak \(F\), this assignment is blocked.
Moreover, backtracking to the second highest decision level in \nolinebreak
\(D\) results in \nolinebreak \(D\) becoming unit, and its literal with highest
decision level is propagated.
A main ingredient of the clause learning algorithm is \emph{resolution} \nolinebreak
\cite{DBLP:journals/jacm/Robinson65}.
Given two clauses \((A \vee \ell)\) and \((B \vee \negated{\ell})\),
where \nolinebreak \(A\) and \nolinebreak \(B\) are disjunctions of literals and
\nolinebreak \(\ell\) is a literal, their \emph{resolvent} \((A \vee \ell)
\otimes_{\ell} (B \vee \negated{\ell}) = (A \vee B)\) is obtained by resolving
them on \nolinebreak \(\ell\).
The clause \nolinebreak \(D\) is determined by a sequence of resolution steps
which can be read off either the implication graph or the trail.
First, the conflicting clause is resolved with the reason of one
of its literals.
This procedure is repeated with the reason of one literal in the resolvent and
continued until the resolvent contains one single literal at \emph{conflict
  level}, which is the decision level of the conflicting clause.
In the implication graph, we start with the conflict node \nolinebreak
\(\kappa\) and follow the edges in reverse direction to determine the next
literal to resolve on.
Considering the trail, we start with the conflicting clause and choose the most
recent propagated literal for resolution.

\begin{example}[Trail-based conflict-driven clause learning]
  \label{ex:cdclimpgraphtrail}
  Consider the situation in \autoref{ex:impgraphtrail}.
  The conflicting clause is \nolinebreak \(C_3 = (\negated{b} \vee \negated{c}
  \vee \negated{d})\) with \(\level(C_3)=2\) and
  \(I = \decided{a}\,\propagated{b}{C_1}\,\decided{c}\,\propagated{d}{C_2}\).
  The most recent unit literal is \nolinebreak \(d\) with reason \nolinebreak
  \(C_2\).
  We resolve \(C_3\) with \(C_2\) on \(d\) and obtain 
  \(C_3 \otimes_{d} C_2 = (\negated{b} \vee \negated{c} \vee \negated{d})
  \otimes_{d} (\negated{c} \vee d) = (\negated{b} \vee
  \negated{c})\).\footnote{After applying a factorization step, \iet, removing
    one \(\negated{c}\) from the resolvent \((\negated{b} \vee \negated{c} \vee
    \negated{c})\).}
  Now \(\level(\negated{b})=1\neq 2=\level(d)=\level(C_3)\), and the clause
  \((\negated{b} \vee \negated{c})\) is learned.
\end{example}

Following the trail in reverse assignment order gives us a deterministic
sequence of resolution steps.
In contrast, when determining the clause to be learned based on the
implication graph, this choice need not be deterministic.

\begin{example}[Implication-graph-based conflict-driven clause learning]  
  Consider the implication graph given in \autoref{ex:impgraphtrail} for clause
  learning.
  The choice of the clauses to resolve need not be deterministic. 
  For instance, we can resolve \nolinebreak \(C_3\) either on \nolinebreak \(d\)
  with \nolinebreak \(C_2\) or on \nolinebreak \(b\) with \nolinebreak \(C_1\).  
  If we choose \nolinebreak \(C_1\), the resolvent \((\negated{a} \vee
  \negated{c} \vee \negated{d})\) contains two literals at conflict level,
  namely \nolinebreak \(\negated{c}\) and \nolinebreak \(\negated{d}\), and a
  second resolution step is needed, whereas by resolving \nolinebreak \(C_3\) with
  \nolinebreak \(C_2\) on \nolinebreak \(d\) first, see
  \autoref{ex:cdclimpgraphtrail}, one resolution step is saved.  
\end{example}

\subsection{Incremental SAT Solving} 
\label{sec:incsat}

The basic idea of incremental \nolinebreak SAT solving is to exploit the
progress made during the search process if similar formulae need be solved.
So, instead of the learned clauses to be discarded, they are retained between the
single \nolinebreak SAT calls.

Hooker \nolinebreak \cite{DBLP:journals/jlp/Hooker93} presented the idea of
incremental \nolinebreak SAT solving in the context of knowledge-based reasoning.
E{\'e}n and S{\"o}rensson \nolinebreak \cite{DBLP:journals/entcs/EenS03} introduced
the concept of \emph{assumptions} for incremental \nolinebreak SAT
solving which fits our needs best.
Assumptions can be viewed as unit clauses added to the formula.
They basically represent a (partial) assignment whose literals remain set to
true during the solving process.
In particular, backtracking does not occur past any assumed literal.

\subsection{Projection}
\label{sec:project}

We are interested in enumerating the models of a propositional
formula \emph{projected} onto a subset of its variables.
To this end we partition the set of variables \(V = X \cup Y\) into the set of
\emph{relevant variables} \nolinebreak \(X\) and the set of \emph{irrelevant
  variables} \nolinebreak \(Y\) and write \nolinebreak \(F(X \cup Y)\) to
express that \nolinebreak \(F\) depends on the variables in \(X \cup Y\).
Accordingly, we decompose the assignment \nolinebreak \(\assign = \assign_X \cup
\assign_Y\) into its relevant part \(\assign_X \colon X \mapsto \Bool\) and its
irrelevant part \(\assign_Y \colon Y \mapsto \Bool\) following the convention
introduced in our earlier work on dual projected model counting \nolinebreak
\cite{DBLP:conf/ictai/MohleB18}.
The main idea of projection onto the relevant variables is to existentially
quantify the irrelevant variables.
The models of \nolinebreak \(F(X \cup Y)\) projected onto \nolinebreak \(X\)
are therefore
\begin{align*}
  \mods(\exists \,Y.\,F(X,Y))
  =\;
  & \{
  \tau \colon X \to \Bool \mid
  \mbox{exists~}\sigma: X\cup Y \rightarrow \Bool\mbox{~with~}\\
  & \qquad\qquad\quad\enspace\,\sigma\bigl(F(X,Y)\bigr) = 1
  \mbox{~~and~~}
  \tau=\sigma_X
  \},
\end{align*}
\noindent
and enumerating all models of \nolinebreak \(F\) without projection is therefore
the special case where 
\(Y = \emptyset\).
The projection of the trail \nolinebreak \(I\) onto the set of
variables \nolinebreak \(X\) is denoted by \nolinebreak \(\project{I}{X}\) and
\(\project{F(X,Y)}{X} \equiv \exists \,Y \, [\,F(X,Y)\,]\).

\begin{example}[Projected models]
  \label{ex:projenum}
  Consider again the formula \nolinebreak \(F\) in \autoref{ex:multenumcdcl}.
  Its unprojected models are \(\mods(F) = \{a\,b\,c\,d,
  a\,b\,c\,\negated{d}, a\,b\,\negated{c}\,d,
  a\,b\,\negated{c}\,\negated{d}\}\). 
  Its models projected onto \(X = \{a, c\}\) are \nolinebreak \(a\,c\) and
  \nolinebreak \(a\,\negated{c}\).
\end{example}

In order to benefit from the efficient methods \nolinebreak SAT solvers execute
on \nolinebreak CNF formulae, we transform an arbitrary formula
\nolinebreak \(F(X,Y)\) into \nolinebreak CNF by, \egt, the Tseitin transformation
\nolinebreak \cite{Tseitin1968complexity}.\footnote{It turns out that in the
  context of dual projected model enumeration also the Plaisted-Greenbaum
  transformation \nolinebreak \cite{DBLP:journals/jsc/PlaistedG86} might be used
although in general it does not preserve the model count.
}
By this transformation, auxiliary variables, also referred to as \emph{Tseitin}
or \emph{internal} variables, are introduced. 
The Tseitin transformation is \emph{satisfiability-preserving}, \iet, a
satisfiable formula is not turned into an unsatisfiable one and, similarly, an
unsatisfiable formula is not turned into a satisfiable one.
The Tseitin variables, which we denote by \nolinebreak \(S\), are defined in terms of
the variables in \nolinebreak \(X \cup Y\), which we call \emph{input variables}.
As a consequence, for each total assignment to the variables in \(X \cup Y\)
there exists one single assignment to the variables in \nolinebreak \(S\) such that the resulting
assignment is a model of \nolinebreak \(F\), and therefore the model count is preserved.  
Due to the introduction of the Tseitin variables, the resulting formula
\(P(X,Y,S) = \tseitin(F(X,Y))\) is not logically equivalent to \nolinebreak \(F(X,Y)\), \iet,
\(\mods(F) \neq \mods(P)\), and the Tseitin transformation is not
\emph{equivalence-preserving}.  
However, the models of \nolinebreak \(P(X,Y,S)\) projected onto the input
variables are exactly the models of \nolinebreak \(F(X,Y)\), and
\begin{equation}
  \label{eq:pf}
  \exists \,S\, [\,P(X,Y,S)\,] \equiv F(X,Y).
\end{equation}
The models of \nolinebreak \(F\) projected onto \nolinebreak \(X\) are accordingly 
given by
\begin{equation}
  \label{eq:modelspf}
  \mods(\exists \,Y,S \,[\,P(X,Y,S)\,]) = \mods(\exists \,Y\, [\,F(X,Y)\,]).
\end{equation}

\subsection{Dual Representation of a Formula}
\label{sec:dual}

We make use of the dual representation of a formula introduced in our earlier
work \nolinebreak \cite{DBLP:conf/ictai/MohleB18}. 
Let \nolinebreak \(F(X,Y)\) and \nolinebreak \(P(X,Y,S)\) be defined as in
\autoref{sec:project}, and
let \(N(X,Y,T) = \tseitin(\negated{F(X,Y)})\) be a \nolinebreak CNF
representation of \nolinebreak \(\negated{F}\), where \nolinebreak \(T\) denotes
the set of Tseitin variables introduced by the transformation, \iet,
\begin{equation}
  \label{eq:nnotf}
  \exists \,T\, [\,N(X,Y,T)\,] \equiv \negated{F(X,Y)}.
\end{equation}
The formulae \nolinebreak \(P(X,Y,S)\) and \nolinebreak \(N(X,Y,T)\) are
a \emph{dual
  representation}\footnote{Referred to as \emph{combined formula pair} of
  \nolinebreak \(F(X,Y)\) in our previous work \nolinebreak
  \cite{DBLP:conf/ictai/MohleB18}.} of the input formula \(F(X,Y)\). 

\begin{example}[Dual formula representation]
  Let \(F(X,Y)=(a\wedge\negated{b})\vee(\negated{a}\wedge b)\) be defined over
  \(X=\{a\}\) and \(Y=\{b\}\) and suppose
  \(\negated{F(X,Y)}=(a\wedge b)\vee(\negated{a}\wedge\negated{b})\).
  A dual representation of \nolinebreak \(F(X,Y)\) consists of the Tseitin
  transformations of \nolinebreak \(F(X,Y)\) and its negation,
  \(P(X,Y,S)=
  (\negated{s_1}\vee a)\wedge
  (\negated{s_1}\vee\negated{b})\wedge
  (s_1\vee\negated{a}\vee b)\wedge
  (\negated{s_2}\vee\negated{a})\wedge
  (\negated{s_2}\vee b)
  (s_2\vee a\vee\negated{b})\wedge
  (s_1\vee s_2)\) and
  \(N(X,Y,T)=
  (\negated{t_1}\vee a)\wedge
  (\negated{t_1}\vee b)\wedge
  (t_1\vee\negated{a}\vee\negated{b})\wedge
  (\negated{t_2}\vee\negated{a})\wedge
  (\negated{t_2}\vee\negated{b})\wedge
  (t_2\vee a\vee b)\wedge
  (t_1\vee t_2)\), respectively, where \(S=\{s_1,s_2\}\) and \(T=\{t_1,t_2\}\).   
\end{example}

For the sake of readability, we also may write \nolinebreak \(F\), \nolinebreak
\(P\), and \nolinebreak \(N\).
Notice that this representation is not unique in general.
Besides that,
\nolinebreak \(P(X,Y,S)\) and \nolinebreak \(N(X,Y,T)\) share
the set of input variables \nolinebreak \(X \cup Y\), and \(S \cap T =
\emptyset\), and 
\begin{equation}
    \exists \,S\,[\,P(X,Y,S) \,]\, \equiv
    \negated{\exists T\,[\,N(X,Y,T)\,]\,},
  \label{eq:pnotn}
\end{equation}

In an earlier work \nolinebreak \cite{DBLP:conf/ictai/MohleB18} we showed that
during the enumeration process a generalization of the following always holds
assuming we first decide variables in \nolinebreak \(X\) and then variables 
in \(Y \cup S\) but never variables in \nolinebreak \(T\):
\begin{equation}
  \left(\,\neg\exists\, T\,[\,\residual{N(X,Y,T)}{I} \,]\,\right)
  \models \left(\,\exists\, S\,[\,P(X,Y,S)|_I\,]\,\right)
  \label{eq:ntop}
\end{equation}
where \nolinebreak \(I\) is a trail over variables in \nolinebreak \(X \cup Y
\cup S \cup T\). 
Obviously, also
\begin{equation}
  \left(\,\exists\, S\,[\,\residual{P(X,Y,S)}{I}\,]\,\right)
  \models \left(\,\negated{\exists}\, T\,[\,\residual{N(X,Y,T)}{I}\,]\,\right),
  \label{eq:pton}
\end{equation}
saying that whenever \nolinebreak \(I\) can be extended to a model of
\nolinebreak \(P\), all extensions of it falsify \nolinebreak \(N\).
This property
is a basic ingredient of our dual model shrinking method.

\section{Dual Reasoning for Model Shrinking} 
\label{sec:shrinking}

In a previous work, we adopted dual reasoning for obtaining partial models
\nolinebreak \cite{DBLP:conf/ictai/MohleB18}.
Basically, we executed \nolinebreak CDCL on the formula under consideration and
its negation simultaneously exploiting the fact that \nolinebreak CDCL is biased
towards detecting conflicts.
Our experiments showed that dual reasoning detects short models.
However, processing two formulae simultaneously turned out to be computationally
expensive.

In another work \nolinebreak \cite{DBLP:conf/sat/MohleSB20} we propose, before
taking a decision, 
to check whether the
current (partial) assignment logically entails the formula under consideration.
We present four flavors of the entailment check, some of which use a
\nolinebreak SAT oracle and rely on dual reasoning.

The method introduced in this work, instead, 
exploits the effectiveness of dual reasoning in
detecting short partial models while avoiding both processing two formulae
simultaneously and oracle calls, which might be computationally expensive.
In essence, we let the enumerator find total models and shrink them by means of
dual reasoning.

Assume our task is to determine the models of a formula 
\nolinebreak \(F(X,Y)\) over the set of
relevant variables \nolinebreak \(X\) and irrelevant variables \nolinebreak
\(Y\) projected
onto \nolinebreak \(X\), and
let \nolinebreak \(P(X,Y,S)\) and \nolinebreak \(N(X,Y,T)\) be \nolinebreak CNF
representations of \nolinebreak \(F\) and \nolinebreak \(\negated{F}\),
respectively, as introduced in \autoref{sec:prelim}.
Obviously, \autoref{eq:modelspf}--\autoref{eq:pton} hold.
Suppose standard \nolinebreak CDCL is executed on \nolinebreak \(P\).
We denote with \nolinebreak \(I\) the trail which ranges over variables in \(X
\cup Y \cup S \cup T\), where \nolinebreak \(S\) and \nolinebreak \(T\) are the
Tseitin variables occurring in \nolinebreak \(P\) and \(N\), respectively.

Now assume a total model \nolinebreak \(I\) of \(P\) is found.
A second \nolinebreak SAT solver is incrementally invoked on \nolinebreak
\(\project{I}{X \cup Y} \wedge N\). 
Since \(\project{I}{X \cup Y} \models F\) and all variables in \nolinebreak \(X
\cup Y\) are assigned,
due to \autoref{eq:pton}, a conflict in
\nolinebreak \(N\) occurs by propagating variables in \nolinebreak \(T\) only. 
If conflict analysis is carried out as described in \autoref{sec:cdclsat}, the learned
clause \nolinebreak \(\negated{\istar}\) contains only negated assumption
literals.\footnote{See also the work by Niemetz et al. \nolinebreak
  \cite{DBLP:conf/fmcad/NiemetzPB14}.} On the one hand, \nolinebreak
\(\negated{\istar}\) represents a 
cause for the conflict in \nolinebreak \(N\). On the other hand, due to
\autoref{eq:ntop}, its negation \nolinebreak \(\istar\) represents a (partial)
model of \nolinebreak \(F\).
More precisely, \nolinebreak \(\istar\) represents all total models of
\nolinebreak \(F\) projected onto \nolinebreak \(X \cup Y\) in which the
variables in \nolinebreak \(X \cup Y\) not occurring in \nolinebreak \(\istar\)
may assume any truth value. 

\begin{example}[Model shrinking by dual reasoning]
  \label{ex:dualshrink}
  Let \(F = (a \vee b) \wedge (c \vee d)\) be a propositional formula over the
  set of variables \(V = \{a, b, c, d\}\).
  Without loss of generalization, suppose we want to enumerate the models of
  \nolinebreak \(F\) projected onto \nolinebreak \(V\).
  Assume a total model \(I = \decided{a}\,\decided{b}\,\decided{c}\,\decided{d}\) has been found. 
  We call a second SAT solver on \nolinebreak \(N \wedge I\), where
  \begin{align*}
    N = \,
    &
      \underbrace{(\negated{t_1} \vee \negated{a})}_{C_1} \wedge
      \underbrace{(\negated{t_1} \vee \negated{b})}_{C_2} \wedge
      \underbrace{(a \vee b \vee t_1)}_{C_3} \wedge\\
    &
      \underbrace{(\negated{t_2} \vee \negated{c})}_{C_4} \wedge
      \underbrace{(\negated{t_2} \vee \negated{d})}_{C_6} \wedge
      \underbrace{(c \vee d \vee t_2)}_{C_6} \wedge\\
    &
      \underbrace{(t_1 \vee t_2)}_{C_7}
  \end{align*}
  is the Tseitin encoding of \(\negated{F} = (\negated{a} \wedge \negated{b})
  \vee (\negated{c} \wedge \negated{d})\) with Tseitin variables
  \(T = \{t_1, t_2\}\).
  The clauses \nolinebreak \(C_1\) to \nolinebreak \(C_3\) encode \nolinebreak
  \((t_1 \leftrightarrow (\negated{a} \wedge \negated{b}))\), the clauses \nolinebreak \(C_4\) to
  \nolinebreak \(C_6\) encode \nolinebreak \((t_2 \leftrightarrow (\negated{c} \wedge \negated{d}))\),
  and \nolinebreak \(C_7\) encodes \nolinebreak \((t_1 \vee t_2)\).

  The literals on \nolinebreak \(I\) are considered assumed variables, annotated
  by, \egt, \nolinebreak \(\assumed{a}\), and
  \(I = \assumed{a}\,\assumed{b}\,\assumed{c}\,\assumed{d}\).
  After propagating \nolinebreak \(\negated{t_1}\) with reason \nolinebreak
  \(C_1\) and \nolinebreak \(\negated{t_2}\) with reason \nolinebreak \(C_4\),
  the clause \nolinebreak \(C_7\) becomes empty.
  The trail is \(I' =
  \assumed{a}\,\assumed{b}\,\assumed{c}\,\assumed{d}\,\propagated{\negated{t_1}}{C_1}\,\propagated{\negated{t_2}}{C_4}\).
  We resolve \nolinebreak \(C_7\) with \nolinebreak \(C_4\) to obtain the clause
  \((t_1 \vee \negated{c})\) which we then resolve with \nolinebreak \(C_1\).
  The resolvent is \((\negated{c} \vee \negated{a})\) containing only
  assumed literals which have no reason in \nolinebreak \(I'\), and thus can
  not be resolved further. 
  Below on the left hand side, the implication graph is visualized, and the
  corresponding resolution steps are depicted on the right hand side.

  \begin{center}
    \begin{tikzpicture}
      \node(a) {\(a\)};
      \node[below = 8pt of a] (b) {\(b\)};
      \node[below = 8pt of b] (c) {\(c\)};
      \node[below = 8pt of c] (d) {\(d\)};
      \node[right = 40pt of a] (t1) {\(\negated{t_1}\)};
      \node[right = 40pt of c] (t2) {\(\negated{t_2}\)};
      \node[right = 100pt of b] (kappa) {\(\kappa\)};
      \draw[-{to[length=5.5mm]}] (a) to node[above, midway] (TextNode) {\(C_1\)} (t1);
      \draw[-{to[length=5.5mm]}] (c) to node[above, midway] (TextNode) {\(C_4\)} (t2);
      \draw[-{to[length=5.5mm]}] (t1) to [bend left=15] node[above, midway] (TextNode) {\(C_7\)} (kappa);
      \draw[-{to[length=5.5mm]}] (t2) to [bend right=15] node[below, midway] (TextNode) {\(C_7\)} (kappa);
      \node (11) at ($(kappa.east)+(60pt,8pt)$) {\((t_1 \vee t_2)\)};
      \node[right = 0.5cm of 11] (12) {\((\negated{t_2} \vee \negated{c})\)};
      \draw (11.south west) -- (12.south east);
      \node[below = 0.3cm of $(11)!0.5!(12)$] (21) {\((t_1 \vee \negated{c})\)};
      \node[right = 0.5cm of 21] (22) {\((\negated{t_1} \vee \negated{a})\)};
      \draw (21.south west) -- (22.south east);
      \node[below = 0.3cm of $(21)!0.5!(22)$] {\((\negated{c} \vee \negated{a})\)};
    \end{tikzpicture}
  \end{center}

  The negation of the clause \nolinebreak \((\negated{c} \vee \negated{a})\),
  \nolinebreak \(c\,a\), is a counter-model of \nolinebreak \(\negated{F}\) and
  hence a
  model of \nolinebreak \(F\).
  In this case, it is also minimal \wrtt the number of literals.
\end{example}

\note
The gain obtained by model shrinking is twofold.
On the one hand, it enables the (implicit) exploration of multiple models in one
pass: \egt, in \autoref{ex:dualshrink}, the model \nolinebreak \(c\,a\) represents
four total models, namely, \(a\,b\,c\,d\), \(a\,b\,c\,\negated{d}\),
\(a\,\negated{b}\,c\,d\), and \(a\,\negated{b}\,c\,\negated{d}\). 
On the other hand, short models result in short blocking clauses
ruling out a larger part of the search space, as mentioned earlier.

\section{Dual Encoding of Blocking Clauses}
\label{sec:dualblock}

Recall that \CHANGED{in our dual model shrinking approach we rely on} \autoref{eq:pnotn}.
If a blocking clause is added to \nolinebreak \(P\) and \nolinebreak \(N\) is
not updated accordingly, then \nolinebreak \(P\) and \nolinebreak \(N\) do not
represent the negation of each other anymore, and \autoref{eq:pnotn} ceases to
hold.
This might lead to multiple model enumerations in the further search\CHANGED{,
  if dual model shrinking is applied}.
This issue can be remediated by adding the shrunken models disjunctively to
\nolinebreak \(N\).
The basic idea is the following.
If a trail \nolinebreak \(I\) satisfies a formula \nolinebreak \(F\),
then it
falsifies its negation \nolinebreak \(\negated{F}\), \iet,
\(F \wedge \negated{I} \equiv \false\) and
\(\negated{(\negated{F} \wedge \negated{I})} = \negated{F} \vee I \equiv
\true\).
To retain \nolinebreak \(\negated{F}\) in \nolinebreak CNF and ensure
\autoref{eq:pnotn}, we propose the following \emph{dual encoding} of the
blocking clauses. 

We denote with \(\tseitin()\) the function which takes as argument an arbitrary
propositional formula and returns its Tseitin transformation. 
For the sake of readability, we write \nolinebreak \(F\), \nolinebreak \(P\),
and \nolinebreak \(N\) as well as their indexed variants instead of \nolinebreak
\(F(X \cup Y)\), \nolinebreak \(P(X \cup Y \cup S)\) and \nolinebreak \(N(X \cup
Y \cup T)\). 
We define
\begin{align}
  P_0 & = \tseitin(F) \label{eq:pzero}\\
  N_0 & = t_0 \wedge \tseitin(t_0\,\leftrightarrow\, \negated{F}). \label{eq:nzero}
\end{align}
Let \nolinebreak \(I_1\) be a trail such that \nolinebreak \(I_1\) evaluates \nolinebreak \(F\) to
\nolinebreak true, \iet, \(I_1 \vdash F\).
A second \nolinebreak SAT solver \nolinebreak \SATtwo is called on \nolinebreak \(\project{I_1}{X \cup Y} \wedge N_0\),
and a conflict is obtained as argued above.
Assume \SATtwo returns the assignment \nolinebreak \(\istar_1 \leq I_1\) such that 
\(\text{\SATtwo}(\project{\istar_1}{X \cup Y},N_0) = \text{UNSAT}\).
Then \nolinebreak \(\negated{\project{\istar_1}{X}}\) is added to \nolinebreak \(P_0\)
obtaining \(P_1 = P_0 \wedge \negated{\project{\istar_1}{X}}\).
To ensure \autoref{eq:pnotn}, we define \(N_1 = \left(t_0 \vee
  t_1\right) \wedge \tseitin\left(t_0\,\leftrightarrow\,\negated{F}\right)
\wedge \tseitin{(t_1\,\leftrightarrow\,\project{\istar_1}{X})}\), and we apply
this encoding inductively as follows.
At the \(n^{\textnormal{th}}\) \nolinebreak step, we have
\begin{align}
  P_n & = P_0 \; \underline{\wedge
  \bigwedge\limits_{i=1}^{n}\negated{\project{\istar_i}{X}}} \label{eq:pn}\\
  N_n & = (t_0 \; \underline{\vee \bigvee\limits_{i=1}^{n}t_i}) \wedge
  \tseitin(t_0\,\leftrightarrow\,\negated{F}) \; \underline{\wedge
  \bigwedge\limits_{i=1}^{n}\tseitin{(t_i\,\leftrightarrow\,\project{\istar_i}{X})}} \label{eq:nn}
\end{align}
where the underlined parts denote the additions to \nolinebreak \(P_0\) and
\nolinebreak \(N_0\).

Let \nolinebreak \(I_{n+1}\) be a trail evaluating \(P_n\) 
to true, \iet, \(I_{n+1} \vdash
P_n\). 
We invoke \nolinebreak \SATtwo on \(\project{I_{n+1}}{X \cup Y} \wedge N_n\)
leading to a conflict as described above.
Assume \nolinebreak \SATtwo returns \nolinebreak \(\istar_{n+1} \leq I_{n+1}\), such that
\(\text{\SATtwo}(\project{\istar_{n+1}}{X \cup Y},N_n) = \text{UNSAT}\).
We add \nolinebreak \(\negated{\project{\istar_{n+1}}{X}}\) to \nolinebreak \(P_n\)
and update \nolinebreak \(N_n\) accordingly.
Now we have\footnote{Notice that here, with a little abuse of notation, by
  ``\(N_n\setminus\{C\}\)'' we refer to the formula resulting from dropping clause \nolinebreak
  \(C\) from formula \nolinebreak \(N_n\).} 
\begin{align}
  P_{n+1} & = P_n \wedge \negated{\project{\istar_{n+1}}{X}} \label{eq:p2} \\
  N_{n+1} & = N_n \setminus \{(t_0 \vee \bigvee\limits_{i=1}^{n}t_i)\} \wedge
  (t_0 \vee \bigvee\limits_{i=1}^{n+1}t_i) \wedge
  \tseitin{(t_{n+1}\,\leftrightarrow\,\project{\istar_{n+1}}{X})}\label{eq:imp2}\\
  \text{where} & \qquad I_{i+1} \vdash P_i ~\text{for}~ 0 \leqslant i \leqslant n
  \nonumber\\
  \text{and} & \qquad \istar_{i+1} \leq I_{i+1} \enspace\text{is \stt}\enspace
  \text{\SATtwo}(\project{\istar_{i+1}}{X \cup Y}, N_i) = \text{UNSAT}.
  \nonumber
\end{align}

\begin{proposition}
  \label{prop:dualblockirredundant}
  Let \nolinebreak \(F(X,Y)\) be an arbitrary propositional formula over the
  relevant variables \nolinebreak \(X\) and the irrelevant variables
  \nolinebreak \(Y\).
  Let \nolinebreak \(F\) and \nolinebreak \(\negated{F}\) be encoded into
  \nolinebreak CNFs \nolinebreak \(P_0\) and \nolinebreak \(N_0\), respectively,
  according to \autoref{eq:pzero} and \autoref{eq:nzero}.
  If for all models found blocking clauses are added to \nolinebreak \(P_0\) and
  \nolinebreak \(N_0\) according to \autoref{eq:pn} and
  \autoref{eq:nn}, then only
  pairwise contradicting models are found, \iet,
  \(\project{\istar_i}{X}\) and \(\project{\istar_j}{X}\) are pairwise
  contradicting for every \(i \neq j\). 
\end{proposition}
\begin{proof} 
  By construction, \(N_{n} \equiv \negated{F} \vee
  \bigvee_{i=1}^{n} \project{\istar_{i}}{X \cup Y}\) and, given a shrunken model
  \nolinebreak \(\istar_{n+1}\) of \(P_n\),
  \(\project{\istar_{n+1}}{X \cup Y} \wedge N_n \equiv \false\). 
  Furthermore, \(\project{\istar_{n+1}}{X \cup Y} \wedge \negated{F} \equiv \false\).
  We have
  \begin{align*}  
    \false \equiv\ 
    &
      \project{\istar_{n+1}}{X \cup Y} \wedge (\negated{F}  \vee
      \bigvee\limits_{i=1}^n \project{\istar_{i}}{X \cup Y})\\
    =\
    &
      (\project{\istar_{n+1}}{X \cup Y} \wedge \negated{F}) \vee
      (\project{\istar_{n+1}}{X \cup Y} \wedge
      \bigvee\limits_{i=1}^{n} \project{\istar_i}{X \cup Y})\\
    \equiv\
    &
      (\project{\istar_{n+1}}{X \cup Y} \wedge
      \bigvee\limits_{i=1}^{n} \project{\istar_i}{X \cup Y})\\
    \equiv\
    &
      (\project{\istar_{n+1}}{X} \wedge
      \bigvee\limits_{i=1}^{n} \project{\istar_i}{X}),
  \end{align*}
  since \nolinebreak \(\istar_i\) contains only relevant variables.
  Hence,
  \(\project{\istar_{n+1}}{X} \wedge \project{\istar_i}{X} \equiv \false\) for
  \(i = 1, \ldots, n\).
\end{proof}

\note
\autoref{eq:pnotn} always holds:
\begin{equation*}
  \exists S\,[\,P_i(X,Y,S)\,]\, \equiv \negated{\exists T\,[\,N_i(X,Y,T)\,]\,}
  \enspace \text{for all} \enspace 0 \leq i \leq n+1.
  \label{eq:pnnotnn}
\end{equation*}
Consequently, also \autoref{eq:ntop} and \autoref{eq:pton} hold:
\begin{equation*}
  \left(\,\neg\exists\,
  T\,[\,\residual{N_i(X,Y,T)}{I}\,]\,\right) \models
  \left(\,\exists\,S\,[\,\residual{P_i(X,Y,S)}{I}\,]\,\right)
  \enspace
  \text{for all} \enspace 0 \leq i \leq n+1
  \label{eq:nnpn}
\end{equation*}
\begin{equation*}
  \left(\,\exists\, S\,[\,\residual{P_i(X,Y,S)}{I}\,]\,\right)
  \models \left(\,\negated{\exists}\, T\,[\,\residual{N_i(X,Y,T)}{I}\,]\,\right)
  \enspace \text{for all} \enspace 0 \leq i \leq n+1
  \label{eq:pnnn}
\end{equation*}
However, for our usage we may use the implication in the forward direction only
and write \(t_i\,\rightarrow\,\project{\istar_i}{X}\) and
\(t_{n+1}\,\rightarrow\,\project{\istar_{n+1}}{X}\) in \autoref{eq:nn} and
\autoref{eq:imp2} without compromising correctness for the following reason:
the formula \nolinebreak \(N_i\) is called always under \nolinebreak \(I_{i+1}\)
which falsifies all \nolinebreak \(\istar_k\) for \(0 \leq k \leq i\).
Hence, \(\project{\istar_i}{X} \rightarrow t_i\) is always true.
  
\begin{example}[Dual blocking clauses]
  We clarify the proposed encoding by a small example and show that it prevents
  multiple model counts.
  Let our example be \(F = x_1 \vee (x_2 \wedge x_3)\) and assume we have found the model
  \(\textcolor{blue}{\istar_1 = x_1}\).
  Then
  \begin{align*}
    P_1 =\,
    &
    (\negated{s_1} \vee x_2) \wedge
    (\negated{s_1} \vee x_3) \wedge
    (s_1 \vee \negated{x_2} \vee \negated{x_3}) \wedge
    (x_1 \vee s_1) \wedge
    \textcolor{blue}{(\negated{x_1})} ~~~~\\[-3ex]
  \end{align*}
  and
  \begin{align*}
    N_1 =\,
    &
      \underbrace{(\negated{t_1} \vee \negated{x_1})}_{C_1} \wedge
      \underbrace{(\negated{t_1} \vee \negated{x_2})}_{C_2} \wedge
      \underbrace{(t_0 \vee x_1 \vee x_2)}_{C_3} \wedge\\
    &
      \underbrace{(\negated{t_2} \vee \negated{x_1})}_{C_4} \wedge
      \underbrace{(\negated{t_2} \vee \negated{x_3})}_{C_5} \wedge
      \underbrace{(t_2 \vee x_2 \vee x_3)}_{C_6} \wedge\\
    &
      \underbrace{(t_1 \vee t_2 \vee \textcolor{blue}{t_3})}_{C_7} \wedge 
      \,\textcolor{blue}{\underbrace{(\negated{t_3} \vee x_1)}_{C_8}} \wedge
      \textcolor{blue}{\underbrace{(t_3 \vee \negated{x_1})}_{C_9}}
  \end{align*}
  where the unit clause $(\neg{x_1})$, the literal $t_3$ in $C_7$ as well as the
  clauses $C_8$ and $C_9$, all
  emphasized in blue, denote the corresponding additions to \nolinebreak 
  \(P_0\) and \nolinebreak \(N_0\).
  If now we find a total model \(I_2 = \negated{x_1}\,x_2\,x_3\), we obtain a conflict
  in \nolinebreak \(N_1\) by unit propagating variables \nolinebreak \(t_1\),
  \(t_3\), and \(t_3\) 
  only.
  The conflicting clause is 
  \((t_1 \vee t_2 \vee t_3)\).
  The implication graph is depicted below on the left hand side, the corresponding
  resolution steps for conflict analysis below on the right hand side.\\

  \begin{center}
  \begin{tikzpicture}
    \node(x1) {\(\negated{x_1}\)};
    \node[below = 16pt of x1] (x2) {\(x_2\)};
    \node[below = 16pt of x2] (x3) {\(x_3\)};
    \node[right = 30pt of x1] (t2) {\(\negated{t_3}\)};
    \node[right = 30pt of x2] (t0) {\(\negated{t_1}\)};
    \node[right = 30pt of x3] (t1) {\(\negated{t_2}\)};
    \node[right = 30pt of t0] (kappa) {\(\kappa\)};
    \draw[-{to[length=5.5mm]}] (x1) to node[above, midway] (TextNode) {\(C_8\)} (t2);
    \draw[-{to[length=5.5mm]}] (x2) to node[above, midway] (TextNode) {\(C_2\)} (t0);
    \draw[-{to[length=5.5mm]}] (x3) to node[above, midway] (TextNode) {\(C_5\)} (t1);
    \draw[-{to[length=5.5mm]}] (t2) to [bend left=10] node[above, midway] (TextNode) {\(C_7\)} (kappa);
    \draw[-{to[length=5.5mm]}] (t0) to node[above, midway] (TextNode) {\(C_7\)} (kappa);
    \draw[-{to[length=5.5mm]}] (t1) to [bend right=15] node[below, midway] (TextNode) {\(C_7\)} (kappa);
    \node (11) at ($(t2.south east)+(90pt,-5pt)$) {\((t_1 \vee t_2 \vee t_3)\)};
    \node[right = 0.3cm of 11] (12) {\((\negated{t_2} \vee \negated{x_3})\)};
    \draw (11.south west) -- (12.south east);
    \node[below = 0.3cm of $(11)!0.5!(12)$] (21) {\((t_1 \vee t_3 \vee \negated{  x_3})\)};
    \node[right = 0.3cm of 21] (22) {\((\negated{t_1} \vee \negated{x_2})\)};
    \draw (21.south west) -- (22.south east);
    \node[below = 0.3cm of $(21)!0.5!(22)$] (31) {\((t_3 \vee \negated{x_3} \vee \negated{x_2})\)};
    \node[right = 0.3cm of 31] (32) {\((\negated{t_3} \vee x_1)\)};
    \draw (31.south west) -- (32.south east);
    \node[below = 0.3cm of $(31)!0.5!(32)$] {\((\negated{x_3} \vee \negated{x_2} \vee x_1)\)};
  \end{tikzpicture}
  \end{center}

  Conflict analysis returns the clause \((\negated{x_3} \vee \negated{x_2} \vee x_1)\),
  which, after being added to \nolinebreak \(P\), blocks the model \nolinebreak
  \(\negated{x_1}\,x_2\,x_3\), which does not overlap with the previously found model
  \nolinebreak \(x_1\).
\end{example}

\section{Projected Model Enumeration Without Repetition}
\label{sec:pmeirred}

We are given a propositional formula \nolinebreak \(F(X,Y)\) over the set of
\CHANGED{relevant} variables \nolinebreak \(X\) and the set of \CHANGED{irrelevant} variables
\nolinebreak \(Y\), and our task is to enumerate its models projected onto the
variables in \nolinebreak \(X\).
Let \nolinebreak \(P(X,Y,S)\) and \nolinebreak \(N(X,Y,T)\) be 
a dual representation of \nolinebreak \(F\)
according to \autoref{sec:prelim}. 
Obviously, \autoref{eq:modelspf}--\autoref{eq:pton} hold.
Now we call a CDCL-based SAT solver on \nolinebreak \(P\).
Whenever it finds a total model \nolinebreak \(I\) of \nolinebreak \(P\), it is
shrunken by dual reasoning obtaining \nolinebreak \(\istar\) which also
satisfies \nolinebreak \(P\).
The decision level of \nolinebreak \(\istar\) might be significantly smaller
than the one of \nolinebreak \(I\), and backtracking to decision level
\(b = \level(\negated{\project{\istar}{X}})-1\) mimics non-chronological
backtracking in \nolinebreak CDCL.
Notice that \nolinebreak \(\istar\) is used solely for determining the
backtracking level. 
The blocking clause added to \(P\) consists of all decisions with decision level
smaller or equal to \nolinebreak \(b+1\) and propagates after
backtracking.

In \autoref{fig:enumirred}, we consider the case 
with permanent learning of the blocking clauses.  
Let \SATone execute standard \nolinebreak CDCL on \nolinebreak \(P\) and let
\nolinebreak \(I\) denote its trail.
Obviously it finds only total models of \nolinebreak \(P\).
Due to \autoref{eq:modelspf}, these models satisfy \nolinebreak \(F\), too.
Now assume a (total) model \nolinebreak \(I\) of \nolinebreak \(P\) is found. 
A second \nolinebreak SAT solver \nolinebreak \SATtwo is incrementally invoked on
\nolinebreak \(\project{I}{X \cup Y} \wedge N\) with the aim to shrink
\nolinebreak \(I\)
obtaining \nolinebreak \(\istar\) as
described in \autoref{sec:shrinking}.

\begin{figure}
  \setcounter{Line}{1}

  \begin{tabular}{@{}r@{\hspace{.85em}}l@{}}
    \L{} \hspace{-3cm}\begin{minipage}[l]{2.9cm}
      {\hfill\textbf{Input:}\hspace{.75em}~}\\[4.2ex] 
    \end{minipage}
    \begin{minipage}[l]{0.9\textwidth}
      formulae \(P(X,Y,S)\) and \(N(X,Y,T)\) \stt \\
      \(\text{\hspace{4.1em}} \exists S \,[\, P(X,Y,S) \,]\, 
       \equiv \negated{\exists T \,[\, N(X,Y,T) \,]\,}\),\\
       set of variables \(X \cup Y \cup S \cup T\)
       \\[-1ex]
    \end{minipage}
    \N
    \L{} \hspace{-3cm}\begin{minipage}[l]{2.9cm}
      {\hfill\textbf{Output:}\hspace{.75em}~}\\
    \end{minipage}
    \begin{minipage}[l]{0.7\textwidth}
      DNF representation of \(\project{P}{X}\)\\
    \end{minipage} \N
    \L{}
    \kern-1.5em\scalebox{1.05}{\mainalgirredname}\hspace{2.0pt}(\,\(P\),~\(N\)\,) \hspace{2em}
    \C{\(P_0 = \S{CNF}\hspace{2.0pt}(\,F\,)\)}\\
    \L{} \hspace{14.24em} \C{\(N_0 = t_0 \wedge \S{CNF}\hspace{2.0pt}(\,t_0\,\leftrightarrow\,\negated{F}\,\))} \NH
    \L{\theLine} \(I \mdef \emptytrail\) \N
    \stepcounter{Line}
    \L{\theLine} \(\level\levelupd{V \mapsto \infty}\) \N
    \stepcounter{Line}
    \L{\theLine} \(M \mdef \false\) \N
    \stepcounter{Line}
    \L{\theLine} \(i \mdef 0\) \N
    \stepcounter{Line}
    \L{\theLine} \K{forever do} \N
    \stepcounter{Line}
    \L{\theLine} \I \(i \mdef i+1\) \N
    \stepcounter{Line}
    \L{\theLine} \I \(C \mdef\) \S{\unitpropalgirredname}\hspace{2.0pt}(\,\(P\),~\(I\),~\(\level\)\,)   \N
    \stepcounter{Line}
    \L{\theLine} \I \K{if} \(C \neq \false\) \K{then} \N
    \stepcounter{Line}
    \L{\theLine} \I\I \(c \mdef \level(C)\) \N
    \stepcounter{Line}
    \L{\theLine} \I\I \K{if} \(c = 0\) \K{then} \N
    \stepcounter{Line}
    \L{\theLine} \I\I\I \K{return} \(M\)\N
    \stepcounter{Line}
    \L{\theLine} \I\I \K{else} \N
    \stepcounter{Line}
    \L{\theLine} \I\I\I \S{\conflanaalgirredname}\hspace{2.0pt}(\,\(P\),~\(I\),~\(C\),~\(\level\)\,) \N
    \stepcounter{Line}
    \L{\theLine} \I \K{else}  \N
    \stepcounter{Line}
    \L{\theLine} \I\I \K{if} all variables in \(X \cup Y \cup S\) are assigned
    \K{then} \N
    \stepcounter{Line}
    \L{\theLine} \I\I\I \C{\(I\) is total model of \(P\) and \(F\)} \N
    \stepcounter{Line}
    \L{\theLine} \I\I\I \K{if} \(\var(\decs(I)) \cap X = \emptyset\) \K{then} \N
    \stepcounter{Line}
    \L{\theLine} \I\I\I\I \K{return} \(M \vee \project{I}{X}\) \NB
    \stepcounter{Line}
    \L{\theLine} \B \I\I\I \K{else} \NB
    \stepcounter{Line}
    \L{\theLine} \B \I\I\I\I \(\istar \mdef\)
    \S{CSet}\hspace{2.0pt}(\,\(N\),~\(\project{I}{X \cup Y}\)\,) \NB
    \stepcounter{Line}
    \L{\theLine} \B \I\I\I\I \C{\(\istar\) is model of \(\project{F}{X \cup Y}\) and
    conflict set of \(I\) \wrtt \(N\)} \NB 
    \stepcounter{Line}
    \L{\theLine} \B \I\I\I\I \(P \mdef P \wedge \negated{\project{\decs(\istar)}{X}}\) \NB
    \stepcounter{Line}
    \L{\theLine} \B \I\I\I\I \(N \mdef N \setminus \{(t_0 \vee \bigvee_{j=1}^{i-1}t_j)\} \wedge \,\) \NB
    & \B \hspace{8.8em} \((t_0 \vee \bigvee_{j=1}^{i} t_j) \wedge 
     \S{CNF}\hspace{2.0pt}(\,t_i\,\leftrightarrow\,\project{\decs(\istar)}{X}\,)\) \NB
      \stepcounter{Line}
    \L{\theLine} \B \I\I\I\I \(M \mdef M \vee \project{\istar}{X}\) \NB
    \stepcounter{Line}
    \L{\theLine} \B \I\I\I\I \(b \mdef \level(\negated{\project{\istar}{X}})\) \NB
    \stepcounter{Line}
    \L{\theLine} \B \I\I\I\I \S{\backtrackalgname}\hspace{2.0pt}(\,\(I\),~\(b-1\)\,) \N
    \stepcounter{Line}
    \L{\theLine} \I\I \K{else} \N 
    \stepcounter{Line}
    \L{\theLine} \I\I\I \S{\decidealgirredname}\hspace{2.0pt}(\,\CHANGED{\(P\),~\(I\),~\(\level\)}\,)
  \end{tabular}
  \caption{
    Irredundant model enumeration.
    The black lines 1--18 and 27--28 describe \nolinebreak CDCL returning a model if one is found
    and the empty clause otherwise.
    \CHANGED{Notice that for decisions, variables in \nolinebreak \(X\) are prioritized
    over variables in \nolinebreak \((Y \cup S)\) to avoid multiple enumeration
    of projected models.
    Similarly, in line 17 it suffices to check whether no relevant decision
    literals are left on the trail.}
    The blue part, i.e., lines 19--26, represents the extension to model enumeration.
    A second \nolinebreak SAT solver is called incrementally on \nolinebreak
    \(N\)
    assuming the literals on \(\project{I}{X \cup Y}\).
    A conflict occurs by unit propagation only, and \nolinebreak
    \(\project{\istar}{X \cup Y}\)
    is a (partial) model of \nolinebreak \(F\).
    It is encoded as a dual blocking clause, and \nolinebreak \(P\) and
    \nolinebreak \(N\) are updated accordingly.
    \label{fig:enumirred}
  }
\end{figure}

Let \(b\) denote the decision level of \nolinebreak
\(\project{\negated{\istar}}{X}\) and \nolinebreak \(\ell\) be the  literal in
\nolinebreak \(\project{\negated{\istar}}{X}\) with decision level \nolinebreak
\(b\).
We now add the clause \nolinebreak \(\project{\negated{\CHANGED{\decs(\istar)}}}{X}\) to \(P\)
and backtrack to decision level \nolinebreak \(b-1\).
Notice that \nolinebreak \(\project{\negated{\CHANGED{\decs(\istar)}}}{X}\) acts in \nolinebreak
\(P\) as a blocking clause and must not be deleted anytime which might blow up
\nolinebreak \(P\) and slow down \nolinebreak \SATone.
Moreover, the dual encoding of the blocking clause according to
\autoref{sec:dualblock} ensures \autoref{eq:pnotn} on which our
method relies.

In \autoref{sec:mainalg}, we present the main function \mainalgirredname.
Unit propagation (\autoref{sec:unitprop}) and the schema for conflict analysis
(\autoref{sec:confana}) are the same as in 
\nolinebreak CDCL for \nolinebreak SAT.

\subsection{Main Algorithm}
\label{sec:mainalg}

The function \mainalgirredname in \autoref{fig:enumirred} describes the main algorithm.
(Black rows \nolinebreak 1--18 and \nolinebreak 27--28 represent standard
\nolinebreak CDCL \CHANGED{with projection} returning a model if the formula under consideration is
satisfiable and the empty clause otherwise, blue 
rows \noindent 19--26 the rest of the algorithm.)

Initially, the trail \nolinebreak \(I\) is empty, the target \nolinebreak DNF
\nolinebreak \(M\) is
\nolinebreak \(\false\), and all variables are unassigned, \iet, assigned 
decision level \nolinebreak \(\infty\).
Unit propagation is executed until either a conflict occurs or all
variables are assigned a value (line \nolinebreak 7).

If a conflict occurs at decision level zero, the search space has been processed
exhaustively, and the enumeration terminates (lines \nolinebreak 8--11).
If a conflict occurs at a decision level higher than zero, conflict analysis is
executed (line \nolinebreak 13).

If no conflict occurs and all variables are assigned, a total model has been
found \nolinebreak (line \nolinebreak 15).
If no relevant decisions are left on the trail \nolinebreak \(I\), the
\CHANGED{relevant} search
space has been processed exhaustively, the found model is output and the search
terminates (lines \nolinebreak 17--18).
If \(I\) contains a relevant decision, the found model is shrunken \nolinebreak
(line \nolinebreak 20) by means of dual reasoning as described in
\autoref{sec:shrinking}.
It is blocked, and the last relevant decision literal is flipped (lines
\nolinebreak 22--26).
If no conflict occurs and not all variables are assigned, 
a decision is taken (line \nolinebreak 28)\CHANGED{, where the variables in
  \nolinebreak \(X\) are prioritized over the variables in \nolinebreak \(Y
  \cup S\) to avoid enumerating models which only differ in irrelevant and
  Tseitin variables}.  

\begin{figure}[!t]
  \centering
  \begin{tabular}{@{}r@{\hspace{.85em}}l@{}}
    \setcounter{Line}{1}
    \L{} \kern-1.5em\scalebox{1.05}{\unitpropalgirredname}\hspace{2.0pt}(\,\(F\),~\(I\),~\(\level\)\,) \NH  
    \L{\theLine} \K{while} some \(C \in F\) is unit \((\ell)\) under \(I\) \K{do} \N
    \stepcounter{Line}
    \L{\theLine} \I \(I \mdef I\,\ell\) \N
    \stepcounter{Line}
    \L{\theLine} \I \(\level(\ell) \mdef \level(I)\) \N
    \stepcounter{Line}
    \L{\theLine} \I \K{for all} clauses \(D \in F\) containing \(\negated\ell\) \K{do} \N
    \stepcounter{Line}
    \L{\theLine} \I\I \K{if} \(I(D) = \false\) \K{then} \K{return} \(D\) \N
    \stepcounter{Line}
    \L{\theLine} \K{return} \(\false\) \NS
    \setcounter{Line}{1}
    \L{} \kern-1.5em\scalebox{1.05}{\conflanaalgirredname}\hspace{2.0pt}(\,\(F\),~\(I\),~\(C\),~\(\level\)\,) \NH
    \L{\theLine} \I \(D \mdef\) \S{Learn}\hspace{2.0pt}(\,\(I\),~\(C\)\,) \N
    \stepcounter{Line}
    \L{\theLine} \I \(F \mdef F \wedge D\) \N
    \stepcounter{Line}
    \L{\theLine} \I \(\ell \mdef\) literal in \(D\) at decision level \(\level(I)\) \N
    \stepcounter{Line}
    \L{\theLine} \I \(j \mdef \level(D \setminus \{\ell\})\) \N
    \stepcounter{Line}
    \L{\theLine} \K{for all} literals \(k \in I\) with decision level \(> j\) \K{do} \N
    \stepcounter{Line}
    \L{\theLine} \I assign \(k\) decision level \(\infty\) \N
    \stepcounter{Line}
    \L{\theLine} \I remove \(k\) from \(I\) \N
    \stepcounter{Line}
    \L{\theLine} \(I \mdef I\,\ell\) \N
    \stepcounter{Line}
    \L{\theLine} \(\level(\ell) \mdef j\)
  \end{tabular}
  \caption{
    The function \unitpropalgirredname implements unit propagation
    in \nolinebreak \(F\).
    The unit literal \nolinebreak \(\ell\) is assigned the decision level of
    \nolinebreak \(I\). 
    If some clause \nolinebreak \(D \in F\) containing the complement of
    \nolinebreak \(\ell\) becomes falsified, \unitpropalgirredname
    returns \nolinebreak \(D\).
    Otherwise it returns the empty clause \nolinebreak \(\false\) indicating
    that no conflict has occurred.
    The function \conflanaalgirredname is called whenever a clause
    \nolinebreak \(C \in F\) becomes empty under the current assignment.
    It learns a clause \nolinebreak \(D\) starting with the conflicting
    clause \nolinebreak \(C\). 
    The solver then backtracks to the second highest decision level
    \nolinebreak \(j\) in \nolinebreak \(D\) upon which \nolinebreak\(D\)
    becomes unit with unit literal \nolinebreak \(\ell\), and propagates
    \nolinebreak \(\ell\).
    \label{fig:prop_ana}
  }
\end{figure}

\subsection{Unit Propagation}
\label{sec:unitprop}

Unit propagation is described by the function \unitpropalgirredname in
\autoref{fig:prop_ana}.
It takes as input the formula \nolinebreak \(F\), the trail \nolinebreak \(I\),
and the decision level function \nolinebreak \(\delta\).
If a clause \nolinebreak \(C \in F\) is unit under \(I\), its unit literal
\nolinebreak \(\ell\) is propagated, \iet, \nolinebreak \(I\) is extended by
\nolinebreak \(\ell\) (line \nolinebreak 2).
Propagated literals are assigned at the current decision level (line \nolinebreak 3) as is usual
in modern CDCL-based \nolinebreak SAT solvers.
If the resulting trail falsifies some clause \nolinebreak \(D \in F\), this clause
is returned (lines \nolinebreak 4--5).
Otherwise the function returns the empty clause \nolinebreak \(\false\) (line
\nolinebreak 6).

\subsection{Conflict Analysis}
\label{sec:confana}

Conflict analysis is described by the function \conflanaalgirredname in
\autoref{fig:prop_ana}.
It takes as input the formula \nolinebreak \(F\), the trail \nolinebreak \(I\),
the conflicting clause \nolinebreak \(C\), and the decision level function
\nolinebreak \(\level\).
A clause \nolinebreak \(D\) is learned as described in \autoref{sec:learn} and
added to \nolinebreak \(F\) (lines \nolinebreak 1--2).
The second highest decision level \nolinebreak \(j\) in \nolinebreak \(D\) is
determined (lines \nolinebreak 3--4), and the enumerator backtracks
(non-chronologically) to decision level \nolinebreak \(j\).
Backtracking involves unassigning all literals with decision level higher than
\nolinebreak \(j\) (lines \nolinebreak 5--7).
After backtracking, the clause \nolinebreak \(D\) becomes unit with unit literal
\nolinebreak \(\ell\), which is propagated and assigned decision level
\nolinebreak \(j\) (lines \nolinebreak 8--9).

\section{Formalizing Projected Irredundant Model Enumeration}
\label{sec:formalenumirred}

In this section, we provide a formalization of our algorithm presented in
\nolinebreak \autoref{sec:pmeirred}.
Let \nolinebreak \(F(X,Y)\) be a formula defined onto the set of relevant
(input) variables \nolinebreak \(X\) and the set of irrelevant (input) variables
\nolinebreak \(Y\), and assume our task is to enumerate its models projected onto
\nolinebreak \(X\).

Our formalization works on a dual representation of \nolinebreak \(F\), 
given by the formulae \nolinebreak \(P(X,Y,S)\) and \nolinebreak \(N(X,Y,T)\), as introduced
in \autoref{sec:dual}.
So, \nolinebreak \(P(X,Y,S)\) and \nolinebreak \(N(X,Y,T)\) are defined
over the same sets of relevant variables \nolinebreak \(X\) and irrelevant
variables \nolinebreak \(Y\) as well as the disjoint sets of variables
\nolinebreak \(S\) and \nolinebreak \(T\), respectively, which are defined in
terms of the variables in \nolinebreak \(X \cup Y\).
Recall that \autoref{eq:modelspf}--\autoref{eq:pton} hold.
We start by sketching the enumeration process. 
After presenting our calculus, 
we show its working by means of an example before providing a
correctness proof.

The process works as follows.
Let \nolinebreak \(I\) denote the current trail.
Unit propagation is carried out as long as \nolinebreak \(\residual{P}{I}\)
contains a unit literal (rule \UnitName).
If all variables in \(X \cup Y \cup S\) are assigned and no conflict has
occurred, a total model of \nolinebreak \(P\) has been found.
In case there is no decision left on \nolinebreak \(I\), its projection onto
\nolinebreak \(X\) is enumerated, and we are done (\EtopName).
Otherwise, the model \nolinebreak \(I\) is shrunken and blocked by means of the
dual blocking clause encoding (\BtopIrredName).
If a conflict  occurred and there is no decision left on \nolinebreak \(I\), the
process terminates (\EbotName).
Otherwise, conflict analysis is carried out and backtracking occurs
(\BbotName).
If after executing exhaustive unit propagation there are still unassigned
variables in \(X \cup Y \cup S\), a decision need be taken.
We are interested in the models projected onto \nolinebreak \(X\).
To avoid detecting models which differ only in variables \(Y \cup S\), we
first decide variables in \nolinebreak \(X\) (\DecXIrredName), before deciding
variables in \(Y \cup S\) (\DecYSIrredName).

\subsection{Calculus}
\label{sec:calculusenumirred}

We formalize the algorithm presented in \autoref{sec:pmeirred} as a state
transition system with transition relation \nolinebreak \(\transenumirred\).
Non-terminal states are described by tuples \nolinebreak
\(\state{P}{N}{M}{I}{\level}\).
The third element, \nolinebreak \(M\), is a \nolinebreak DSOP formula over
variables in \nolinebreak \(X\).
The fourth element, \nolinebreak \(I\), denotes the trail defined
over variables in \(X \cup Y \cup S \cup T\), and 
\nolinebreak \(\level\)
denotes the decision level function.
The initial state is \nolinebreak
\(\state{P_0}{N_0}{\false}{\emptytrail}{\level_0}\), where \nolinebreak \(P_0\)
and \(N_0\) denote the initial \nolinebreak CNF representations of \nolinebreak
\(F\) and \nolinebreak \(\negated{F}\), respectively, \nolinebreak
\(\emptytrail\) denotes the empty trail, and \nolinebreak \(\level_0 \equiv
\infty\).
The terminal state is given by a \nolinebreak DSOP formula \nolinebreak \(M\),
which is equivalent
to the projection of \nolinebreak \(P\) onto \nolinebreak \(X\).
The transition relation \nolinebreak \(\transenumirred\) is the union of
transition relations \nolinebreak \(\leadstoname{R}\), where
\(\functionname{R}\) is either \EtopName, \EbotName, \UnitName, \BtopIrredName,
\BbotName, \DecXIrredName, or \DecYSIrredName.
The rules are listed in \autoref{fig:calculusenumirred}.

\begin{figure}
  \setlength{\fboxsep}{0pt}
  \noindent\fbox{
    \parbox{0.979\textwidth}{
      ~\\[\vspacemidrule]
      \renewcommand{\arraystretch}{1.2}
      \setlength\tabcolsep{2pt}
      \centering
      \begin{tabular}{@{}ll@{}}
        \Etop\\[\vspacebetweenrules]
        \Ebot\\[\vspacemidrule]
        \midrule\\[\vspaceaftermidrule]
        \Unit\\[\vspacemidrule]
        \midrule\\[\vspaceaftermidrule]
        \BtopIrred\\[\vspacebetweenrules]
        \Bbot\\[\vspacemidrule]
        \midrule\\[\vspaceaftermidrule]
        \DecXIrred\\[\vspacebetweenrules]
        \DecYSIrred\\[\vspacemidrule]
      \end{tabular}
    }
  }
  \caption{
    Rules for projected model enumeration without repetition.
    States are represented as tuples \(\state{P}{N}{M}{I}{\level}\).
    The formulae \nolinebreak \(P(X,Y,S)\) and \nolinebreak \(N(X,Y,T)\) are a
    dual representation of the formula \nolinebreak \(F(X.Y)\) whose models
    projected onto \nolinebreak \(X\) are sought.
    These models are recorded in the initially empty \nolinebreak DNF
    \nolinebreak \(M\).
    The last two elements, \nolinebreak \(I\) and \nolinebreak \(\level\),
    denote the current trail and decision level function, respectively.
    If a model is found or a conflict encountered and the search space has been
    processed exhaustively, the search terminates (rules \nolinebreak \EtopName
    and \nolinebreak \EbotName).
    Otherwise, either the model is shrunken and a dual blocking clause is added (rules
    \nolinebreak \BtopIrredName) or conflict analysis is executed followed by
    non-chronological backtracking (rule \nolinebreak \BbotName).
    If the residual of \nolinebreak \(P\) under the current trail \nolinebreak
    \(I\,J\) contains a unit literal, it is propagated (rule \nolinebreak
    \UnitName).
    If none of the mentioned preconditions are met, a decision is taken.
    Relevant variables are prioritized (rule \nolinebreak \DecXIrredName) over
    irrelevant and internal ones (rule \nolinebreak \DecYSIrredName).
  }
  \label{fig:calculusenumirred}
\end{figure}

\skipbetweenrules

\noindent
\EtopName. \skipafterrulename
All variables are assigned and no conflict in \nolinebreak \(P\) occurred,
hence
the trail \nolinebreak \(I\) is a total model of \nolinebreak \(P\).
It contains no relevant decision
indicating that the relevant search space has been processed exhaustively.
The model projected onto \nolinebreak \(X\) is added to \nolinebreak \(M\), and
the search terminates.
It is sufficient to check for relevant decisions, since flipping an irrelevant
one would result in detecting redundant models projected onto \nolinebreak \(X\).
However, due to the addition of blocking clauses, a conflict would occur,
and checking for relevant decisions essentially saves work.

\skipbetweenrules

\noindent
\EbotName. \skipafterrulename
A conflict at decision level zero has occurred indicating that the search space
has been processed exhaustively.
The search terminates leaving \nolinebreak \(M\) unaltered.
We need to make sure no decision is left on the trail, which in particular includes the irrelevant
ones.
The reason is that after flipping any decision---in particular
also irrelevant and internal ones---the resulting trail might be extended to a model of
\nolinebreak \(P\).

\skipbetweenrules

\noindent
\UnitName. \skipafterrulename
No conflict in \nolinebreak \(P\) occurred, and a clause in \nolinebreak \(P\) is unit under
\nolinebreak \(I\).
Its unit literal \nolinebreak \(\ell\) is propagated and assigned the current
decision level.

\skipbetweenrules

\noindent
\BtopIrredName. \skipafterrulename
All variables are assigned, no conflict in \nolinebreak \(P\) occurred, and
the trail \nolinebreak \(I\)
is a total model of \nolinebreak \(P\). 
It is shrunken as described in \autoref{sec:shrinking} obtaining \nolinebreak
\(\istar\).
The projection \nolinebreak \(m\) of  \nolinebreak \(\istar\) onto \nolinebreak \(X\)
is added to \nolinebreak \(M\).
The clause \nolinebreak \(B\) consisting of the negated decision literals of \nolinebreak \(m\) is
added as a blocking clause to \nolinebreak \(P\).
Its negation \nolinebreak \(\negated{B}\) is added disjunctively to \nolinebreak \(N\),
which is transformed back into \nolinebreak CNF by means of the Tseitin
transformation.\footnote{Notice that although not stated explicitly in favor of
simplifying the presentation, the dual blocking clause encoding introduced in
\autoref{sec:dualblock} may be used (see lines \nolinebreak 22--23 of algorithm
\mainalgirredname in \autoref{fig:enumirred}).} 
The solver backtracks to the second highest decision level in \nolinebreak
\(B\) and propagates \nolinebreak \(\ell\) at the current decision level, \iet,
basically flips the relevant
decision literal with highest decision level.

\skipbetweenrules

\noindent
\BbotName. \skipafterrulename
The current trail falsifies a clause in \nolinebreak \(P\) at a decision level
greater than zero indicating that the search space has not yet been processed
exhaustively.
Conflict analysis returns a clause \nolinebreak \(D\) implied by \nolinebreak \(P\),
which is added to \nolinebreak \(P\) and marked as redundant.
The solver backtracks to the second highest decision level \nolinebreak \(j\) in \nolinebreak
\(D\).
The learned clause \nolinebreak \(D\) becomes unit and its unit literal
\nolinebreak \(\ell\) is
propagated at decision level \nolinebreak \(j\).
In contrast to \EtopName, every decision literal need be flipped, which particularly
applies to irrelevant and internal decision literals.

\skipbetweenrules

\noindent
\DecXIrredName. \skipafterrulename
No conflict has occurred, the residual of \nolinebreak \(P\) under \nolinebreak
\(I\) contains no units, and there is an unassigned relevant literal
\nolinebreak \(\ell\).
The current decision level is incremented to \nolinebreak \(d\), the literal
\nolinebreak \(\ell\) is decided and assigned to decision level \nolinebreak \(d\). 

\skipbetweenrules

\noindent
\DecYSIrredName. \skipafterrulename
No conflict has occurred, and the residual of \nolinebreak \(P\) under \nolinebreak
\(I\) contains no units.
All relevant literals are assigned, and there is an unassigned irrelevant or
internal literal \nolinebreak \(\ell\).
The current decision level is incremented to \nolinebreak \(d\),
\nolinebreak \(\ell\) is decided and assigned decision level \nolinebreak \(d\).

\subsection{Example}
\label{sec:calculusenumirredexample}
The working of our calculus is shown by means of an example.
Consider again \autoref{ex:multenumcdcl} and \autoref{ex:projenum}.
We have
  \begin{equation*}
    F =
    \underbrace{(a \vee c)}_{C_1} \wedge
    \underbrace{(a \vee \negated{c})}_{C_2} \wedge
    \underbrace{(b \vee d)}_{C_3} \wedge
    \underbrace{(b \vee \negated{d})}_{C_4}
  \end{equation*}
  and assume the set of relevant variables is \(X = \{a, c\}\) and the set of
  irrelevant variables is \(Y = \{b, d\}\).
  The formula \nolinebreak \(F\) is already in \nolinebreak CNF, therefore we
  define \(P_0 = F\) and accordingly \(S_0 = \emptyset\). 
  For its negation
  \begin{equation*}
    \negated{F} =
    (\negated{a} \wedge \negated{c}) \vee
    (\negated{a} \wedge c) \vee
    (\negated{b} \wedge \negated{d}) \vee
    (\negated{b} \wedge d)
  \end{equation*}
  we define
  \begin{align*}
    N_0 = \,
    &
      \underbrace{(\negated{t_1} \vee \negated{a})}_{D_1} \wedge
      \underbrace{(\negated{t_1} \vee \negated{c})}_{D_2} \wedge
      \underbrace{(a \vee c \vee t_1)}_{D_3} \wedge\\
    &
      \underbrace{(\negated{t_2} \vee \negated{a})}_{D_4} \wedge
      \underbrace{(\negated{t_2} \vee c)}_{D_5} \wedge
      \underbrace{(a \vee \negated{c} \vee t_2)}_{D_6} \wedge\\
    &
      \underbrace{(\negated{t_3} \vee \negated{b})}_{D_7} \wedge
      \underbrace{(\negated{t_3} \vee \negated{d})}_{D_8} \wedge
      \underbrace{(b \vee d \vee t_3)}_{D_9} \wedge\\
    &
      \underbrace{(\negated{t_4} \vee \negated{b})}_{D_{10}} \wedge
      \underbrace{(\negated{t_4} \vee d)}_{D_{11}} \wedge
      \underbrace{(b \vee \negated{d} \vee t_4)}_{D_{12}} \wedge\\
    &
      \underbrace{(t_1 \vee t_2 \vee t_3 \vee t_4)}_{D_{13}}
  \end{align*}
  with the set of internal variables \(T_0 = \{t_1, t_2, t_3, t_4\}\).
  Assume a lexicographic ordering of the input variables, \iet, \(a \succlex b
  \succlex c \succlex d\), and assume we choose the decision variable according
  to this ordering.  
  The execution steps are depicted in \autoref{fig:calculusenumirredexample}.

  \begin{figure}
    \begin{center}
      \begin{tabular}{cllccc}
        \toprule
        Step & ~Rule & ~\(I\) & \(\residual{P}{I}\) & \(N\) & \(M\)\\
        \midrule
        0
        &
        &
          \(\emptytrail\)
        &
          \(P_0\)
        &
          \(N_0\)
        &
          \(\false\)\\
        \rowcolor{black!10} 1
        &
          \DecXIrredName
        &
          \(\decided{a}\)
        &
          \((b \vee d) \wedge (b \vee \negated{d})\)
        &
          \(N_0\)
        &
          \(\false\)\\
        2
        &
          \DecXIrredName
        &
          \(\decided{a}\,\decided{c}\)
        &
          \((b \vee d) \wedge (b \vee \negated{d})\)
        &
          \(N_0\)
        &
          \(\false\)\\
        \rowcolor{black!10} 3 
        &
          \DecYSIrredName
        &
          \(\decided{a}\,\decided{c}\,\decided{b}\)
        &
          \(\true\)
        &
          \(N_0\)
        &
          \(\false\)\\
        4
        &
          \DecYSIrredName
        &
          \(\decided{a}\,\decided{c}\,\decided{b}\,\decided{d}\)
        &
          \(\true\)
        &
          \(N_0\)
        &
          \(\false\)\\
        \rowcolor{black!10} 5
        &
          \BtopIrredName
        &
          \(\propagated{\negated{a}}{B_1}\)
        &
          \((c) \wedge (\negated{c}) \wedge (b \vee d) \wedge (b \vee \negated{d})\)
        &
          \(N_1\)
        & \(a\)\\
        6
        &
          \UnitName
        &
          \(\propagated{\negated{a}}{B_1}\,\propagated{c}{C_1}\)
        &
          \(() \wedge (b \vee d) \wedge (b \vee \negated{b})\)
        &
          \(N_1\)
        &
          \(a\)\\
        \rowcolor{black!10} 7
        &
          \EbotName
        &
          
        &
          
        &
         
        &
          \(a\)\\
        \bottomrule
      \end{tabular}
    \end{center}
    \caption{Execution trace for \(F = (a \vee c) \wedge (a \vee \negated{c})
      \wedge (b \vee d) \wedge (b \vee \negated{d})\) defined over the set of
      relevant variables \(X = \{a, c\}\) and the set of irrelevant variables
      \(Y = \{b, d\}\) (see also \autoref{ex:multenumcdcl} and
      \autoref{ex:projenum}).}
    \label{fig:calculusenumirredexample}
  \end{figure}

\skipbetweensteps

\noindent
Step 0: \skipafterstepnumber
The initial state is given by the empty trail \nolinebreak \(\emptytrail\), the
\nolinebreak CNF formulae \nolinebreak \(P_0\) and \nolinebreak \(N_0\), and
the empty \nolinebreak DNF formula \nolinebreak \(\false\).

\skipbetweensteps

\noindent
Step 1: \skipafterstepnumber
The formula \nolinebreak \(P_0\) contains no units and there are unassigned
relevant variables.
The preconditions of rule \DecXIrredName are met, and decision \(a\) is taken.

\skipbetweensteps

\noindent
Step 2:
No conflict occurred, \nolinebreak \(\residual{P_0}{I}\) contains no units, and
there are unassigned relevant variables. 
The preconditions of rule \DecXIrredName are met, and \(c\) is decided.

\skipbetweensteps

\noindent
Step 3: \skipafterstepnumber
No conflict occurred, and \nolinebreak \(\residual{P_0}{I}\) contains no units.
All relevant variables are assigned and
there are unassigned irrelevant variables. 
The preconditions of rule \DecYSIrredName are met, and decision \(b\) is taken.
Notice that \nolinebreak \(I\) already satisfies \nolinebreak \(P_0\), but the
solver is not able to detect this fact.

\skipbetweensteps

\noindent
Step 4: \skipafterstepnumber
Again, the preconditions of rule \DecYSIrredName are met, and decision \(d\)
is taken.

\skipbetweensteps

\noindent
Step 5: \skipafterstepnumber
No conflict occurred and all variables are assigned, hence \nolinebreak \(I\) is
a model of \nolinebreak \(P_0\).
It is shrunken following the procedure described in \autoref{sec:shrinking}.
We call a \nolinebreak SAT solver on \(N_0 \wedge I\)
assuming the literals on \nolinebreak \(I\).
A conflict in \nolinebreak \(N_0\) occurs by propagation of variables in
\nolinebreak \(T_0\), and conflict analysis provides us with the shrunken
model \nolinebreak \(a\,b\) of \nolinebreak \(F\).
The resulting implication graph and trail are as follows:

\begin{center}
  \begin{tikzpicture}
    \node(a) {\(a\)};
    \node[below = 8pt of a] (c) {\(c\)};
    \node[below = 8pt of c] (b) {\(b\)};
    \node[below = 8pt of b] (d) {\(d\)};
    \node[right = 50pt of a] (t2) {\(\negated{t_2}\)};
    \node[above = 8pt of t2] (t1) {\(\negated{t_1}\)};
    \node[right = 50pt of b] (t3) {\(\negated{t_3}\)};
    \node[below = 8pt of t3] (t4) {\(\negated{t_4}\)};
    \node[right = 70pt of $(t2)!0.5!(t3)$] (kappa) {\(\kappa\)};
    \draw[-{to[length=5.5mm]}] (a) to [bend left=15] node[above, midway] (TextNode) {\(D_1\)} (t1);
    \draw[-{to[length=5.5mm]}] (a) to node[above, midway] (TextNode) {\(D_4\)} (t2);
    \draw[-{to[length=5.5mm]}] (b) to node[below, midway] (TextNode) {\(D_7\)} (t3);
    \draw[-{to[length=5.5mm]}] (b) to [bend right=15] node[below, midway] (TextNode) {\(D_{11}\)} (t4);
    \draw[-{to[length=5.5mm]}] (t1) to [bend left=15] node[above, midway] (TextNode) {\(D_{13}\)} (kappa);
    \draw[-{to[length=5.5mm]}] (t2) to [bend left=5] node[above, midway] (TextNode) {\(D_{13}\)} (kappa);
    \draw[-{to[length=5.5mm]}] (t3) to [bend right=5] node[below, midway] (TextNode) {\(D_{13}\)} (kappa);
    \draw[-{to[length=5.5mm]}] (t4) to [bend right=15] node[below, midway] (TextNode) {\(D_{13}\)} (kappa);
    \node[right = 10pt of kappa] (i) {\(
      I = \assumed{a}\,\assumed{c}\,\assumed{b}\,\assumed{d}\,\propagated{\negated{t_1}}{D_1}\,\propagated{\negated{t_2}}{D_4}\,\propagated{\negated{t_3}}{D_7}\,\propagated{\negated{t_4}}{D_{10}}\)};
  \end{tikzpicture}
\end{center}  

\noindent
The conflicting clause is \nolinebreak \(D_{13}\).
For conflict analysis, we resolve \nolinebreak \(D_{13}\) with \nolinebreak
\(D_{10}\) and the resolvent with \nolinebreak \(D_7\) followed by resolution with
\nolinebreak \(D_4\) and \nolinebreak \(D_1\).
The obtained clause \nolinebreak \((\negated{b} \vee \negated{a})\) contains
only assumed literals.
The assumptions \nolinebreak \(c\) and \nolinebreak \(d\) do not participate in
the conflict and therefore do not occur in the resulting clause.
Below, the resolution steps are visualized.

\begin{center}
  \begin{tikzpicture}
    \node (11) {\((t_1 \vee t_2 \vee t_3 \vee t_4)\)};
    \node[right = 0.5cm of 11] (12) {\((\negated{t_4} \vee \negated{b})\)};
    \draw (11.south west) -- (12.south east);
    \node[below = 0.3cm of $(11)!0.5!(12)$] (21) {\((t_1 \vee t_2 \vee t_3 \vee \negated{b})\)};
    \node[right = 0.5cm of 21] (22) {\((\negated{t_3} \vee \negated{b})\)};
    \draw (21.south west) -- (22.south east);
    \node[below = 0.3cm of $(21)!0.5!(22)$] (31) {\((t_1 \vee t_2 \vee \negated{b})\)};
    \node[right = 0.5cm of 31] (32) {\((\negated{t_2} \vee \negated{a})\)};
    \draw (31.south west) -- (32.south east);
    \node[below = 0.3cm of $(31)!0.5!(32)$] (41) {\((t_1 \vee \negated{b} \vee \negated{a})\)};
    \node[right = 0.5cm of 41] (42) {\((\negated{t_1} \vee \negated{a})\)};
    \draw (41.south west) -- (42.south east);
    \node[below = 0.3cm of $(41)!0.5!(42)$] (41) {\((\negated{b} \vee \negated{a})\)};
  \end{tikzpicture}
\end{center}

\noindent
The negation of \nolinebreak \((\negated{b} \vee \negated{a})\) is \nolinebreak
\(\istar = a\,b \leq I\) we are looking for.
The first model is \(m_1 =
\project{\istar}{X} = a\) and  
accordingly \(M_1 = M_0 \vee m_1\).
Furthermore, we have \(B_1 = \negated{\decs(m_1}) = (\negated{a})\), hence
\begin{align*}
  P_1 = \,
  &
    P_0 \wedge \underbrace{(\negated{a})}_{B_1} \quad\text{and}\\
  N_1 = \,
  &
    N_0 \vee \underbrace{(a)}_{\negated{B_1}}\\
  =\,
  &
    N_0 \setminus \{(\bigvee\limits_{j=1}^{4} t_j)\} \wedge (\bigvee\limits_{j=1}^{5} t_j) \wedge (t_5 \leftrightarrow a)\\
  =\,
  &
  (\bigwedge\limits_{i=1}^{12} D_i) \wedge
  \underbrace{(\negated{t_5} \vee a)}_{D_{14}} \wedge
  \underbrace{(\negated{a} \vee t_5)}_{D_{15}} \wedge
  \underbrace{(t_1 \vee t_2 \vee t_3 \vee t_4 \vee t_5)}_{D_{16}}
\end{align*}
where \(D_{14} \wedge D_{15} = (t_5 \leftrightarrow a)\) is the Tseitin
transformation of \(m_1\).
The clause \(\negated{B_1}\) is added disjunctively to \nolinebreak \(N_0\).
To retain \nolinebreak \(N\) in \nolinebreak CNF, \nolinebreak \(\negated{B_1}\) is encoded as \((t_5
\leftrightarrow \negated{B_1})\), \(t_5\)
is added to \nolinebreak \(D_{13}\) resulting in \nolinebreak \(D_{16}\) and
\(T_1 = T_0 \cup \{t_5\} = \{t_1, t_2, t_3, t_4, t_5\}\) as described in \autoref{sec:dualblock}.
The clause \nolinebreak \(B_1\) acts in \nolinebreak \(P\) as blocking clause.
The solver backtracks to decision level zero and propagates \nolinebreak
\(\negated{a}\) with reason \nolinebreak \(B_1\).

\skipbetweensteps

\noindent
Step 6: \skipafterstepnumber
The formula \nolinebreak \(\residual{P_1}{I}\) contains two units, \nolinebreak
\(\residual{C_1}{I} = (c)\) and \nolinebreak \(\residual{C_2}{I} =
(\negated{c})\). 
The literal \nolinebreak \(c\) is propagated with reason \nolinebreak \(C_1\).

\skipbetweensteps

\noindent
Step 7: \skipafterstepnumber
The trail falsifies \nolinebreak \(C_2\) and the current decision level is zero.
The preconditions of rule \EbotName are met and the search terminates without
altering \nolinebreak \(M = a\), which represents exactly the models of
\nolinebreak \(F\) projected onto \nolinebreak \(X\), namely \nolinebreak
\(a\,c\) and \nolinebreak \(a\,\negated{c}\).

\subsection{Proofs}
\label{sec:calculusenumirredproofs}

Our proofs are based on the ones provided for our work addressing chronological
\nolinebreak CDCL for model counting \nolinebreak
\cite{DBLP:conf/gcai/MohleB19}, which in turn rely on the proof of
correctness we provided for chronological \nolinebreak CDCL \nolinebreak
\cite{DBLP:conf/sat/MohleB19}. 
The method presented here mainly differs from the former in the following
aspects:
The total models found are shrunken by means of dual reasoning.
It adopts non-chronological \nolinebreak CDCL instead of chronological
\nolinebreak CDCL and accordingly makes use of blocking clauses, which affects
the ordering of the literals on the trail.
In fact, unlike in chronological \nolinebreak CDCL, the literals on the trail
are ordered in ascending order with respect to their decision level, which 
simplifies not only the rules but also the proofs.
Projection in turn adds complexity to some invariants.
In some aspects our proofs are similar to or essentially the same as those in
our former proofs \nolinebreak
\cite{DBLP:conf/sat/MohleB19,DBLP:conf/gcai/MohleB19}. 
However, they are fully worked out to keep them self-contained.

\begin{figure}[!t]
  \setlength{\fboxsep}{0pt}
  \noindent\fbox{
    \parbox{0.979\textwidth}{
      ~\\[\vspacemidrule]
      \renewcommand{\arraystretch}{1.2}
      \setlength\tabcolsep{2pt}
      \centering
      \begin{tabular}{@{}ll@{}}
        \InvDualPN\label{inv:InvDualPN}\\[\vspacebetweeninvs]
        \InvDecs\\[\vspacebetweeninvs]
        \InvImplIIrred\\[\vspacebetweeninvs]
        \InvDSOPIrred\\[\vspacemidrule]
      \end{tabular}
    }
  }
  \caption{
    Invariants for projected model enumeration without repetition.
  }
  \label{fig:invenumirred}
\end{figure}

In order to prove the correctness of our method, we make use of the invariants
listed in \autoref{fig:invenumirred}.
\hyperref[fig:invenumirred]{Invariant \InvDualPNName} in essence is \autoref{eq:pnotn}.
It ensures that
the shrunken model is again a model of \nolinebreak \(P\) projected onto the
input variables, stating
that \nolinebreak \(P\) and \nolinebreak \(N\) projected
onto the input variables \(X \cup Y\) are each other's negation. 
Intuitively, \hyperref[fig:invenumirred]{Invariant \InvDualPNName} holds because the found
models are blocked in \nolinebreak \(P\) and added to its negation \nolinebreak
\(N\).
\hyperref[fig:invenumirred]{Invariants \InvDecsName} and \hyperref[fig:invenumirred]{\InvImplIIrredName} equal
Invariants \nolinebreak (2) and \nolinebreak (3) in our proofs of correctness of
chronological \nolinebreak CDCL \nolinebreak \cite{DBLP:conf/sat/MohleB19} and
model counting by means of chronological \nolinebreak CDCL \nolinebreak
\cite{DBLP:conf/gcai/MohleB19}.
\hyperref[fig:invenumirred]{Invariant \InvImplIIrredName} differs from the latter in that we
need not consider the negation of the DNF \nolinebreak \(M\) explicitly.
The negation of \nolinebreak \(M\) is exactly the conjunction of the blocking 
clauses associated with the found models, and these are added to \nolinebreak
\(P\).
\hyperref[fig:invenumirred]{Invariant \InvImplIIrredName} is needed to show that the literal
propagated after backtracking is implied by the resulting trail.
Its reason is either a blocking clause (rule \BtopIrredName) or a clause learned
by means of conflict analysis (rule \BbotName). 

Our proof is split into several parts.
We start by showing that the invariants listed in \autoref{fig:invenumirred} hold in non-terminal states
(\autoref{sec:calculusenumirredproofsinv}). 
Then we prove that our
procedure always makes progress
(\autoref{sec:calculusenumirredproofsprog}) before showing its
termination (\autoref{sec:calculusenumirredproofsterm}).
We conclude the proof by showing that every total model is found exactly once and that all total
models are detected, \iet, that upon termination \(M \equiv \project{P}{X}\) holds (\autoref{sec:calculusenumirredproofseq}).

\subsubsection{Invariants in Non-Terminal States}
\label{sec:calculusenumirredproofsinv}

\begin{proposition}
  \label{prop:invenumirred}
  \hyperref[fig:invenumirred]{Invariants \InvDualPNName}, \hyperref[fig:invenumirred]{\InvDecsName}, \hyperref[fig:invenumirred]{\InvDSOPIrredName}, and
  \hyperref[fig:invenumirred]{\InvImplIIrredName} hold in non-terminal states. 
\end{proposition}
\begin{proof}
  The proof is carried out by induction over the number of rule applications.
  Assuming \hyperref[fig:invenumirred]{Invariants \InvDualPNName} to \nolinebreak \hyperref[fig:invenumirred]{\InvImplIIrredName}
  hold in a non-terminal state \(\state{P}{N}{M}{I}{\level}\), we show that they
  are met after the transition to another non-terminal state for all rules.

  \skipbetweenrules


  \noindent
  \underline{\UnitName} 
  \vspaceafterrulename
  
  \noindent
  \textit{\hyperref[fig:invenumirred]{Invariant \InvDualPNName}:}  \skipafterinvname
  Neither \nolinebreak \(P\) nor \nolinebreak \(N\) are altered, hence \hyperref[fig:invenumirred]{Invariant
  \InvDualPNName} holds after the application of rule \UnitName.

  \skipbetweeninvs
  
  \noindent
  \textit{\hyperref[fig:invenumirred]{Invariant \InvDecsName}:} \skipafterinvname
  The trail \nolinebreak \(I\) is extended by a literal \nolinebreak \(\ell\).
  We need to show that \nolinebreak \(\ell\) is not a decision literal.
  Only the case where \(a > 0\) need be considered, since at decision level zero
  all literals are propagated.
  There exists a clause \(C \in P\) \stt \(\residual{C}{I} = \{\ell\}\).
  Now, \(a = \level(I)\), \iet, there is already a literal \(k \neq \ell\)
  on \nolinebreak \(I\) with \(\level(k) = a\).
  From this it follows that \nolinebreak \(\ell\) is not a decision literal.
  The decisions remain unchanged, and \hyperref[fig:invenumirred]{Invariant\InvDecsName} holds
  after applying rule \UnitName.

  \skipbetweeninvs

  \noindent
  \textit{\hyperref[fig:invenumirred]{Invariant \InvImplIIrredName}:} \skipafterinvname
  Due to \(\residual{C}{I} = \{\ell\}\), 
  \(P \wedge \decsf{\leq n}(I) \models \negated{(C \setminus \{\ell\})}\).
  Since \nolinebreak \(C \in P\), also
  \(P \wedge \decsf{\leq n}(I) \models C\).
  Modus ponens gives us
  \(P \wedge \decsf{\leq n}(I) \models \filter{I}{\leq n}\).
  Hence,
  \(P \wedge \decsf{\leq n}(I\ell) \models \filter{I\ell}{\leq n}\),
  and \hyperref[fig:invenumirred]{Invariant \InvImplIIrredName} holds after executing rule
  \UnitName.

  \skipbetweeninvs

  \noindent
  \textit{\hyperref[fig:invenumirred]{Invariant \InvDSOPIrredName}:} \skipafterinvname
  Due to the premise, \nolinebreak \(M\) is a \nolinebreak DSOP.
  It is not altered by rule \UnitName and after its application is therefore
  still a \nolinebreak DSOP.
 
  \skipbetweenrules


  \noindent
  \underline{\BtopIrredName} 
  \vspaceafterrulename

  \noindent
  \textit{\hyperref[fig:invenumirred]{Invariant \InvDualPNName}:}  \skipafterinvname
  We have
  \(\exists \,S\,[\,P(X,Y,S) \,]\, \equiv
  \negated{\exists \,T\,[\,N(X,Y,T)\,]\,}\)
  and we need to show 
  \(\exists \,S\,[\,(P \wedge B)(X,Y,S) \,]\, \equiv
  \negated{\exists \,T\,[\,O(X,Y,T)\,]\mbox{,}\,}\)
  where \(B = \negated{\decs(m)}\) and \(O = \tseitin(N \vee \negated{B})\)
  and \(m = \project{\istar}{X}\) is a model of \nolinebreak \(P\) projected
  onto \nolinebreak \(X\).
  Since we have that
  \(\exists \,T\,[\,O(X,Y,T)\,] \equiv \exists \,T\,[\,(N \vee
  \negated{B})(X,Y,T)\,]\)
  and furthermore
  \(\negated{\exists \,T\,[\,(N \vee \negated{B})(X,Y,T)\,]\,} \equiv  
  \forall \,T \,[\,(\negated{N} \wedge B)(X,Y,T)\,]\),
  we reformulate the claim as
  \(\exists \,S\,[\,(P \wedge B)(X,Y,S) \,]\, \equiv
  \forall \,T \,[(\negated{N} \wedge B)(X,Y,T)\,]\).
  Together with
  \(\exists \,S\,[\,P(X,Y,S)\,] \equiv
  \forall \,T\,[\,\negated[N(X,Y,T)\,]\)
  and observing that \nolinebreak \(B\) contains no variable in \nolinebreak
  \(S \cup T\), the claim holds.

  \skipbetweeninvs
  
  \noindent
  \textit{\hyperref[fig:invenumirred]{Invariant \InvDecsName}:} \skipafterinvname
  We show that the decisions remaining on the trail are unaffected and that
  no new decision is taken, \iet, \nolinebreak \(\ell\) in the post state is not
  a decision. 
  It is sufficient to consider the case where \(\level(I) > 0\).
  Now, \(J = \filter{I}{\leq b}\) by the definition of \nolinebreak
  \(J\), and the decisions on \nolinebreak \(J\) are not affected by rule
  \nolinebreak \BtopIrredName.
  We have \(\level(B \setminus \{\ell\}) = b = \level(J)\) and \(\level(B) =
  b+1\).
  Since relevant decisions are prioritized, also
  \(B = \negated{\decsf{\leq b+1}(\project{I}{X})} = \negated{\decsf{\leq b+1}(I)}\).
  By the induction hypothesis, 
  there exists exactly one
  decision literal for each decision level and in particular in \nolinebreak \(B\). 
  Since \(\ell \in B\), we have \(\negated{\ell} \in \decs(I)\).
  Precisely, \(\negated{\ell} \in K\), and \(\negated{\ell}\) is unassigned upon
  backtracking.
  Due to the definition of \nolinebreak \(B\), there exists a literal \(k \in B\)
  where \(k \neq \ell\) such that \(\level(k) = b\), \iet, \(k \in J\), hence
  \nolinebreak \(k\) precedes \nolinebreak \(\ell\) on the resulting trail.
  By the definition of the blocks on the trail, \nolinebreak \(\ell\) is not a
  decision literal.
  Since the decisions on \nolinebreak \(J\) are unaffected, as
  argued above, \hyperref[fig:invenumirred]{Invariant \InvDecsName} is met.

  \skipbetweeninvs

  \noindent
  \textit{\hyperref[fig:invenumirred]{Invariant \InvImplIIrredName}:} \skipafterinvname
  We need to show that
  \(P \wedge \decsf{\leq n}(J\,\ell) \models \filter{(J\,\ell)}{\leq n}\)
  for all \nolinebreak \(n\).
  First notice that the decision levels of the literals in \nolinebreak \(J\) do
  not change by applying rule \nolinebreak \BtopIrredName.
  Only the decision level of the variable of \nolinebreak \(\ell\) is decremented from
  \nolinebreak \(b+1\) to \nolinebreak \(b\).
  It also stops being a decision.
  Since \(\level(J\,\ell) = b\), we can assume \(n \leq b\).
  Observe that
  \(P \wedge \decsf{\leq n}(J\,\ell) \equiv P \wedge \decsf{\leq n}(J)\),
  since \nolinebreak \(\ell\) is not a decision in \nolinebreak
  \(J\,\ell\) and \(\filter{I}{\leq b} = J\) and thus
  \(\filter{I}{\leq n} = \filter{J}{\leq n}\) by definition.
  Now the induction hypothesis is applied and we get
  \(P \wedge \decsf{\leq n}(J\,\ell) \models \filter{I}{\leq n}\).
  Again using \(\filter{I}{\leq n} = \filter{J}{\leq n}\) this almost closes the
  proof except that we are left to prove
  \(P \wedge \decsf{\leq b}(J\,\ell) \models \ell\) as \nolinebreak \(\ell\) has
  decision level \nolinebreak \(b\) in \nolinebreak \(J\,\ell\) after applying the
  rule and thus \nolinebreak \(\ell\) disappears in the proof obligation for
  \nolinebreak \(n < b\).
  To see this notice that \(P \wedge \negated{B} \models \filter{I}{\leq b+1}\)
  using again the 
  induction hypothesis for \nolinebreak \(n = b+1\), and recalling that relevant 
  decisions are prioritized, \iet, \nolinebreak \(\filter{I}{\leq b+1}\)
  contains only relevant decisions, and
  \(\negated{B} = \decs(\project{\istar}{X}) = \decsf{\leq b+1}(I)\).
  This gives 
  \(P \wedge \negated{\decsf{\leq b}(J)} \wedge \negated{\ell} \models
    \filter{I}{\leq b+1}\) and thus
  \(P \wedge \negated{\decsf{\leq b}(J)} \wedge \negated{\filter{I}{\leq b+1}}
  \models \ell\) by conditional contraposition.
  Therefore, \hyperref[fig:invenumirred]{Invariant \InvImplIIrredName} holds. 
  
  \skipbetweeninvs
  
  \noindent
  \textit{\hyperref[fig:invenumirred]{Invariant \InvDSOPIrredName}:} \skipafterinvname
  We assume that \nolinebreak \(M\) is a \nolinebreak DSOP and need to show that
  \(M \vee m\) is also a \nolinebreak DSOP.
  Due to the use of the dual blocking clause encoding,
  \autoref{prop:dualblockirredundant} holds, and \hyperref[fig:invenumirred]{Invariant \InvDSOPIrredName}
  is met after executing \BtopIrredName.
  
  \skipbetweenrules


  \noindent
  \underline{\BbotName}
  \vspaceafterrulename

  \noindent
  \textit{\hyperref[fig:invenumirred]{Invariant \InvDualPNName}:}  \skipafterinvname
  We have
  \(\exists \,S\,[\,P(X,Y,S) \,]\, \equiv
  \negated{\exists \,T\,[\,N(X,Y,T)\,]\,}\),
  and we need to show that
  \(\exists \,S\,[\,(P \wedge D)(X,Y,S) \,]\, \equiv
  \negated{\exists \,T\,[\,N(X,Y,T)\,]\,}\).
  By the premise, \(P \models D\), hence \(P \wedge D \equiv P\).
  Now 
  \(\exists \,S\,[\,(P \wedge D)(X,Y,S) \,]\, \equiv
    \exists \,S\,[\,P(X,Y,S) \,]\, \equiv
    \negated{\exists \,T\,[\,N(X,Y,T)\,]\,}\),
  and \hyperref[fig:invenumirred]{Invariant \InvDualPNName} holds.

  \skipbetweeninvs

  \noindent
  \textit{\hyperref[fig:invenumirred]{Invariant \InvDecsName}:} \skipafterinvname
  We have \(J \leq I\), hence the decisions on \nolinebreak \(J\) remain
  unaltered.
  Now we show that \nolinebreak \(\ell\) is not a decision literal.
  As in the proof for rule \nolinebreak \UnitName, it is sufficient to consider
  the case where \nolinebreak \(j > 0\).
  There exists a clause \nolinebreak \(D\) where \nolinebreak \(P \models D\)
  such that \(\level(D) > 0\) and a literal \nolinebreak \(\ell \in D\) for which
  \nolinebreak \(\residual{\ell}{K} = \false\) and \nolinebreak \(\negated{\ell}
  \in K\), hence \nolinebreak \(\ell\) is unassigned during backtracking.
  Furthermore, there exists a literal \nolinebreak \(k \in D\) where \(k \neq
  \ell\) and such that \(\level(k) = j\) which precedes \nolinebreak \(\ell\) on
  the trail \nolinebreak \(J\,\ell\).
  Therefore, following the argument in rule \nolinebreak \UnitName, the literal
  \(\ell\) is not a decision literal.
  Since the decisions remain unchanged, \hyperref[fig:invenumirred]{Invariant \InvDecsName}
  holds after applying rule \nolinebreak \BbotName.

  \skipbetweeninvs
  
  \noindent
  \textit{\hyperref[fig:invenumirred]{Invariant \InvImplIIrredName}:} \skipafterinvname
  Let \nolinebreak \(n\) be arbitrary but fixed.
  Before executing \nolinebreak \BbotName, we have
  \(P \wedge \decsf{\leq n}(I) \models \filter{I}{\leq n}\).
  We need to show that
  \(P \wedge \decsf{\leq n}(J\,\ell) \models \filter{(J\,\ell)}{\leq n}\).
  Now, \(I = J\,K\) and \(J < I\), \iet,
  \(P \wedge \decsf{\leq n}(J) \models \filter{J}{\leq n}\).
  From \(j = \level(D \setminus \{\ell\}) = \level(J)\) we get
  \(\residual{D}{J} = \{\ell\}\).
  On the one hand,
  \(P \wedge \decsf{\leq n}(J) \models \negated{(D \setminus \{\ell\})}\),
  and on the other hand
  \(P \wedge \decsf{\leq n}(J) \models D\).
  Therefore, by modus ponens,
  \(P \wedge \decsf{\leq n}(J) \models \ell\).
  Since \nolinebreak \(\ell\) is not a decision literal, as shown above,
  \(P \wedge \decsf{\leq n}(J) \equiv P \wedge \decsf{\leq n}(J\,\ell)\) and
  \(P \wedge \decsf{\leq n}(J\,\ell) \models J\,\ell\), and  \hyperref[fig:invenumirred]{Invariant 
  \InvImplIIrredName} holds after applying rule \nolinebreak \BbotName.

  \skipbetweeninvs
  
  \noindent
  \textit{\hyperref[fig:invenumirred]{Invariant \InvDSOPIrredName}:} \skipafterinvname
  The \nolinebreak DSOP \(M\) remains unaltered, and
  \hyperref[fig:invenumirred]{Invariant \InvDSOPIrredName} still
  holds after executing rule \BbotName.

  \skipbetweenrules


  \noindent
  \underline{\DecXIrredName}
  \vspaceafterrulename

  \noindent
  \textit{\hyperref[fig:invenumirred]{Invariant \InvDualPNName}:}  \skipafterinvname
  Both \nolinebreak \(P\) and \nolinebreak \(N\) remain unaltered, hence
  \hyperref[fig:invenumirred]{Invariant \nolinebreak \InvDualPNName} still holds after executing rule
  \nolinebreak \DecXIrredName.

  \skipbetweeninvs

  \noindent
  \textit{\hyperref[fig:invenumirred]{Invariant \InvDecsName}:} \skipafterinvname
  The literal \nolinebreak \(\ell\) is a decision literal by definition.
  It is assigned decision level \(d = \level(I) + 1\).
  Since \(\ell \in \decs(I\,\ell)\), we have
  \(\level(\decs(I\,\ell)) = \{1, \ldots, d\}\), and \hyperref[fig:invenumirred]{Invariant
  \InvDecsName} holds after applying rule \nolinebreak \DecXIrredName.

  \skipbetweeninvs
  
  \noindent
  \textit{\hyperref[fig:invenumirred]{Invariant \InvImplIIrredName}:} \skipafterinvname
  Le \nolinebreak \(n\) be arbitrary but fixed.
  Since \nolinebreak \(\ell\) is a decision literal, we have
  \(P \wedge \decsf{\leq n}(I\,\ell) \equiv
  P \wedge \decsf{\leq n}(I) \wedge \ell \models
  \filter{I}{\leq n} \wedge \ell \equiv \filter{(I\,\ell)}{\leq n}\).
  Hence, \hyperref[fig:invenumirred]{Invariant \InvImplIIrredName} holds after applying rule
  \nolinebreak \DecXIrredName.

  \skipbetweeninvs
  
  \noindent
  \textit{\hyperref[fig:invenumirred]{Invariant \InvDSOPIrredName}:} \skipafterinvname
  The \nolinebreak DSOP \nolinebreak \(M\) remains unaltered by rule
  \DecXIrredName, hence after applying rule \nolinebreak \DecXIrredName
  \hyperref[fig:invenumirred]{Invariant \InvDSOPIrredName} still holds . 

  \skipbetweenrules


  \noindent
  \underline{\DecYSIrredName.}
  \vspaceafterrulename

  The proofs of \hyperref[fig:invenumirred]{Invariants \InvDualPNName}, \hyperref[fig:invenumirred]{\InvDecsName}, \hyperref[fig:invenumirred]{\InvDSOPIrredName}, and
  \hyperref[fig:invenumirred]{\InvDSOPIrredName} are identical to 
  the ones for rule \nolinebreak \DecXIrredName.
\end{proof}

\subsubsection{Progress}
\label{sec:calculusenumirredproofsprog}

Our method can not get caught in an endless loop, as shown next.

\begin{proposition}
  \label{prop:irredundantprogress}
  \mainalgirredname always makes progress, \iet, in every non-terminal state a rule
  is applicable. 
\end{proposition}
\begin{proof}
  The proof is executed by induction over the number of rule applications.
  We show that in any non-terminal state \nolinebreak
  \(\state{P}{N}{M}{I}{\level}\) a rule is applicable.

  Assume all variables are assigned and no conflict has occurred.
  If no relevant decision is left on the trail \nolinebreak \(I\), rule
  \nolinebreak \EtopName can be applied.
  Otherwise, we execute an incremental \nolinebreak SAT call \SAT{\((N,
  \project{I}{X \cup Y})\)}.
  Since all input variables are assigned, we obtain a conflict by propagating
  internal variables only.
  Conflict analysis gives us the subsequence \nolinebreak \(\istar\) of \nolinebreak \(\project{I}{X
    \cup Y}\) consisting of
  the literals involved in the conflict, which is a
  model of \nolinebreak \(F\).
  Since we are interested in the models of \nolinebreak \(F\) projected onto \nolinebreak
  \(X\), 
  we choose \(B = \negated{\decs(\project{\istar}{X})}\).
  Now, \(\level(B) = b+1\), and due to \hyperref[fig:invenumirred]{Invariant \InvDecsName},
  \nolinebreak \(B\) contains exactly one decision literal \nolinebreak \(\ell\)
  such that \(\level(\ell) = b+1\) and therefore
  \(\level(B \setminus \{\ell\}) = b\).
  We choose \nolinebreak \(J\) and \nolinebreak \(K\) such that \(I = J\,K\) and
  \(b = \level(J)\) and in particular \(\residual{\ell}{K} = \false\).
  After backtracking to decision level \nolinebreak \(b\), we have
  \(\filter{I}{\leq b} = J\) where \(\residual{B}{J} = \{\ell\}\).
  All preconditions of rule \nolinebreak \BtopIrredName are met.

  If instead a conflict has occurred, there exists a clause \(C \in P\) such
  that \(\residual{C}{I} = \false\).
  If \(\level(C) = 0\), rule \EbotName is applicable.
  Otherwise, by  \hyperref[fig:invenumirred]{Invariant \InvImplIIrredName} we have
  \(P \wedge \decsf{\leq \level(I)}(I) \equiv
  P \wedge \decsf{\leq \level(I)}(I) \wedge \filter{I}{\leq \level(I)}
  \models \filter{I}{\leq \level(I)}\).
  Since \(I(P) \equiv \false\), also 
  \(P \wedge \decsf{\leq \level(I)}(I) \wedge \filter{I}{\leq \level(I)} \equiv
  P \wedge \decsf{\leq \level(I)}(I) \equiv \false\).
  If we choose \(D = \negated{\decs(I)}\) we obtain
  \(P \wedge \negated{D} \wedge \filter{I}{\leq \level(I)} \equiv \false\), thus
  \(P \models D\).
  Clause \nolinebreak \(D\) contains only decision literals and \(\level(D) =
  \level(I)\).
  From \hyperref[fig:invenumirred]{Invariant \InvDecsName} we know that \nolinebreak \(D\)
  contains exactly one decision literal for each decision level in
  \(\{1, \ldots, \level(I)\}\).
  We choose \(\ell \in D\) such that \(\level(\ell) = \level(I)\).
  Then the asserting level is given by \(j = \level(D \setminus \{\ell\})\).
  Without loss of generalization we assume the trail to be of the form
  \(I = J\,K\) where \(\level(J) = j\).
  After backtracking to decision level \nolinebreak \(j\), the trail is equal to
  \nolinebreak \(J\).
  Since \(\residual{D}{J} = \{\ell\}\), all conditions of rule \nolinebreak
  \BbotName hold.
  
  If \(\residual{P}{I} \not\in \{\false, \true\}\), there are unassigned
  variables in \(X \cup Y \cup S\).
  If there exists a clause \(C \in P\) where \(\residual{C}{I} = \{\ell\}\), the
  preconditions of rule \nolinebreak \UnitName are met.
  If instead \(\units(\residual{F}{I}) = \emptyset\), there exists a literal
  \nolinebreak \(\ell\) with \(\var(\ell) \in X \cup Y \cup S\) and
  \(\level(\ell) = \infty\).
  If not all relevant variables are assigned, the preconditions of rule
  \nolinebreak \DecXIrredName are satisfied.
  Otherwise, rule \nolinebreak \DecYSIrredName is applicable.

  All possible cases are covered by this argument.
  Hence, in every non-terminal
  state a rule is applicable, \iet, \mainalgirredname always makes progress.  
\end{proof}

\subsubsection{Termination}
\label{sec:calculusenumirredproofsterm}

\begin{proposition}
  \label{prop:irredundanttermination}
  \mainalgirredname terminates.
\end{proposition}
\begin{proof}
  In our proof we follow the argument by Nieuwenhuis et al. \nolinebreak
  \cite{DBLP:journals/jacm/NieuwenhuisOT06} and Mari{\'c} and Jani\v{c}i{\'c}
  \nolinebreak \cite{DBLP:journals/corr/abs-1108-4368}, or more precisely the
  one by Blanchette et al. \nolinebreak
  \cite{DBLP:journals/jar/BlanchetteFLW18}.  

  We need to show that from the initial state \nolinebreak
  \(\state{P}{N}{\false}{\emptytrail}{\level_0}\) a final state \nolinebreak
  \(M\) is reached in a finite number of steps, \iet, no infinite sequence of
  rule applications is generated.
  Otherwise stated, we need to prove that the relation \nolinebreak
  \(\transenumirred\) is well-founded.
  To this end, we define a well-founded relation \nolinebreak
  \(\relationenumip\) such that any transition \(s \transenumirred s'\) from a
  state \nolinebreak \(s\) to a state \nolinebreak \(s'\) implies \(s
  \relationenumip s'\).
  
  In accordance with Blanchette et al. \nolinebreak
  \cite{DBLP:journals/jar/BlanchetteFLW18} but adopting the notation introduced
  by Fleury \nolinebreak \cite{FleuryMaster2015}, we map states to lists.
  Using the abstract representation of the assignment trail \nolinebreak \(I\)
  by Nieuwenhuis et al. \nolinebreak \cite{DBLP:journals/jacm/NieuwenhuisOT06}, we
  write
  \begin{equation}
    I = I_0\,\ell_1\,I_1\,\ell_2\,I_2\,\ldots\,\ell_m\,I_m \quad\text{where}\quad
    \{\ell_1, \ldots, \ell_m\} = \decs(I). \label{eq:abstractI}
  \end{equation}
  The state \nolinebreak \(\state{P}{N}{M}{I}{\level}\) is then mapped to
  \begin{equation}
    [\,
    \underbrace{\encprop, \ldots, \encprop}_{\length{I_0}},
    \encdec, \underbrace{\encprop, \ldots, \encprop}_{\length{I_1}},
    \encdec, \underbrace{\encprop, \ldots, \encprop}_{\length{I_2}},
    \encdec,
    \ldots,
    \encdec, \underbrace{\encprop, \ldots, \encprop}_{\length{I_m}},
    \underbrace{\encunass, \ldots, \encunass}_{\length{V} - \length{I}}
    \,] \label{eq:abstractImapping}
  \end{equation}
  where \(V = X \cup Y \cup S\).
  In this representation, the order of the literals on \nolinebreak \(I\) is
  reflected.
  Propagated literals are denoted by \nolinebreak \(\encprop\), decisions are denoted
  by \nolinebreak \(\encdec\).
  Unassigned variables are represented by \nolinebreak \(\encunass\) and are moved to the
  end.
  The final state \nolinebreak \(M\) is represented by \nolinebreak
  \(\emptylist\).
  The state containing the trail \nolinebreak \(I\) in \autoref{eq:abstractI} is
  mapped to the list in \autoref{eq:abstractImapping}.
  The first \nolinebreak \(\length{I_0}\) entries represent the literals
  propagated at decision level \nolinebreak zero, the \nolinebreak \(\encdec\) at
  position \(\length{I_0}+1\) represents the decision literal \nolinebreak
  \(\ell_1\), and so on for all decision levels on \nolinebreak \(I\).
  The last \(\length{V} - \length{I}\) entries denote the unassigned variables.
  Notice that we are not interested in the variable assignment itself but in
  its structure, \iet, the number of propagated literals per decision
  level and the number of unassigned variables.
  Furthermore, the states are encoded into lists of the same length.
  This representation induces a lexicographic order \nolinebreak \(\lexlower\)
  on the states.
  We therefore define \nolinebreak \(\relationenumip\) as the restriction of \nolinebreak
  \(\lexlower\) to \(\{[\,v_1, \ldots, v_{\length{V}}] \mid v_i \in \{\encprop,
  \encdec, \encunass\}
  ~\text{for}~ 1 \leq i \leq \length{V}\}\).
  Accordingly, we have that \(s \relationenumip s'\), if
  \(s \lexlower s'\).  

  \begin{figure}[!t]
    \noindent\fbox{
      \parbox{0.96\textwidth}{
        \centering
        \begin{tabular}{ccc}
          \AxiomC{\([\, \ldots\,]\)}
          \RightLabel{\EtopName}
          \UnaryInfC{\(\varepsilon\)}          
          \DisplayProof{}
          &
            \AxiomC{\([\,\encprop, \ldots, \encprop, \encunass, \ldots, \encunass\,]\)}
            \RightLabel{\EbotName}
            \UnaryInfC{\(\varepsilon\)}          
            \DisplayProof{}
          &
            \AxiomC{\([\,\ldots, \encunass, \encunass, \ldots, \encunass\,]\)}
            \RightLabel{\UnitName}
            \UnaryInfC{\([\,\ldots, \encprop, \encunass, \ldots, \encunass\,]\)}
            \DisplayProof{}
        \end{tabular}

        \skipbetweenproofs
        
        \begin{tabular}{cc}
          \AxiomC{\([\,\ldots, \encdec, \ldots, \encdec, \ldots, \encunass, \ldots, \encunass\,]\)}
          \RightLabel{\BtopIrredName}
          \UnaryInfC{\([\,\ldots, \encprop, \encunass, \,\ldots\ldots\ldots\ldots\,, \encunass\,]\)}
          \DisplayProof{}
          &
            \AxiomC{\([\,\ldots, \encdec, \ldots, \encdec, \ldots, \encunass, \ldots, \encunass\,]\)}
            \RightLabel{\BbotName}
            \UnaryInfC{\([\,\ldots, \encprop, \encunass, \,\ldots\ldots\ldots\ldots\,, \encunass\,]\)}
            \DisplayProof{}
        \end{tabular}
        
        \skipbetweenproofs

        \begin{tabular}{ccc}
          \AxiomC{\([\,\ldots, \encunass, \encunass, \ldots, \encunass\,]\)}
          \RightLabel{\DecXIrredName}
          \UnaryInfC{\([\,\ldots, \encdec, \encunass, \ldots, \encunass\,]\)}
          \DisplayProof{}
          &
            \AxiomC{\([\,\ldots, \encunass, \encunass, \ldots, \encunass\,]\)}
            \RightLabel{\DecYSIrredName}
            \UnaryInfC{\([\,\ldots, \encdec, \encunass, \ldots, \encunass\,]\)}
            \DisplayProof{}
        \end{tabular}
      }
    }
    \caption{Transitions of states mapped to lists according to
      \autoref{eq:abstractImapping}.
      The initial state is depicted above the horizontal rule, the resulting
      state below.
      The two end rules lead to the minimal element \nolinebreak
      \(\varepsilon\) representing the final state.
      Rule \nolinebreak \UnitName replaces an unassigned variable (denoted by
      \nolinebreak \(\encunass\)) by a propagated one (denoted by \nolinebreak \(\encprop\))
      and leaves the rest unchanged. 
      Rules \nolinebreak \BtopIrredName and \nolinebreak \BbotName replace a
      decision literal (denoted by \nolinebreak \(\encdec\)) by a propagated one.
      Finally, the two decision rules replace an unassigned literal by a
      decision.
      Clearly, \wrtt the lexicographic order, the states decrease by a
      rule application.
  }
    \label{fig:transenumirred}
  \end{figure}

In \autoref{fig:transenumirred}, the state transitions for the rules are
visualized.
In this representation, the unspecified elements occurring prior to the first
digit are not altered by the application of the rule.
We show that \(s \relationenumip s'\) for each rule, where \nolinebreak \(s\)
and \nolinebreak \(s'\) encode the state before and after applying the
corresponding rule, respectively, and end states are encoded by \(\varepsilon\).

\skipbetweenrules

\noindent
\EtopName. \skipafterrulename
The state \nolinebreak \(s'\) is mapped to \nolinebreak \(\emptylist\) which is the minimal element
with respect to \nolinebreak \(\lexlower\), hence \(s \relationenumip s'\)
trivially holds.
The representation of the state may contain both \nolinebreak \(\encprop\)'s and
\nolinebreak \(\encdec\)'s but no \nolinebreak \(\encunass\)'s, since our
algorithm detects only total
models.\footnote{This restriction may
  be weakened in favor of finding partial models.
  In this section, we refer to the rules introduced in
  \autoref{sec:calculusenumirred} and discuss a generalization of our algorithm
  enabling the detection of partial models further down.}
Recall that the associated trail must not contain any relevant decision, which
is not reflected in the structure of the trail.

\skipbetweenrules

\noindent
\EbotName. \skipafterrulename
The state \nolinebreak \(s'\) is mapped to \nolinebreak \(\emptylist\), which is the minimal element
with respect to \nolinebreak \(\lexlower\), hence \(s \relationenumip s'\)
trivially holds.
The representation of the state may contain both \nolinebreak \(\encprop\)'s and
\nolinebreak \(\encunass\)'s but no \nolinebreak \(\encdec\)'s, since any
decision need be flipped.

\skipbetweenrules

\noindent
\UnitName. \skipafterrulename
An unassigned variable is propagated.
Its representation changes from \nolinebreak \(\encunass\) to \nolinebreak \(\encprop\), and
all elements preceding it remain unaffected.
Due to \(\encunass \lexlower \encprop\), we also have that \(s \relationenumip s'\).

\skipbetweenrules

\noindent
\BtopIrredName~/~\BbotName. \skipafterrulename
A decision literal, \egt, \nolinebreak \(\negated{\ell}\), is flipped and
propagated at a lower decision level, let's say \nolinebreak \(d\).
The decision level \nolinebreak \(d\) is extended by \nolinebreak \(\ell\),
which is represented by \nolinebreak \(\encprop\) and replaces the decision literal at
decision level \nolinebreak \(d+1\).
All variables at decision level \nolinebreak \(d+1\) and higher are
unassigned and thus represented by \nolinebreak \(\encunass\).
Therefore, \(s \relationenumip s'\).
Notice that, although different preconditions of the rules \nolinebreak
\BtopIrredName and \nolinebreak \BbotName apply and the two rules differ, the
structure of their states is the same.

\skipbetweenrules

\noindent
\DecXIrredName~/~\DecYSIrredName. \skipafterrulename
An unassigned variable is decided, \iet, the first
occurrence of \nolinebreak
\(\encunass\) is replaced by \nolinebreak \(\encdec\) in the representation.
The other elements remain unaltered, hence \(s \relationenumip s'\).
As for the backtracking rules, whether a relevant or irrelevant or internal
variable is decided, is irrelevant and not reflected in the mapping of the
state, as for rules \BtopIrredName and \BbotName.

\skipbetweenrules

We have shown that after any rule application the resulting state is smaller
than the preceding one with respect to the lexicographic order on which
\nolinebreak \(\relationenumip\) is based.
This argument shows that \nolinebreak \(\relationenumip\) is well-founded and
that therefore \mainalgirredname terminates.
\end{proof}

\subsubsection{Equivalence}
\label{sec:calculusenumirredproofseq}

The final state is given by a \nolinebreak DSOP \(M\) such that \(M \equiv
\project{F}{X}\).
The proof is split into several steps.
We start by proving that, given a total model \nolinebreak \(I\) of \nolinebreak
\(P\), its subsequence \nolinebreak \(\istar\) returned by \nolinebreak \SATtwo
(line 20 of \mainalgirredname in
\autoref{fig:enumirred}) is a (partial) model of \nolinebreak \(\project{P}{X}\) and that any
total model of \nolinebreak \(P\) found during execution 
either was already found or is found for the first time.
Then we show that all models of \nolinebreak \(P\) are found and that each model is
found exactly once, before
concluding by proving that \(M \equiv \project{F}{X}\).

\begin{proposition}
  \label{prop:istarmodel}
  Let \nolinebreak \(I\) be a total model of \nolinebreak \(P\) and \nolinebreak
  \(\istar = \SAT(N, \project{I}{X \cup Y})\).
  Then \nolinebreak \(\istar\) is a model of \nolinebreak \(\project{P}{X}\). 
\end{proposition}
\begin{proof}
  All variables in \(X \cup Y \cup S\) are assigned and \nolinebreak
  \(P(X,Y,S)\) and \nolinebreak 
  \(N(X,Y,T)\) are a dual representation of \nolinebreak \(F(X,Y)\).
  \hyperref[fig:invenumirred]{Invariant \InvDualPNName} holds.
  In particular it holds for the values of the variables in
  \nolinebreak \(X \cup Y \cup S\) set to their values in \nolinebreak \(I\),
  \iet, we have that 
  \(\exists \,S\,[\,\residual{P(X,Y,S)}{I} \,]\, \equiv
  \negated{\exists \,T\,[\,\residual{N(X,Y,T)}{I}\,]\,}\)
  where only the unassigned variables in \(\subvars{(X \cup Y)}{I}\) are universally quantified.
  Since \nolinebreak \(I\) is a total model of \nolinebreak \(P\), \hyperref[fig:invenumirred]{Invariant
  \InvDualPNName} can be rewritten as 
  \(\residual{P(X,Y,S)}{I} \equiv \negated{\exists
    \,T\,[\,\residual{N(X,Y,T)}{I}}\),
  and \nolinebreak \(\project{I}{X \cup Y}\) can not be extended to a model of
  \nolinebreak \(N\).
  Since the variables in \nolinebreak \(T\) are defined in terms of variables
  in \nolinebreak \(X \cup Y\),
  an incremental \nolinebreak SAT call on \nolinebreak \(N \wedge I\)
  yields a conflict in \nolinebreak \(N\) exclusively by propagating variables
  in \nolinebreak \(T\).

  Exhaustive conflict analysis yields a clause \nolinebreak \(D\) consisting of
  the negations of the (assumed) literals in \nolinebreak \(I\) involved in the
  conflict.  
  Its negation is a counter-model of \nolinebreak \(\project{N}{X \cup Y}\)
  which, due to \autoref{eq:ntop}, is a model of \nolinebreak
  \(\project{P}{X \cup Y}\).
  Obviously, the same holds for the projection onto \nolinebreak \(X\), and
  \(\project{\negated{D}}{X} \models \project{P}{X}\).
  Since \(\istar = \project{\negated{D}}{X}\), we have 
  \(\istar \models \project{P}{X}\), and the claim holds.
\end{proof}

\begin{proposition}
  \label{prop:irredtotalmodel}
  A total model \nolinebreak \(I\) of \nolinebreak \(P\) is either 
  \begin{enumerate}[label=(\roman*)]
  \item contained in \nolinebreak \(M\) or \label{it:Icontained}
  \item subsumed by a model in \nolinebreak \(M\) or \label{it:Isubsumed}
  \item a model of \nolinebreak \(P_0 \wedge \bigwedge\limits_i B_i\) where
    \nolinebreak \(B_i\) are the blocking clauses added to \nolinebreak
    \(P_0\) \label{it:Imodel} 
  \end{enumerate}
\end{proposition}
\begin{proof}
  All variables in \(X \cup Y \cup S\) are assigned and \(I \models P\).
  If \nolinebreak \(I\) was already found earlier, it was shrunken and the
  resulting model projected onto \nolinebreak \(X\) to obtain \nolinebreak
  \(\istar\) which was then added to \nolinebreak \(M\) (rule \BtopIrredName
  and line 24 in \mainalgirredname in \autoref{fig:enumirred}). 
  If all assumed variables participated in the conflict and furthermore \(Y = S
  = \emptyset\), then \(\project{\istar}{X \cup Y} = \project{\istar}{X} = I\), and \autoref{it:Icontained} holds.
  Otherwise, \(\istar < I\) and \(\istar\) subsumes \nolinebreak \(I\).
  Since \(\istar \in M\), in this case \autoref{it:Isubsumed} holds.

  Suppose the model \nolinebreak \(I\) is found for the first time.
  Since \(I \models P\), also \(I \models C\) for all clauses \(C \in P\).
  This in particular holds for all blocking clauses which were added to the
  original formula \(P = P_0\) (rule \BtopIrredName and line 22 in
  \mainalgirredname), and \autoref{it:Imodel} holds.
\end{proof}

\begin{proposition}
  \label{prop:irredallfound}
  Every model is found.
\end{proposition}
\begin{proof}
  According to \autoref{prop:irredundanttermination}, \mainalgirredname
  terminates.
  The final state \nolinebreak \(M\) is reached
  by either rule \EbotName or rule \EtopName.
  Assume \(P = P_0 \bigwedge_i B_i\) with \nolinebreak \(B_i\) denoting the
  blocking clauses added to the original formula \nolinebreak \(P_0\), and let 
  \nolinebreak \(I\) denote the current trail.
  If rule \EbotName is applied, then \nolinebreak \(P\) is unsatisfiable, since
  \nolinebreak \(\residual{P}{I} = \false\) and \nolinebreak \(I\) contains 
  no decision.
  If rule \EtopName is applied, then by \autoref{prop:irredtotalmodel},
  \autoref{it:Imodel}, \(m = \project{I}{X}\) is a model of \nolinebreak \(P\).
  Now, \(\residual{(P \wedge \negated{m})}{I} = \false\) and \nolinebreak \(I\)
  contains no decision literal.
  Therefore, \(\residual{(P \wedge \negated{m})}{I}\) is unsatisfiable, and all
  models have been found. 
  \end{proof}

\begin{proposition}
  \label{prop:irredmodels}
  Every model is found exactly once.
\end{proposition}

\begin{proof}
  We recall \autoref{prop:dualblockirredundant} stating that only pairwise
  contradicting models are detected. 
  In essence, this says that every model is found exactly once.
\end{proof}

\begin{theorem}[Correctness]
  \label{th:correctirred}
  If \(\state{P}{N}{\false}{\emptytrail}{\level_0} \transenumirred^* M\), then
  \begin{enumerate}[label=(\roman*)]
  \item \(M \equiv \project{F}{X}\) \label{it:meqf}
  \item \(C_i \wedge C_j \equiv \false\) for \(C_i, C_j \in M\) and \(C_i \neq
    C_j\) \label{it:mdsop}
  \end{enumerate}
\end{theorem}

\begin{proof}
  The cubes in \nolinebreak \(M\) are exactly the \nolinebreak \(\istar\)
  computed from the total models of \nolinebreak \(P\).
  These are models of \nolinebreak \(\project{P}{X}\)
  (\autoref{prop:istarmodel}).
  Since by \autoref{prop:irredallfound} all models are found, \nolinebreak \(M
  \equiv \project{P}{X}\).
  But \(\mods(\exists \,Y,S \stf P(X,Y,S)) = \mods(\exists \,Y \stf F(X,Y))\)  by \autoref{eq:modelspf},
  \iet \(\mods{(\project{P}{X})} = \mods{(\project{F}{X})}\).
  Therefore, \(M \equiv \project{F}{X}\), and \autoref{it:meqf} holds.

  Due to \autoref{prop:dualblockirredundant}, the found models are pairwise
  contradicting, and \autoref{it:mdsop} holds as well.
  Notice that one could also use \autoref{prop:irredmodels}, since, as its
  consequence, only pairwise contradicting models are found.
\end{proof}

\subsection{Generalization to Partial Model Detection}
\label{sec:calculusenumirredpm}

\mainalgirredname only finds total models of \nolinebreak \(P\).
In \nolinebreak SAT solving, this makes sense from an computational point of
view, because checking whether a partial assignment satisfies a formula is more
expensive than extending it to a total one.
However, model enumeration is computationally more expensive than \nolinebreak
SAT solving, hence satisfiability checks, \egt, in the form of entailment checks
\nolinebreak \cite{DBLP:conf/sat/MohleSB20}, might pay off. 
Notice that it still might make sense to shrink the models found.
In this section, we discuss the changes to be made to our approach in order
to support the detection of partial models. 

First, the satisfiability condition need be changed such that it
complies with any other strategy determining whether the (partial) assignment
\nolinebreak \(I\) satisfies \nolinebreak \(P\).
The check now reads ``\(I(P) \equiv \true\)'' and replaces the one on line
\nolinebreak 15 of \mainalgirredname (\autoref{fig:enumirred}).
The rest of the algorithm remains unaltered.
In our calculus (\autoref{fig:calculusenumirred}), the precondition
``\(I(P)\equiv\true\)'' replaces the check whether all variables in \(X \cup Y
\cup S\) are assigned in rules \nolinebreak \EtopName
and \nolinebreak \BtopIrredName.
The other preconditions as well as rules \nolinebreak \EbotName, \UnitName,
\BbotName, \DecXIrredName, and \DecYSIrredName remain unaltered.

Second, the computation of \nolinebreak \(\istar\) need be adapted.
It is based on the assumption that \nolinebreak \(I\) is total such that a
conflict in \nolinebreak \(\residual{N}{I}\) is obtained by propagating only
variables in \nolinebreak \(T\).  
Now \hyperref[fig:invenumirred]{Invariant \InvDualPNName} ensures that a
conflict in \nolinebreak \(\residual{N}{I}\) is obtained also if \nolinebreak
\(I\) is a partial assignment, although in order to obtain this conflict
variables in \(X \cup Y\) might need be propagated or decided.
Projecting the so-obtained assignment \nolinebreak \(I'\) onto \nolinebreak
\(I\) solves the issue.
Hence, we replace line 20 in \mainalgirredname (\autoref{fig:enumirred}) by
``\(I' \mdef \functionname{SAT}(N,\project{I}{X\cup Y});
\istar = \project{I'}{\var(I)}\)''.
These changes need also be reflected in rule \nolinebreak \BtopIrredName and in
the proof of \autoref{prop:irredundantprogress}.

\section{Conflict-Driven Clause Learning for Redundant All-SAT}
\label{sec:learn}

Conflict analysis is based on the assumption that the reason of every propagated
literal is contained in the formula.  
In irredundant model enumeration (see \autoref{sec:pmeirred} and
\autoref{sec:formalenumirred}), this reason is either a clause learned by 
means of conflict-driven clause learning (CDCL) or a blocking clause.
The motivation for adding blocking clauses is to ensure that the (partial)
models detected by our calculus represent pairwise disjoint sets of total
models.
In some tasks, however, enumerating models multiple times causes no harm,
and we can refrain from adding blocking clauses to the formula and
avoid its blowing-up in size.

Consider a formula \nolinebreak \(P(X,Y,S)\) over relevant variables \nolinebreak
\(X\), irrelevant variables \nolinebreak \(Y\) and Tseitin variables
\nolinebreak \(S\), and let \nolinebreak \(I\) with variables in
\(X \cup Y \cup S\) be a total assignment satisfying \nolinebreak \(P\).
Remember that by applying rule \BtopIrredName, the trail \nolinebreak \(I\) is
shrunken to \(\istar\) and projected onto \nolinebreak 
\(X\) obtaining \nolinebreak \(m = \project{\istar}{X}\).
Backtracking to
decision level \(\level(m)-1\) occurs, and the decision literal
\(\decided{\ell}\) at decision level \(\level(m)\) is flipped, \iet, propagated
with reason \(B = \negated{m}\). 
If now \(B\) is not added to \nolinebreak \(P\) and a conflict
involving \nolinebreak \(\comp{\ell}\) occurs, the reason of \nolinebreak \(\comp{\ell}\) is
not available for conflict analysis.
To ensure the functioning of conflict analysis, we propose to annotate
\nolinebreak \(\comp{\ell}\) on the trail with \nolinebreak \(\negated{m}\) but without
adding \nolinebreak \(B\) to \nolinebreak \(P\).

\begin{example}[Conflict analysis for model enumeration]
  \label{ex:conflana}

  Consider the following formula over the set of variables
  \(X = \{a, b, c, d\}\), \(Y = \{e\}\), \(S = \emptyset\): 
  \begin{equation*}
    P(X,Y,S) = \underbrace{(a \vee b \vee \negated{c})}_{C_1} \wedge
    \underbrace{(\negated{b} \vee c)}_{C_2} \wedge
    \underbrace{(d \vee \negated{c} \vee e)}_{C_3} \wedge
    \underbrace{(d \vee \negated{c} \vee \negated{e})}_{C_4}
  \end{equation*}
  Suppose we decide \nolinebreak \(a\) and \nolinebreak \(b\), propagate
  \nolinebreak \(c\) with reason \nolinebreak \(C_2\) and decide \nolinebreak
  \(d\) followed by deciding \nolinebreak \(e\).
  The resulting trail \(I_1 =
  \decided{a}\,\decided{b}\,\propagated{c}{C_2}\,\decided{d}\,\decided{e}\) is a model of
  \nolinebreak \(F\).
  This model is blocked by \(B_1 = (\negated{a} \vee \negated{b} \vee
  \negated{d})\) with decision level \(\level(B_1)=3\), which consists of the
  negated relevant decision literals on \nolinebreak \(I_1\).
  Considering only the decisions ensures that \nolinebreak \(B_1\) contains
  exactly one literal per decision level and that after backtracking
  to decision level \(\level(B_1)-1=2\), the clause \nolinebreak
  \(B_1\) becomes unit.
  Recall that \nolinebreak \(B_1\) is not added to \nolinebreak \(P\).
  The decision literal \nolinebreak \(\decided{d}\) is flipped with reason
  \(B_1\), and \nolinebreak \(e\) is propagated
  with reason \nolinebreak \(C_3\).
  But now \nolinebreak \(C_4\) is falsified.
  The current trail is \(I_2 =
  \decided{a}\,\decided{b}\,\propagated{c}{C_2}\,\propagated{\negated{d}}{B_1}\,\propagated{e}{C_3}\),
  which is visualized by the following implication graph:
  \begin{center}
    \begin{tikzpicture}
      \node (a) {\(a@1\)};
      \node[below = 16pt of a] (b) {\(b@2\)};
      \node[right=60pt of a] (nd) {\(\negated{d}\)};
      \node[below = 16pt of nd] (c) {\(c\)};
      \node (e) at ($(a.south east)+(120pt,-8pt)$) {\(e\)};
      \node[right = 50pt of e] (kappa) {\(\kappa\)};
      \draw[-{to[length=5.5mm]}] (a) to node [above, midway] (TextNode) {\(B_1\)} (nd);
      \draw[-{to[length=5.5mm]}] (b) to [bend right=5] node [above, midway] (TextNode) {\(B_1\)} (nd);
      \draw[-{to[length=5.5mm]}] (b) to node [below, midway] (TextNode) {\(C_2\)} (c);
      \draw[-{to[length=5.5mm]}] (nd) to [bend left=5] node [above, midway] (TextNode) {\(C_3\)} (e);
      \draw[-{to[length=5.5mm]}] (c) to [bend right=5] node [below, midway] (TextNode) {\(C_3\)} (e);
      \draw[-{to[length=5.5mm]}] (nd) to [bend left=30] node [above, midway] (TextNode) {\(C_4\)} (kappa);
      \draw[-{to[length=5.5mm]}] (e) to node [above, midway] (TextNode) {\(C_4\)} (kappa);
      \draw[-{to[length=5.5mm]}] (c) to [bend right=30] node [below, midway] (TextNode) {\(C_4\)} (kappa);
    \end{tikzpicture}
  \end{center}
  For conflict analysis, we first resolve the conflicting clause \nolinebreak \(C_4\) with the reason
  of \nolinebreak \(e\), \nolinebreak \(C_3\), obtaining the resolvent
  \((d \vee \negated{c})\).
  Both \nolinebreak \(d\) and \nolinebreak \(\negated{c}\) have the highest
  decision level \nolinebreak \(2\), and we continue by resolving \((d \vee \negated{c})\)
  with \(B_1\) obtaining \((\negated{c} \vee \negated{a} \vee b)\), followed by resolution with
  \nolinebreak \(C_2\) resulting in \(C_5 = (\negated{a} \vee \negated{b})\),
  which has only one literal at decision level \nolinebreak \(2\).
  The resolution process stops, and \nolinebreak \(C_5\) is
  added to \nolinebreak \(P\).
\end{example}

\section{Projected Redundant Model Enumeration}
\label{sec:pmered}

Now we turn our attention to the case where enumerating models multiple times
is permitted.
This allows for refraining from adding blocking clauses to the formula under
consideration, since they might significantly slow down the enumerator.
This affects both our algorithm and our calculus for irredundant projected model
enumeration.
Omitting the use of blocking clauses has a minor impact on our algorithm and its
formalization.
For this reason, in this section we point out the differences between the two methods.

\subsection{Algorithm and Calculus}

\begin{figure}[!t]
 \setlength{\fboxsep}{0pt}
 \noindent\fbox{
   \parbox{0.979\textwidth}{
     ~\\[\vspacemidrule]
     \renewcommand{\arraystretch}{1.2}
     \setlength\tabcolsep{2pt}
     \centering
     \begin{tabular}{@{}ll@{}}
       \BtopRed\\[\vspacemidrule]
     \end{tabular}
   }
 }
 \caption{Rule for backtracking after detection of a model in redundant model
   enumeration. 
   The calculus for redundant projected model enumeration differs from its
   irredundant counterpart only in the fact that no blocking clauses are used.
   Hence, all rules in \autoref{fig:calculusenumirred} are maintained except
   for rule \nolinebreak \BtopIrredName, which is replaced by rule \nolinebreak
   \BtopRedName.} 
 \label{fig:btopenumred}
\end{figure}

The only difference compared to \mainalgirredname consists in the fact that no
blocking clauses are added to \nolinebreak \(P\).
However, they are remembered as annotations on the trail in order to enable
conflict analysis after finding a model.
Our algorithm \mainalgredname therefore is exactly the same as \mainalgirredname
listed in \autoref{fig:enumirred} without lines \nolinebreak 22--23.
The annotation of flipped literals happens in function
\functionname{Backtrack()} in line 26.

Accordingly, our formalization consists of all rules in
\autoref{fig:calculusenumirred} except for rule \nolinebreak \BtopIrredName
which is replaced by rule \BtopRedName shown in \autoref{fig:btopenumred}.
Rule \BtopRedName differs from rule \BtopIrredName only in the fact that both
\nolinebreak \(P\) and \nolinebreak \(N\) remain unaltered.

\subsection{Example}
\label{sec:calculusenumredexample}

\begin{example}[Projected redundant model enumeration]
  \label{ex:multenumcdclproject}
  Consider again \autoref{ex:multenumcdcl} elaborated in detail in
  \autoref{sec:calculusenumirredexample} for \mainalgirredname.
  We have
  \begin{equation*}
    P = \underbrace{(a \vee c)}_{C_1} \wedge
    \underbrace{(a \vee \negated{c})}_{C_2} \wedge
    \underbrace{(b \vee d)}_{C_3} \wedge
    \underbrace{(b \vee \negated{d})}_{C_4}
  \end{equation*}
  and
  \begin{align*}
    N = \,
      &
        \underbrace{(\negated{t_1} \vee \negated{a})}_{D_1} \wedge
        \underbrace{(\negated{t_1} \vee \negated{c})}_{D_2} \wedge
        \underbrace{(a \vee c \vee t_1)}_{D_3} \wedge\\
      &
        \underbrace{(\negated{t_2} \vee \negated{a})}_{D_4} \wedge
        \underbrace{(\negated{t_2} \vee c)}_{D_5} \wedge
        \underbrace{(a \vee \negated{c} \vee t_2)}_{D_6} \wedge\\
      &
        \underbrace{(\negated{t_3} \vee \negated{b})}_{D_7} \wedge
        \underbrace{(\negated{t_3} \vee \negated{d})}_{D_8} \wedge
        \underbrace{(b \vee d \vee t_3)}_{D_9} \wedge\\
      &
        \underbrace{(\negated{t_4} \vee \negated{b})}_{D_{10}} \wedge
        \underbrace{(\negated{t_4} \vee d)}_{D_{11}} \wedge
        \underbrace{(b \vee \negated{d} \vee t_4)}_{D_{12}} \wedge\\
      &
        \underbrace{(t_1 \vee t_2 \vee t_3 \vee t_4)}_{D_{13}}
  \end{align*}
  Suppose \nolinebreak \(X = \{a, c\}\) and \nolinebreak \(Y = \{b, d\}\). 
  The execution trail is depicted in \autoref{fig:calculusenumredexample}.

  \begin{figure}[!t]
    \begin{center}
      \begin{tabular}{cllccc}
        \toprule
        Step & ~Rule & ~\(I\) & \(\residual{P}{I}\) & \(M\)\\
        \midrule
        0
          &
          &
            \(\emptytrail\)
          &
            \(P\) 
          &
            \(\false\)\\
        \rowcolor{black!10} 1
          &
            \DecXIrredName
          &
            \(\decided{a}\)
          &
            \((b \vee d) \wedge (b \vee \negated{d})\)
          &
            \(\false\)\\
        2
          &
            \DecXIrredName
          &
            \(\decided{a}\,\decided{b}\)
          &
            \(\true\)
          &
            \(\false\)\\
        \rowcolor{black!10} 3
          &
            \DecYSIrredName
          &
            \(\decided{a}\,\decided{b}\,\decided{c}\)
          &
            \(\true\)
          &
            \(\false\)\\
        4
          &
            \DecYSIrredName
          &
            \(\decided{a}\,\decided{b}\,\decided{c}\,\decided{d}\)
          &
            \(\true\)
          &
            \(\false\)\\
        \rowcolor{black!10} 5 
          &
            \BtopRedName
          &
            \(\decided{a}\,\propagated{\negated{b}}{B_1}\)
          &
            \((d) \wedge (\negated{d})\)
          &
            \(a \wedge b\)\\
        6
          &
            \UnitName
          &
            \(\decided{a}\,\propagated{\negated{b}}{B_1}\propagated{d}{C_3}\)
          &
            \(0\)
          &
            \(a \wedge b\)\\
        \rowcolor{black!10} 7
          &
            \BbotName
          &
            \(\propagated{b}{C_5}\)
          &
            \((a \vee c) \wedge (a \vee \negated{c})\)
          &
            \(a \wedge b\)\\
        8
          &
            \DecXIrredName
          &
            \(\propagated{b}{C_5}\,\decided{a}\)
          &
            \(\true\)
          &
            \(a \wedge b\)\\
        \rowcolor{black!10} 9
          &
            \DecYSIrredName
          &
            \(\propagated{b}{C_5}\,\decided{a}\,\decided{c}\)
          &
            \(\true\)
          &
            \(a \wedge b\)\\
        10
          &
            \DecYSIrredName
          &
            \(\propagated{b}{C_5}\,\decided{a}\,\decided{c}\,\decided{d}\)
          &
            \(\true\)
          &
            \(a \wedge b\)\\
        \rowcolor{black!10} 11
          &
            \BtopRedName
          &
            \(\propagated{b}{C_5}\,\propagated{\negated{a}}{B_2}\)
          &
            \((c) \wedge (\negated{c})\)
          &
            \((a \wedge b) \vee (b \wedge a)\)\\
        12
           &
             \UnitName
           &
             \(\propagated{b}{C_5}\,\propagated{\negated{a}}{B_2}\,\propagated{c}{C_1}\)
           &
             \(\false\)
           &
             \((a \wedge b) \vee (b \wedge a)\)\\
        \rowcolor{black!10} 13
             &
               \EbotName{}
             &
             &
             &
               \((a \wedge b) \vee (b \wedge a)\)\\
        \bottomrule
        \end{tabular}
      \end{center}
      \caption{Execution trace for \(F = (a \vee c) \wedge (a \vee \negated{c})
        \wedge (b \vee d) \wedge (b \vee \negated{d})\) defined over the set of
        relevant variables \(X = \{a, b\}\) and the set of irrelevant variables
        \(Y = \{c, d\}\) (see \autoref{ex:multenumcdcl}).}
      \label{fig:calculusenumredexample}
    \end{figure}

  Assume we decide \nolinebreak \(a\), \nolinebreak \(b\), \nolinebreak \(c\), and
  \nolinebreak \(d\) (steps \nolinebreak 1--4) obtaining the trail 
  \(I_1 = \decided{a}\,\decided{b}\,\decided{c}\,\decided{d}\) which is a model
  of \nolinebreak \(P\).
  Dual model shrinking occurs as in step 5 in the example elaborated in
  \autoref{sec:calculusenumirredexample}, except that the assumed literals
  \nolinebreak \(b\) and \nolinebreak \(c\) occur in a different order,
  and the same model \nolinebreak \(a\,b\) is obtained.
  Notice that the clause \(B_1 = (\negated{a} \vee \negated{b})\) is
  not added to \nolinebreak \(P\).

  After backtracking, \(\residual{P}{I_1} = (d) \wedge (\negated{d})\),
  and after propagating \nolinebreak \(d\) \nolinebreak (step \nolinebreak  6),
  we obtain a conflict.
  The current trail is
  \(I_3 = \decided{a}\,\propagated{\negated{b}}{B_1}\,\propagated{d}{C_3}\) and
  \(\residual{C_4}{I_3} = ()\).
  Resolution of the reasons on \nolinebreak \(I_3\) in reverse
  assignment order is executed, starting with the conflicting clause \nolinebreak \(C_4\).
  We obtain \(C_4 \otimes C_3 = (b) = C_5\), which contains exactly one literal
  at the maximum decision level, hence no further resolution steps are required. 
  Since \nolinebreak \((b)\) is unit, the enumerator backtracks to decision
  level \nolinebreak \(0\) and propagates \nolinebreak \(b\) with reason
  \nolinebreak \(C_5\) (step \nolinebreak 7). 
  After deciding \nolinebreak \(a\), \nolinebreak \(c\), and \nolinebreak \(d\),
  we find the same model \nolinebreak \(b\,a\,c\,d\) as in step \nolinebreak 4 (steps 8--10).
  Obviously, model shrinking provides us with the same model \nolinebreak
  \(b\,a\), which is added to \nolinebreak \(M\), and the last relevant decision
  is flipped (step \nolinebreak 11).
  Now unit propagation leads to a conflict (step \nolinebreak 12), and since
  there are no decisions on the trail, the procedure stops (step \nolinebreak
  13).
  Now the cubes in \nolinebreak \(M\), which represent the models of
  \nolinebreak \(P\), are not pairwise disjoint anymore.
  However, we still have \(M \equiv \project{P}{X} \equiv \project{F}{X}\).
\end{example}

\subsection{Proofs}
\label{sec:calculusenumredproofs}

\hyperref[fig:invenumirred]{Invariants \InvDualPNName} and \hyperref[fig:invenumirred]{\InvDecsName} listed in
\autoref{fig:invenumirred} are applicable also for redundant model enumeration,
since they involve none of \nolinebreak \(P\) and \nolinebreak \(N\).
\hyperref[fig:invenumirred]{Invariant \InvImplIIrredName} instead need be adapted since no
blocking clauses are added to \nolinebreak \(P\) and therefore it ceases to
hold.
Assume a model \nolinebreak \(I\) has been found and shrunken to \nolinebreak
\(m\) and that the last relevant decision
literal \nolinebreak \(\ell\) has been flipped.
Since its reason \nolinebreak \(B = \decs(\negated{m})\) is not added to \nolinebreak
\(P\), from \(P \wedge \decs(I)\) we can not infer \nolinebreak \(I\).
Recall that instead \nolinebreak \(m\) is added to \nolinebreak \(M\), hence
\nolinebreak \(\negated{M}\) contains the reasons of all decision literals which
were 
flipped after having found a model.
A closer look reveals that this case is analog to the one in our previous work
\nolinebreak \cite{DBLP:conf/gcai/MohleB19}.
In this work, we avoided the use of blocking clauses by means of chronological
backtracking.
However, the basic idea is the same, and we replace \hyperref[fig:invenumirred]{Invariant 
\InvImplIIrredName} by \hyperref[fig:invenumred]{Invariant \InvImplIRedName} listed in
\autoref{fig:invenumred}.
This is exactly Invariant \nolinebreak (3) in our previous work on model
counting \nolinebreak \cite{DBLP:conf/gcai/MohleB19}, hence in our proof we
use a similar argument.
The invariants for redundant model enumeration under projection are given in
\autoref{fig:invenumred}.

\begin{figure}[!t]
    \setlength{\fboxsep}{0pt}
    \noindent\fbox{
      \parbox{0.979\textwidth}{
        ~\\[\vspacemidrule]
        \renewcommand{\arraystretch}{1.2}
        \setlength\tabcolsep{2pt}
        \centering
        \begin{tabular}{@{}ll@{}}
          \InvDualPN\\[\vspacebetweeninvs]
          \InvDecs\\[\vspacebetweeninvs]
          \InvImplIRed\\[\vspacemidrule]
        \end{tabular}
      }
    }
    \caption{
      Invariants for projected model enumeration with repetition.
      Notice that Invariants \InvDualPNName and \InvDecsName are the same as for
      irredundant model enumeration while, due to the lack of blocking clauses,
      in Invariant \InvImplIIrredName the models recorded in \nolinebreak \(M\)
      need be considered.
    }
    \label{fig:invenumred}
  \end{figure}

\subsubsection{Invariants in Non-Terminal States}
\label{sec:calculusenumredproofsinv}

\begin{proposition}[Invariants in \mainalgredname]
  \label{prop:invenumred}
  The \hyperref[fig:invenumred]{Invariants \InvDualPNName},
  \hyperref[fig:invenumred]{\InvDecsName}, and
  \hyperref[fig:invenumred]{\InvImplIRedName} hold in 
  non-terminal states. 
\end{proposition}
\begin{proof}
  The proof is carried out by induction over the number of rule applications.
  Assuming \hyperref[fig:invenumred]{Invariants \InvDualPNName}, \hyperref[fig:invenumred]{\InvDecsName},
  and \hyperref[fig:invenumred]{\InvImplIRedName}
  hold in a non-terminal state \(\state{P}{N}{M}{I}{\level}\), we show that they
  are met after the transition to another non-terminal state for all rules.

  Now rules \nolinebreak \EtopName, \EbotName, \BbotName, \DecXIrredName, and
  \nolinebreak \DecYSIrredName are the same as for \mainalgirredname
  \nolinebreak (\autoref{fig:calculusenumirred}).
  In \autoref{sec:calculusenumirredproofsinv} we already proved that after the
  execution of these rules \hyperref[fig:invenumred]{Invariants 
  \InvDualPNName} and \hyperref[fig:invenumred]{\InvDecsName} still hold.

  As for \hyperref[fig:invenumred]{Invariant \InvImplIRedName}, from
  \autoref{prop:irredtotalmodel}, \nolinebreak \autoref{it:Icontained} and
  \nolinebreak \autoref{it:Isubsumed}, and observing that \(m \leq I\), where \nolinebreak
  \(I\) is a total model of \nolinebreak \(P\) and \nolinebreak \(m \in M\) its
  projection onto the relevant variables, we can conclude that \hyperref[fig:invenumred]{Invariant
  \InvImplIRedName} holds as well.
  To see this, remember that in \hyperref[fig:invenumred]{Invariant \InvImplIIrredName} we consider
  \nolinebreak \(P = P_0 \wedge \bigwedge_i B_i\) where the \nolinebreak \(B_i\)
  are the clauses added to \nolinebreak \(P_0\) blocking the models \nolinebreak
  \(m_i\).
  But \(B_i \leq m_i\), hence \hyperref[fig:invenumred]{Invariant \InvImplIRedName} holds after applying rules \UnitName, \BbotName, \DecXIrredName, and \DecYSIrredName, and we
  are left to carry out the proof for rule \BtopRedName. 

  \skipbetweenrules


  \noindent
  \underline{\BtopRedName} 
  \vspaceafterrulename

  \skipbetweeninvs
  
  \noindent
  \textit{\hyperref[fig:invenumred]{Invariant \InvDualPNName}:}  \skipafterinvname
  Both \nolinebreak \(P\) and \nolinebreak \(N\) remain unaltered, therefore
  \hyperref[fig:invenumred]{Invariant \InvDualPNName} holds after the application of
  \nolinebreak \BtopRedName.
  
  \skipbetweeninvs
    
  \noindent
  \textit{\hyperref[fig:invenumred]{Invariant \InvDecsName}:} \skipafterinvname
  The proof is analogous to the one for rule \nolinebreak \BtopIrredName.
  
  \skipbetweeninvs
  
  \noindent
  \textit{\hyperref[fig:invenumred]{Invariant \InvImplIRedName}:} \skipafterinvname
  We need to show that
  \(P \wedge \negated{(M \vee m)} \wedge \decsf{\leqslant n}(J\,\ell)
  \models \filter{(J\,\ell)}{\leqslant n}\) for all \nolinebreak \(n\).
  First, notice that the decision levels of all the literals in \nolinebreak
  \(J\) do not change while applying the rule.
  Only the decision level of \nolinebreak \(\ell\) is
  decremented from \nolinebreak \(b+1\) to \nolinebreak \(b\).
  It also stops being a decision.
  Since \(\level(J\,\ell) = b\), we can assume \(n \leqslant b\).
  Observe that
  \(P \wedge \negated{(M \vee m)} \wedge \decsf{\leqslant n}(J\,\ell) \equiv
  \negated{m} \wedge (P \wedge \negated{M} \wedge \decsf{\leqslant n}(I))\), 
  since \nolinebreak \(\ell\) is not a decision in \nolinebreak \(J\,\ell\) and
  \(\filter{I}{\leqslant b} = J\) and
  \(\filter{I}{\leqslant n} = \filter{J}{\leqslant n}\) by definition.
  Now the induction hypothesis is applied and we get
  \(P \wedge \negated{(M \vee m)} \wedge \decsf{\leqslant n}(J\,\ell) \models
  \filter{I}{\leqslant n}\). 
  Again, using
  \(\filter{I}{\leqslant n} = \filter{J}{\leqslant n}\),
  this almost closes the proof except that we are left to prove
  \(P \wedge \negated{(M \vee m)} \wedge \decsf{\leqslant e}(J\,\ell) \models
  \ell\)
  as \nolinebreak \(\ell\) has decision level \nolinebreak \(b\) in \nolinebreak
  \(J\,\ell\) after applying the rule and thus \nolinebreak \(\ell\) disappears in
  the proof obligation for \(n < b\). 
  To see this notice that \(P \wedge \negated{B} \models \filter{I}{\leq b+1}\)
  using again the induction hypothesis for \(n = b + 1\) and recalling that
  \(\negated{B} = \decsf{\leq b+1}(I)\).
  This gives
  \(P \wedge \negated{\decsf{\leqslant b}(J)} \wedge \negated{\ell} \models
  \filter{I}{\leq b+1}\) and thus
  \(P \wedge \negated{\decsf{\leq b}(J)} \wedge \negated{\filter{I}{\leq b+1}}
  \models \ell\) by conditional contraposition.
\end{proof}

\subsubsection{Progress and Termination}
\label{sec:progtermredenum}

The proofs that our method for redundant projected model enumeration always
makes progress and eventually terminates are the same as in
\autoref{sec:calculusenumirredproofsprog} and
\autoref{sec:calculusenumirredproofsterm}.

\subsubsection{Equivalence}
\label{sec:eqredenum}

Some properties proved for the case of irredundant model enumeration cease to
hold if we allow enumerating redundant models.
Specifically, \autoref{prop:istarmodel}, and \autoref{prop:irredallfound} hold,
while \autoref{prop:irredmodels} does not.
\autoref{it:Icontained} and \autoref{it:Isubsumed} of
\autoref{prop:irredtotalmodel} hold, while
\autoref{it:Imodel} does not.
In \autoref{th:correctirred},
\autoref{it:meqf} holds but
\autoref{it:mdsop} does not.
Their proofs remain the same as for irredundant model
enumeration in \autoref{sec:calculusenumirredproofsinv}.

\subsection{Generalization}
\label{sec:genredenum}

The same observations made for irredundant model enumeration in
\autoref{sec:calculusenumirredpm} apply.

\begin{changed}

\section{Discussion}
\label{sec:discussion}


The complexity bounds for All-SAT are exponential in the number of variables
occurring in the formula for all algorithms, due to the size of the search
space. 
The goal is therefore to reduce the number of assignments to be checked, which
is achieved by CDCL and adding short blocking clauses. 
In the brute-force approach for irredundant model enumeration, a blocking clause
is exactly the negation of the satisfying assignment or consists of its negated
decision literals \nolinebreak
\cite{DBLP:conf/ictai/MorgadoS05,DBLP:journals/jea/TodaS16}. 
In both cases, the resulting blocking clause has size at least the decision
level of the trail, and their number corresponds to the number of models.
The addition of blocking clauses to \nolinebreak \(P\) (see line \nolinebreak 22
of algorithm \mainalgirredname listed in \autoref{fig:enumirred})
slows down unit propagation, in particular if they are long.
Blocking many assignments with few clauses is therefore crucial, since
CDCL-based SAT solvers spend most of their computing time with unit propagation.

For computing short blocking clauses, our dual model shrinking approach relies
on the ability of conflict analysis to determine short clauses.
Given a formula \nolinebreak \(P\) and its negation \nolinebreak \(N\), it
identifies the reason for the conflict in \nolinebreak \(N\) induced by an
assignment satisfying \nolinebreak \(P\) by adopting \nolinebreak CDCL in
\nolinebreak \(N\). 
Due to the effectiveness of conflict analysis, our shrinking method has the
potential to rule out a large number of assignments satisfying \nolinebreak
\(P\) as is shown by an example.  

\begin{example}[Efficiency of dual model shrinking]
  \label{ex:efficiency_dual}
  Consider the set of variables \(X=\{a,b,c,d,e,f\}\) and
  the formula
  \(F(X,Y)=
  (a \wedge e \wedge f) \vee
  (b \wedge e \wedge \negated{f}) \vee
  (c \wedge \negated{e} \wedge f) \vee
  (d \wedge \negated{e} \wedge \negated{f})\).
  Without loss of generality, we assume \(Y=\emptyset\).
  A \nolinebreak CNF representation of \nolinebreak \(F\) is given by
  \begin{align*}
    P(X,Y,S) =\;
    &
      (\negated{v_1} \vee a) \wedge
      (\negated{v_1} \vee e) \wedge
      (\negated{v_1} \vee f) \wedge
      (v_1 \vee \negated{a} \vee \negated{e} \vee \negated{f}) \wedge~\\
    &
      (\negated{v_2} \vee b) \wedge
      (\negated{v_2} \vee e) \wedge
      (\negated{v_2} \vee \negated{f}) \wedge
      (v_2 \vee \negated{b} \vee \negated{e} \vee f) \wedge~\\
    &
      (\negated{v_3} \vee c) \wedge
      (\negated{v_3} \vee \negated{e}) \wedge
      (\negated{v_3} \vee f) \wedge
      (v_3 \vee \negated{c} \vee e \vee \negated{f}) \wedge~\\
    &
      (\negated{v_4} \vee d) \wedge
      (\negated{v_4} \vee \negated{e}) \wedge
      (\negated{v_4} \vee \negated{f}) \wedge
      (v_4 \vee \negated{d} \vee e \vee f) \wedge~\\
    &
      (v_1 \vee v_2 \vee v_3 \vee v_4),
  \end{align*}
  where \(S=\{v_1, v_2, v_3, v_4\}\).
  Let further a \nolinebreak CNF representation \nolinebreak \(\negated{F}\) be
  given by
  \begin{align*}
    N(X,Y,T)=\;
    &
      (t_0) \wedge~\\
    &
      (\negated{u_1} \vee \negated{a} \vee \negated{e} \vee \negated{f}) \wedge
      (u_1 \vee a) \wedge
      (u_1 \vee e) \wedge
      (u_1 \vee f) \wedge~\\
    &     
      (\negated{u_2} \vee \negated{b} \vee \negated{e} \vee f) \wedge
      (u_2 \vee b) \wedge
      (u_2 \vee e) \wedge
      (u_2 \vee \negated{f}) \wedge~\\
    &
      (\negated{u_3} \vee \negated{c} \vee e \vee \negated{f}) \wedge
      (u_3 \vee c) \wedge
      (u_3 \vee \negated{e}) \wedge
      (u_3 \vee f) \wedge~\\
    &
      (\negated{u_4} \vee \negated{d} \vee e \vee f) \wedge
      (u_4 \vee d) \wedge
      (u_4 \vee \negated{e}) \wedge
      (u_4 \vee \negated{f}) \wedge~\\
    &
      (\negated{u_5} \vee u_1) \wedge
      (\negated{u_5} \vee u_2) \wedge
      (\negated{u_5} \vee u_3) \wedge
      (\negated{u_5} \vee u_4) \wedge~\\
    &
      (u_5 \vee \negated{u_1} \vee \negated{u_2} \vee \negated{u_3} \vee
      \negated{u_4}) \wedge~\\
    &
      (u_6 \vee t_0 \vee u_5) \wedge
      (u_6 \vee \negated{t_0} \vee \negated{u_5}) \wedge~\\
    &
      (\negated{u_6} \vee \negated{t_0} \vee u_5) \wedge 
      (\negated{u_6} \vee t_0 \vee \negated{u_5}) \wedge~\\
    &
      (u_6)
  \end{align*}
  Suppose the model 
  \(I=
  \decided{a}\,\decided{b}\,\decided{c}\,\decided{d}\,\decided{e}\,
  \negated{v_3}\,\negated{v_4}\,\decided{f}\,\negated{v_1}\,\negated{v_2}\) of
  \nolinebreak \(P\) has been found, where for better readability we omit the
  reasons of the propagation literals. 
  Obviously, \(\istar=a\,e\,f\) already satisfies \nolinebreak \(F\), and its
  negation, the clause
  \(\negated{\istar}=(\negated{a} \vee \negated{e} \vee \negated{f})\), is
  obtained after calling a \nolinebreak SAT solver on \nolinebreak \(N\) and
  \(\project{I}{X}\) and analyzing the resulting conflict.
  Our algorithm for irredundant model enumeration enumerates exactly the cubes
  in \nolinebreak \(F\), which is optimal.
\end{example}

This conflict is obtained exclusively by unit propagation, which
is linear in the length of \nolinebreak \(I\) and the computation of
\nolinebreak \(\istar\) requires a linear number of resolution steps.
Finally, the dual blocking clause encoding is linear in the size of the
original formula.
The enumerated models in \autoref{ex:efficiency_dual} coincide with the ones
obtained by our dual approach for projected model counting \nolinebreak
\cite{DBLP:conf/ictai/MohleB18}, which in our 
experiments finds minimal partial models, but without the overhead of processing
two formulae throughout the computation. 
As already mentioned, for algorithm \mainalgirredname to be correct
when adopting dual model shrinking in line \nolinebreak 20, \(N\) need represent
the negation of \nolinebreak \(P\) anytime.
Concretely, line \nolinebreak 23 is mandatory in combination with dual model
shrinking in line \nolinebreak 20 for determining the backtracking level in line
\nolinebreak 25.
Experiments further show that if a different, non-trivial blocking
clause computation strategy is adopted in line \nolinebreak 20, restarts are
required, which might lead to repeating the same (partial) assignments multiple
times \nolinebreak  \cite{DBLP:conf/sat/NadelR18,DBLP:journals/jsat/TakRH11}.

The strength of executing entailment checks, as proposed as a generalization, is
shown by \autoref{ex:efficiency_dual}.
In fact, the trail
\(I=\decided{a}\,\decided{b}\,\decided{c}\,\decided{d}\) already logically
entails \nolinebreak \(P\).
The combination of logical entailment checks and dual model shrinking bears the
potential to enumerate even shorter models.
Recent experiments support our claim for the efficiency of executing logical
entailment checks and of dual reasoning for model shrinking \nolinebreak
\cite{DBLP:conf/sat/FriedNSS24,DBLP:conf/aaai/SpallittaSB24}.




\end{changed}

\section{Conclusion}
\label{sec:conclusion}

Model enumeration and projection, with and without repetition, is a key element
to several tasks.
We have presented two methods for propositional model enumeration under
projection. 
\mainalgirredname uses blocking clauses to avoid enumerating models multiple
times, while \mainalgredname is exempt from blocking clauses and admits repetitions.
Our CDCL-based model enumerators detect total models and uses dual reasoning to
shrink them.

To ensure correctness of the shrinking mechanism,
we developed a dual encoding of the blocking clauses.
We provided a formalization and proof of correctness of our blocking-based
model enumeration approach and discussed a generalization to the case where partial models are
found.
These partial models might not be minimal, hence shrinking them still might make
sense.
Also, there is no guarantee that the shrunken models are minimal as they depend
on the order of the variable assignments.

We presented a conflict-driven clause learning mechanism for redundant model enumeration,
since standard \nolinebreak CDCL might fail in the absence of blocking clauses.
Basically, those clauses are remembered on the trail without being added to the
input formula.
This prevents a blowup of the formula but also does not further make use of
these potentially short clauses, which in general propagate more eagerly than
long clauses. 
We discussed the modifications of our blocking-based algorithm and calculus to support 
redundant model enumeration and provided a correctness proof.
Intuitively, shorter partial models representing non-disjoint
sets of total models might be found.

Our method does not guarantee that the shrunken model \nolinebreak \(\istar\) is
minimal \wrtt the decision level \nolinebreak \(b\) in line \nolinebreak
\mainalgirredname.
However, finding short \nolinebreak DSOPs is important in circuit design \nolinebreak
\cite{DBLP:conf/iccad/MinatoM98}, and appropriate algorithms have been
introduced by, \egt, Minato \nolinebreak \cite{Minato1993Generation}.
While \nolinebreak DSOP minimization has been proven to be NP-complete
\nolinebreak \cite{DBLP:journals/mst/BernasconiCLP13}, finding a smaller
decision level \nolinebreak \(b\) would already be advantageous, since besides
restricting the search space to be explored it generates shorter models.
To this end, we plan to adapt our dual shrinking algorithm to exploit the Tseitin
encoding as proposed by Iser et al. \nolinebreak \cite{DBLP:conf/sat/IserST13}.

In the presence of multiple conflicting clauses, a related interesting question
might also be which one to choose as a starting point for conflict analysis
with the aim to backtrack as far as possible.
This is not obvious unless all conflicts are analyzed.

Determining short models makes our approach suitable for circuit design.
We are convinced that this work provides incentives not only for the
hardware-near community but also for the enumeration community.

\section*{Acknowledgments} 

The work was supported by the LIT Secure and Correct Systems Lab funded by the
State of Upper Austria, by the QuaSI project funded by D-Wave Systems Inc., and
by the Italian Association for Artificial Intelligence (AI*IA).
We acknowledge the support of the MUR PNRR project
FAIR – Future AI Research (PE00000013), under the NRRP
MUR program funded by the NextGenerationEU. The work
was partially supported by the project ``AI@TN'' funded by
the Autonomous Province of Trento. This research was partially supported by TAILOR, a project funded by the EU
Horizon 2020 research and innovation program under GA
No 952215.
We also thank Mathias Fleury for a number of helpful discussions which in
particular led to a more elegant equivalence proof compared to our original
version.
We are also indebted to the anonymous reviewers for their valuable feedback.

\bibliographystyle{plain}
\bibliography{ms}

\end{document}